%% file: arxiv.tex
\documentclass[%
 reprint,
 twocolumn,
superscriptaddress,
 amsmath,amssymb,
 aps,
pra,
floatfix,
]{revtex4-2}

\usepackage{graphicx}
\usepackage{dcolumn}
\usepackage{bm}
\usepackage{cancel}
\usepackage{bbm}
\usepackage{float}
\usepackage{amssymb}
\usepackage{amsmath}
\usepackage{amsthm}
\usepackage{color}
\usepackage{xcolor}
\usepackage{physics}
\usepackage{siunitx}
\usepackage{hyperref}
\usepackage{booktabs}
\usepackage{makecell}

\newtheorem{theorem}{Theorem}
\newtheorem{corollary}{Corollary}[theorem]
\newtheorem{lemma}{Lemma}

\begin{document}
\input{arxiv_main}
\clearpage
\twocolumngrid

\setcounter{equation}{0} \renewcommand{\theequation}{S\arabic{equation}}
\setcounter{figure}{0}   \renewcommand{\thefigure}{S\arabic{figure}}
\setcounter{table}{0}    \renewcommand{\thetable}{S\arabic{table}}
\setcounter{section}{0}  \renewcommand{\thesection}{S\arabic{section}}

\input{arxiv_si}

\input{arxiv.bbl}

\end{document}

%% file: arxiv_main.tex
\title{Identifying nonequilibrium degrees of freedom in high-dimensional stochastic systems}

\author{Catherine Ji}%
\thanks{Equal contribution}
\affiliation{Department of Physics, Princeton University
}
\affiliation{Department of Computer Science, Princeton University}

\author{Ravin Raj}%
\thanks{Equal contribution}
\affiliation{Department of Physics, Princeton University
}
\author{Benjamin Eysenbach}
\email{eysenbach@princeton.edu}
 \affiliation{Department of Computer Science, Princeton University}
\author{Gautam Reddy}
 \email{greddy@princeton.edu}
\affiliation{Department of Physics, Princeton University
}%

\date{\today}

\begin{abstract}
Any coarse-grained description of a nonequilibrium system should faithfully represent its latent irreversible degrees of freedom. However, standard dimensionality reduction methods typically prioritize accurate reconstruction over physical relevance. Here, we introduce a model-free approach to identify irreversible degrees of freedom in stochastic systems that are in a nonequilibrium steady state. Our method leverages the insight that a black-box classifier, trained to differentiate between forward and time-reversed trajectories, implicitly estimates the local entropy production rate. By parameterizing this classifier as a quadratic form of learned state representations, we obtain nonlinear embeddings of high-dimensional state-space dynamics, which we term \emph{Latent Embeddings of Nonequilibrium Systems (LENS)}. LENS effectively identifies low-dimensional irreversible flows and provides a scalable, learning-based strategy for estimating entropy production rates directly from high-dimensional time series data. 

\end{abstract}
\maketitle

\textit{Introduction}---
Internal reactions and external forces drive active matter systems out of thermodynamic equilibrium, creating irreversible state space flows which describe the extent to which detailed balance is broken at the microscopic level~\cite{Battle2016, Gladrow2017, Lynn2022a}. In biological systems, these flows are typically driven by internal enzymatic processes and the action of molecular machines~\cite{Seifert2012,Qian2016,Gnesotto2018,Mura2018,Tan2021}. Working directly with detailed microscopic descriptions, however, quickly becomes intractable as the number of degrees of freedom in a system is increased, so any reasonable macroscopic description of a nonequilibrium system must appropriately discard non-essential information. Thus, a central problem of faithfully describing nonequilibrium systems lies in identifying relevant entropy-producing degrees of freedom.

Several prior works have developed methods for estimating entropy production. These methods include plug-in estimators and  compression-based algorithms, which directly estimate the densities of forward and time-reversed dynamics~\cite{Wang2005,Ziv1993,Roldan2012}. Thermodynamic uncertainty relations (TURs) enable the estimation of entropy production using fluctuations in finite-time estimates of the probability current~\cite{Barato2015,Gingrich2016,Horowitz2020,Li2019,Roldan2021}. However, the amount of data required for density or current estimation scales poorly with the dimensionality of the system.  Recent works use statistical learning methods such as dimensionality reduction~\cite{Gnesotto2020}, neural networks~\cite{Kim2020,Bae2022,Lyu2024,Vodret2024}, or score-matching~\cite{Boffi2024,Boffi2024a} to directly infer entropy production. While these estimators have been shown to scale favorably with dimensionality, in many applications it is useful to identify which features of the dynamics contribute to irreversibility.

In this Letter, we introduce \textit{Latent Embeddings of Nonequilibrium Systems} (LENS) to identify relevant irreversible degrees of freedom of a high-dimensional stochastic system in a nonequilibrium steady state (NESS). In contrast to commonly used dimensionality reduction methods (such as principal component analysis (PCA)) which prioritize reconstruction, LENS is a model-free approach that leverages contrastive representation learning~\cite{Ma2018} to directly extract entropy-producing degrees of freedom from high-dimensional time series data. 

\textit{LENS Construction}---For a Markovian system in a NESS, the entropy production rate (EPR) $\dot{S}$ 
is defined as the Kullback-Leibler (KL) divergence between the forward $p(\va{x}\rightarrow \va{x}')$ and the reverse dynamics $p(\va{x}'\rightarrow \va{x})$~\cite{Kawai2007}. Here, $\va{x}$ and $\va{x}'$ denote the states of the system at consecutive time steps $t$ and $t+1$ respectively, with $p(\va{x}\rightarrow \va{x}')$ as the probability of observing a transition from $\va{x}$ to $\va{x}'$ and vice versa for  $p(\va{x}'\rightarrow \va{x})$. This KL divergence is given by
\begin{align}
    \dot{S}&=D_{\mathrm{KL}}\big[p(\va{x}\rightarrow \va{x}')\| p(\va{x}'\rightarrow \va{x})\big] \nonumber \\
    &=\mathbb{E}_{p(\va{x}\rightarrow \va{x}')}\log\Big[\frac{p(\va{x}\rightarrow \va{x}')}{p(\va{x}' \rightarrow \va{x})}\Big]. \label{eq:epr_expval_definition}
\end{align}
To construct the LENS objective, we reformulate EPR estimation as learning a binary classifier $C_\theta(\va{x}, \va{x}')$ (parameterized by $\theta)$ that can distinguish forward paths from time-reversed paths. The classifier $C_\theta(\va{x},\va{x}')=\sigma\qty(S_\theta(\va{x},\va{x}'))$, where $\sigma(z)=1/(1+e^{-z})$ is the logistic function, represents the probability that a randomly sampled pair of adjacent states $(\va{x},\va{x}')$ originates from the forward dynamics $(\va{x}\rightarrow \va{x}')$ as opposed to the reverse dynamics $(\va{x}'\rightarrow \va{x})$, and $S_\theta(\va{x},\va{x}')$ is the corresponding logit. The classifier is trained by maximizing the binary cross-entropy objective
\begin{align}
    J(\theta)=&\,\mathbb{E}_{p(\va{x}\rightarrow \va{x}')}\qty[\log \sigma\qty(S_\theta(\va{x},\va{x}'))] \nonumber \\
    &\quad +\mathbb{E}_{p(\va{x}'\rightarrow \va{x})}\qty[\log\qty(1-\sigma\qty(S_\theta(\va{x},\va{x}')))] ,\label{eq:lens_objective_function}
\end{align}
with respect to parameters $\theta$. The expectations are computed over pairs of adjacent states sampled from a provided time series dataset. With infinite data, the ideal classifier that maximizes this objective  (obtained from setting $\delta J/\delta S = 0$) is the optimal logit $S_{\text{opt}}(\va{x},\va{x}') = \log \qty[p(\va{x}\rightarrow \va{x}')/p(\va{x}'\rightarrow \va{x})]$. Thus, a sufficiently expressive classifier can be used after convergence to estimate the EPR by averaging $S_\theta(\va{x},\va{x}')$ over randomly sampled pairs $(\va{x},\va{x}')$ from the dataset (Eq.~\ref{eq:epr_expval_definition}). Maximizing the contrastive objective in Eq.~\eqref{eq:lens_objective_function} is equivalent (for continuous processes) to maximizing a lower bound on an $f$-divergence (see App.~\ref{app:lens_neep_equivalence}), which relates LENS to an alternative learning-based approach for estimating the EPR~\cite{Kim2020}. 

Towards our goal of identifying irreversible degrees of freedom, we parameterize the logit $S_\theta(\va{x},\va{x}')$ (Fig.~\ref{fig:lens_architecture}) with an $M$-dimensional learned representation $\va{\phi}_\theta\qty(\va{x})$, which is the output of a feedforward neural network. In particular, our LENS logit is parameterized as 
\begin{align}
    S_\theta(\va{x},\va{x}') &=\va{\phi}_\theta(\va{x}')^\intercal\vb{A}_\theta\va{\phi}_\theta(\va{x}) +\frac{1}{2}\va{\phi}_\theta(\va{x}')^\intercal \vb{B}_\theta\va{\phi}_\theta(\va{x}') \nonumber \\
    &\qquad -\frac{1}{2}\va{\phi}_\theta(\va{x})^\intercal \vb{B}_\theta\va{\phi}_\theta(\va{x}). \label{eq:lens_logit}
\end{align}
The learned matrix $\vb{A}_\theta$ has a block-diagonal form where each block is a $2\times2$ skew-symmetric matrix.
The block-diagonal form of $\vb{A}_{\theta}$ encourages LENS, without loss in generality, to learn the same representations across realizations of the stochastic system up to permutations of $2\times 2$ subspaces and rotations within each subspace (except in situations where the eigenvalues of $\vb{A}_{\theta}$ are degenerate).  $\vb{B}_\theta$ is symmetric, as any skew-symmetric component does not contribute to the logit (see App.~\ref{subsec:representation_universality}). 
Using deep nonlinear networks to parameterize $\va{\phi}_\theta$ allows for identifying features of arbitrary nonlinear stochastic systems that best approximate the EPR via $S_\theta(\va{x},\va{x}')$.

This particular choice of $S_\theta(\va{x},\va{x}')$ is motivated by two considerations. First, the form of $S_{\mathrm{opt}}(\va{x},\va{x}')$ requires that it is antisymmetric in its arguments, which to the lowest nontrivial order can be expressed as an antisymmetric quadratic form (with skew-symmetric $\vb{A}_{\theta}$) and the difference of two symmetric quadratic forms (with symmetric $\vb{B}_{\theta}$). This antisymmetric quadratic form with nonlinear $\va{\phi}_{\theta}$ is universal: for a sufficiently large $M$ there exists a representation that approximates $S_{\text{opt}}(\va{x}, \va{x}')$ arbitrarily well (see App.~\ref{subsec:representation_universality}). Second, for linear stochastic systems, $S_{\text{opt}}(\va{x}, \va{x}')$ takes the form in Eq. \eqref{eq:lens_logit} with linear $\va{\phi}_{\theta}$ and the representations learned by LENS have a precise physical interpretation (discussed further below). We now demonstrate how LENS can be applied to various linear and nonlinear stochastic systems. 

\begin{figure}[t]
    \includegraphics[width=\linewidth]{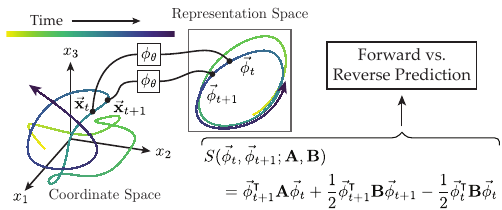} 
    \caption{\textbf{LENS Pipeline for Learning Representations}. Observations $\va{x}$ of high-dimensional stochastic dynamical systems are encoded using a learned function $\va{\phi}_\theta$ that provides a low-dimensional description of the dynamics. The encoded observations are used to compute a logit $S(\va{\phi}_t,\va{\phi}_{t+1};\vb{A},\vb{B})$, where $\va{\phi}_t\equiv\va{\phi}(\va{x}_t),\va{\phi}_{t+1}\equiv\va{\phi}(\va{x}_{t+1})$. This produces a prediction probability for whether the pair of observations $(\va{x}_t,\va{x}_{t+1})$ originates from the forward dynamics or the reverse dynamics.}
    \label{fig:lens_architecture}
\end{figure}

\textit{Linear Stochastic Dynamics}---
To begin, we consider the paradigmatic example of $N$ beads coupled by identical springs (of constant $k$) in a viscous medium~\cite{Rieder1967,Bonetto2004,Mura2018}. Each bead is in contact with a local thermal reservoir as shown in Fig.~\ref{fig:linear_dynamics}(a). Bead displacements $\va{x}_{t}$ exhibit linear stochastic dynamics according to the overdamped Langevin equation
\begin{equation}
    \label{eq:n_beads_langevin}
    \va{x}_{t+\dd t}=\va{x}_{t}+\vb{G}\va{x}_{t}\,\dd t+\vb{F}\va{\xi}_{t}\sqrt{\dd t}.
\end{equation}
The coupling matrix $\mathbf{G}$ has elements $G_{ij}=(-2\delta_{i,j} + \delta_{i,j+1} + \delta_{i,j-1}) k/\gamma$ and each bead experiences uncorrelated thermal noise with scale $F_{ij}=\sqrt{2k_BT_i/\gamma}\delta_{i,j}$, where $T_i$ is the temperature of the $i$-th bead, $\gamma$ is the viscous damping coefficient of the medium and $\va{\xi}_{t}$ is Gaussian white noise with zero mean and identity covariance.

\begin{figure*}[t]
    \includegraphics[width=\linewidth]{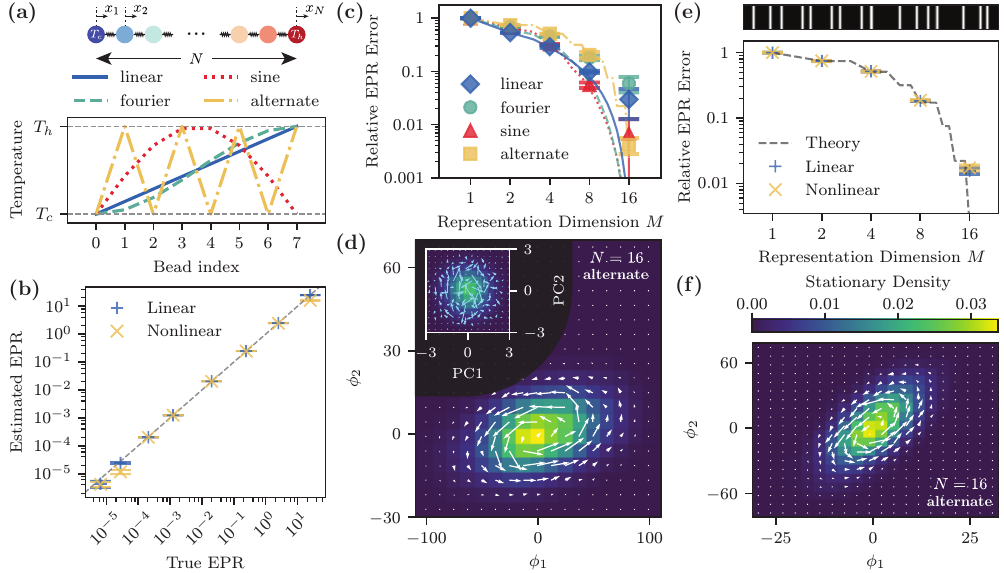} 
    \caption{\label{fig:linear_dynamics}\textbf{Learning Nonequilibrium Degrees of Freedom in Linear Stochastic Dynamical Systems}. (a) Coupled bead-spring model. (Top) $N$ beads are coupled using identical springs, with each bead coupled to a local thermal reservoir at temperature $T_i$. (Bottom) Plots of the various temperature profiles under consideration. (b) Estimated entropy production rate (EPR) for $N=2$ bead-spring model swept over $T_c/T_h$, linear and nonlinear representations of dimension $M=16$. (c) Relative (fractional) EPR error as a function of linear representation dimension $M$ for $N=16$ bead-spring model for various temperature profiles. The corresponding theoretical predictions are obtained by optimizing the reduced LENS objective~\eqref{eq:reduced_lens_objective}. (d) State space flow (white arrows) and stationary density (heatmap) shown in representation space for $M=2$ for fully-observed trajectories of the $N=16$ bead-spring model with alternate temperature profile. (Inset) State space flow (blue arrows) and stationary density (heatmap) of top 2 principal components (PCs) of the corresponding trajectory in coordinate space. (e) (Top) Representative 1D render of Gaussian-convolved bead displacements presented as input to networks. (Bottom) Relative EPR error as a function of representation dimension for Gaussian-convolved trajectories of the $N=16$ bead-spring model with alternate temperature profile, using linear and nonlinear representations. (f) State space flow (white arrows) and stationary density (heatmap) shown in representation space for $M=2$ nonlinear representations of partially-observed trajectories of the $N=16$ bead-spring model with alternate temperature profile. All stationary density heatmaps use the same colorbar. All error bars represent the standard error of the mean (s.e.m) over 5 random seed initializations.}
\end{figure*}

We first apply LENS to such a system of $N=2$ beads, and maximize the LENS objective function over a time series of $T=\num{e7}$ observations with time step $\dd t=\num{e-2}$. In this setting, LENS accurately estimates the EPR to within \SI{10}{\percent} error across 6 orders of magnitude (Fig.~\ref{fig:linear_dynamics}(b)), regardless of whether the learned representations are parameterized as linear or nonlinear (see Sec.~\ref{subsubsec:architectures} for representation architecture). 

To unpack the information contained in these linear representations, we apply LENS to high-dimensional systems of beads ($N=16$) with four different temperature profiles for different $M$ (Fig.~\ref{fig:linear_dynamics}(a),(c)). 
For $M=2$, LENS identifies a circulating flow (Fig.~\ref{fig:linear_dynamics}(d)) that explains 25\% of the total EPR. As a baseline, we compare to two-dimensional representations produced by PCA~\cite{Gnesotto2020}.
Fig.~\ref{fig:linear_dynamics}(d) shows that the identified degrees of freedom from PCA are less predictive of irreversible flows. 

We also apply LENS to rendered systems of coupled beads, where the data are generated by convolving bead displacements $\va{x}$ with a Gaussian point spread function to obtain ``blurred'' observations of $N$ linearly arranged beads (top of Fig.~\ref{fig:linear_dynamics}(e)). 
The relative EPR error is shown in Fig.~\ref{fig:linear_dynamics}(e) for various representation dimensions $M$, where we find that linear representations are able to capture the maximal possible EPR available for a rank-$M$ approximation to the linear dynamics. The benefit from learning a nonlinear approximation to the underlying linear stochastic dynamics is negligible, even when the observations are blurred with a nonlinearity. In Fig.~\ref{fig:linear_dynamics}(f), we similarly visualize the flow of nonlinear representations in the learned representational space for $M=2$, which reveals cyclical irreversible flows similar to those in Fig.~\ref{fig:linear_dynamics}(d). 

The representations learned by LENS for linear stochastic systems have a physical interpretation. To show this, we consider an $N$-dimensional linear system and $M$-dimensional linear representations $\va{\phi}_{\theta}(\va{x}) = \vb{P}\va{x}$. Maximizing the LENS objective Eq.~\eqref{eq:lens_objective_function} is equivalent to minimizing a reduced objective (see App.~\ref{app:lens_low_rank_approximation})
\begin{equation}
    \label{eq:reduced_lens_objective}
    \mathcal{V}_M(\tilde{\vb{\Lambda}})= 2\Tr(\tilde{\vb{\Lambda}}\vb{G}\vb{C}) +\Tr(\tilde{\vb{\Lambda}}^\intercal\vb{D}\tilde{\vb{\Lambda}}\vb{C})
\end{equation}
with respect to a rank $M$ matrix $\tilde{\vb{\Lambda}} = \vb{P}^T(\vb{A}_{\theta} +\vb{B}_{\theta})\vb{P}$, where $\vb{D}=\vb{F}\vb{F}^\intercal/2$ is the diffusion tensor and $\vb{C}$ is the autocorrelation matrix. Computing the stationary points of this objective, we find that LENS implicitly learns an $M$-rank approximation to the thermodynamic force matrix $\vb{\Lambda}=\vb{D}^{-1}\vb{G}+\vb{C}^{-1}$ and thereby the thermodynamic force field, $\vb{\Lambda}\va{x}$. The EPR error decreases monotonically with $M$ and converges to zero when $M=N$ (see App.~\ref{app:lens_low_rank_approximation}).

\textit{Nonlinear Stochastic Dynamics}---Next, we consider a simplified model of molecular motors traversing cytoskeletal filaments by consuming ATP~\cite{Yildiz2003,Yildiz2004}. Data obtained from imaging typically yields traces of motor locations along a one-dimensional track. In our model, a particle in a viscous medium is driven by a constant force $h$ under the influence of a periodic potential $U$, as shown in Fig.~\ref{fig:nonlinear_dynamics}(a). The particle's position $x_t$ at time $t$ is governed by
\begin{equation}
    \label{eq:nonlinear_periodic_langevin}
    x_{t+1}=x_t+\qty(h-\dv{U}{x})\dd t+\sqrt{2D \,\dd t}\,\xi_t,
\end{equation}
where the potential $U(x)=U_0\cos\qty(2\pi x/L)$ has amplitude $U_0$ and periodicity $L$. The system is nonlinear and yet admits an analytically tractable solution for the entropy production rate~\cite{Risken1996} (see App.~\ref{app:periodic_potential_analytic} for analytic result). We simulate trajectories of $T=\num{e7}$ observations with time step $\dd t=\num{e-3}$, and estimate the EPR of these trajectories using LENS with both linear and nonlinear representations. Since the periodic nature of the simulated trajectories is known \textit{a priori}, we include a sinusoidal activation function as an additional layer in our learned LENS encoder.

\begin{figure}[t]
    \includegraphics[width=\linewidth]{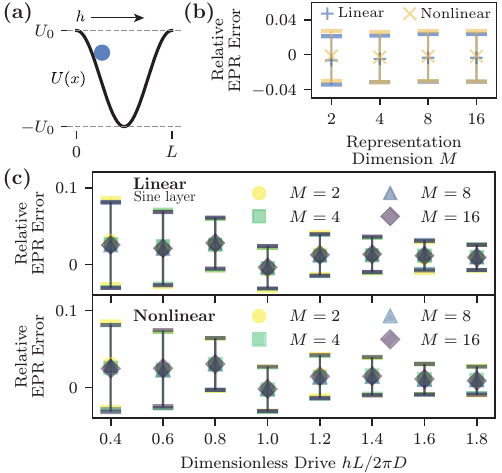} 
    \caption{\label{fig:nonlinear_dynamics}\textbf{Learning Entropy Production in Nonlinear Stochastic Dynamics}. (a) Sketch of particle in a sinusoidal potential $U(x)$ subject to a constant driving force $h$. (b) Relative (fractional) EPR error as a function of representation dimension $M$ for linear and nonlinear representations (with sine activation layer), with dimensionless drive strength $\kappa=hL/2\pi D=1$. (c) Relative EPR error as a function of dimensionless drive strength $\kappa$ for various representation dimensions $M$, for linear (with sine activation) and nonlinear representations. All error bars represent the standard error of the mean (s.e.m) over 5 random seed initializations.}
\end{figure}

LENS effectively captures the entropy production in this nonlinear dynamical system with two representation dimensions (Fig.~\ref{fig:nonlinear_dynamics}(b)) and estimates the EPR to within \SI{3}{\percent} across a range of drive strengths (Fig.~\ref{fig:nonlinear_dynamics}(c)). The relative EPR error remains relatively constant from $M=2$ to $M=16$ regardless of whether the representations are linear or nonlinear. Lower-dimensional linear representations suffice for summarizing the EPR of linear systems, while the nonlinear dynamics in this system require higher dimensional representations to capture entropy production.

\begin{figure}[t]
    \includegraphics[width=\linewidth]{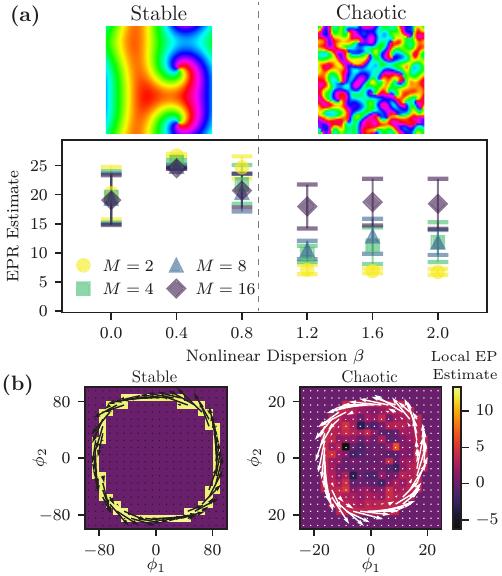} 
    \caption{\label{fig:cgl}\textbf{Learning Low-Dimensional Representations of Complex Ginzburg-Landau (CGL) Dynamics}. (a) (top) Representative images of CGL phase field in the stable (left) and chaotic (right) regimes. (bottom) Estimated EPR as a function of nonlinear dispersion $\beta$ across the defect turbulence transition (gray dotted line), for various representation dimensions $M$ with nonlinear representations $\va{\phi}_\theta$. Defect turbulence transition line from~\citet{Chate1996}. (b) Representative state space flows (arrows) and local EPR estimates (heatmaps) in representation space for $M=2$ in the stable ($\beta=0$, top) and chaotic ($\beta=2.0$, bottom) regimes. Both heatmaps share the same colorbar scale for the estimated local (single time step) EPR. All error bars represent the standard error of the mean (s.e.m) over 5 random seed initializations.}
\end{figure}

\textit{High-Dimensional Complex Dynamics}---Finally, we apply LENS to data generated from the complex Ginzburg-Landau (CGL) model~\cite{Chate1996,Aranson2002}, governed by 
\begin{align}
    \label{eq:cgl_dynamics}
    \pdv{\Tilde{A}(\va{r},t)}{t} &=\Tilde{A}(\va{r},t) + (1 + i\alpha)\grad^2 \Tilde{A}(\va{r},t) \nonumber \\
    &\qquad\qquad\qquad- (1 + i\beta)|A|^{2}\Tilde{A}(\va{r},t),
\end{align}
where $\Tilde{A}(\va{r},t)=|A(\va{r},t)|\exp(i\varphi(\va{r},t))$ is a complex field with spatiotemporal phase field $\varphi(\va{r},t)$ and amplitude $|A(\va{r},t)|$. Here, $\alpha$ and $\beta$ are the linear and nonlinear dispersions of the medium, respectively. The CGL captures statistical features of biochemical wave patterns generated by the Rho-GTP system on the membranes of starfish oocytes~\cite{tan2020topological, Li2024}.  

In two dimensions, Eq.~\ref{eq:cgl_dynamics} admits stable defect solutions for small values of $\alpha, \beta$, with defects characterized by points $\va{r}_0$ where the amplitude vanishes ($|A(\va{r}_0)|=0$) and the phase field $\varphi(\va{r})$ winds up or down by $2\pi$ around the defect point. Increasing $\beta$ leads to defect turbulence, where defects rapidly proliferate and the phase spirals are no longer easily distinguishable~\cite{Chate1996}. CGL dynamics are deterministic, and the EPR estimates diverge logarithmically since every forward step will have no corresponding observation of its reversal~\cite{Li2024}. Despite this, we investigate whether the representations are predictive of the transition from the stable to chaotic regimes.

Fixing $\alpha=-0.8$, we simulate trajectories of length $T=\num{e5}$ observations with time step $\dd t=\num{e-3}$ and estimate the EPR using LENS. The phase fields $\varphi(\va{r},t)$ for all observations are fed into a convolutional neural network (CNN)  (see App.~\ref{subsubsec:architectures} for CNN architecture). 
The EPR estimates computed at various values of $\beta$ across the defect turbulence transition are shown in Fig.~\ref{fig:cgl}(a) for different $M$. The transition is clearly visible for low $M$, where the estimated EPR drops sharply past the transition line. The larger EPR estimate for larger $M$ in the chaotic regime correlates with higher values of the LENS objective (lower binary cross-entropy loss), indicating that the larger $M$ representations are better able to extract the presence of the forward-only dynamics in the chaotic regime. This is also corroborated by the state space flows and local (single time step) EPR estimates for $M=2$ representations (Fig.~\ref{fig:cgl}(b)). 

While the flows in both regimes are largely cyclical and closed, the local EPR predicted in the stable regime is relatively constant and entirely confined to the thin ring where the representation flows lie. In contrast, the learned local EPR map in the chaotic regime is characterized by small positive local EPR along the ring of representational flows with spurious negative local EPR patches within the ring, indicating that the learned representations are unable to capture the deterministic cyclic dynamics with small $M$~\cite{suppvideo}. 

\textit{Discussion and Outlook}---Here, we propose a method for identifying relevant irreversible degrees of freedom in high-dimensional NESS systems. Importantly, our method is model-free, conceptually simple and bypasses direct density or current estimation. In line with recent work to estimate EPR~\cite{Kim2020}, our method uses black-box neural networks but differs by identifying which aspects of the dynamics are irreversible. The quadratic form of the LENS logit suggests that LENS will be particularly insightful in situations where high-dimensional, nonlinear dynamics are generated by a small set of latent, irreversible features with linear dynamics. A mathematical and empirical characterization of LENS in such scenarios is an interesting direction for future work, and could point to a meaningful mapping from nonlinear systems to linear systems with the same nonequilibrium behavior~\cite{Watter2015}. We anticipate our method will open up new avenues to analyze biological active matter systems and other complex dynamical systems with inaccessible internal states.

\textit{Code Availability}---The code for the LENS estimator is written in JAX and is available at~\cite{code}.

\textit{Acknowledgments}---The authors thank Trevor GrandPre, Junang Li, and William Bialek for insightful discussions on the subject, as well as Hugues Chat\'e and Paul Manneville for providing supplementary data for the CGL phase diagram. This material is based upon work supported by the National Science Foundation Graduate Research Fellowship Program under Grant No. DGE-2444107. Any opinions, findings, and conclusions or recommendations expressed in this material are those of the authors and do not necessarily reflect the views of the National Science Foundation. The simulations presented in this article were performed on computational resources managed and supported by Princeton Research Computing, a consortium of groups including the Princeton Institute for Computational Science and Engineering (PICSciE) and the Office of Information Technology's High Performance Computing Center and Visualization Laboratory at Princeton University. GR is partially supported by a joint research agreement between NTT Research Inc and Princeton University.

\textit{Author Contributions}---C.J. and R.R. contributed equally to this work. C.J. led project implementation and experiments; R.R. and G.R. performed theory calculations; R.R., G.R., and B.E. formulated proofs; R.R. and C.J. conducted data analysis; R.R. led paper writing. B.E. and G.R. equally advised this work. All authors designed and performed research. All authors contributed to writing and revising the paper. 


%% file: arxiv_si.tex
\title{Supplementary Information: Identifying nonequilibrium degrees of freedom in high-dimensional stochastic systems}

\author{Catherine Ji}%
\thanks{Equal contribution}
\author{Ravin Raj}%
\thanks{Equal contribution}
\affiliation{Department of Physics, Princeton University
}
\author{Benjamin Eysenbach}
 \affiliation{Department of Computer Science, Princeton University}
\author{Gautam Reddy}%
\affiliation{Department of Physics, Princeton University
}%

\date{\today}

\maketitle
\appendix

\section*{Supplementary Information}
\section{Overview of Methods} \label{app:methods}

The simulation and optimization details presented here are organized by physical system studied in the main text. Beyond implementation details, we include additional clarifying experiments on the behavior and performance of LENS. To summarize, we apply LENS to identify nonequilibrium degrees of freedom and measure entropy production in the following \emph{in silico} systems:
\begin{enumerate}
    \item {\it Linear system with an analytical baseline: }$N$ harmonically-coupled beads with overdamped dynamics,
    \item {\it Nonlinear system with an analytical baseline:}  a driven particle in a periodic potential with overdamped dynamics,
    \item {\it Nonlinear dynamical field with no analytical baseline: } 2D Complex Ginzburg-Landau (CGL) dynamics.
\end{enumerate}
The first two of these systems demonstrate that LENS correctly estimates entropy production rates (EPRs) and identifies nonequilibrium degrees of freedom in simple, exactly-solvable systems. Experimental systems with analytical baselines are necessary to gauge the accuracy and utility of the method. Beyond such systems, the CGL system demonstrates that LENS can learn informative representations for rich stable and chaotic phases. In all cases, we observe that LENS captures significant signal of irreversibility in different physical parameter regimes and phases. Furthermore, the learned representations follow trajectories with distinctive probability flows that characterize nonequilibrium steady states (NESS), where representations have reduced dimensionality compared to the original system or describe nonlinear systems by projecting to a higher dimensional space.

\subsection{General Optimization Details} \label{subsec:simulation_optimization_details}

\subsubsection{Representation and Training Architectures} \label{subsubsec:architectures}
The shared experimental setup over these systems consists of a simulation step and a training step, where models optimize the LENS objective~\eqref{eq:lens_objective_function} over the simulated trajectory data. LENS maps pairs of temporally-consecutive states $(\va{x},\va{x}')$ to their learned representations $(\va{\phi}(\va{x}), \va{\phi}(\va{x}'))$, and it should be noted that the representations do not mix degrees of freedom between these two states (Fig.~1 in the main text). The learned representations solely contain bits of information of the original state and learned parameters $\theta$, and are thus faithful learned representations of the original dynamical system. Specific system features led to additional architectural modifications that deviated from these three architectures. In general, we prioritized simple, standard architectures that led to minimal hyperparameter tuning between systems. Further details for each architecture and system are included in the subsections below. 

\begin{figure*}[t]
    \centering
    \includegraphics[width=0.65\textwidth]{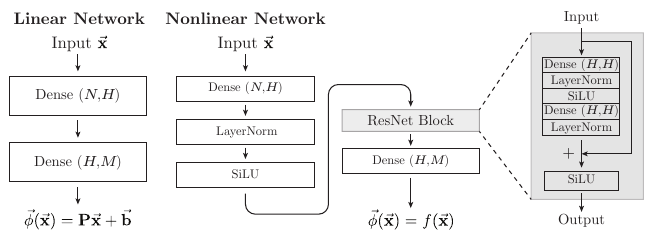}
    \caption{\textbf{Representation Architectures for LENS}. (Left) Linear network architecture with two dense (fully-connected linear) layers of width $H=256$ to obtain a linear projection $\phi:\mathbb{R}^N\mapsto\mathbb{R}^M$ of the coordinate space into representation space. (Right) Nonlinear network architecture using residual network (ResNet) blocks~\cite{He2016}, LayerNorm~\cite{Ba2016}, and SiLU activations~\cite{hendrycks2016gaussian} after each dense layer.}
    \label{fig:arch}
\end{figure*}

To maximize the LENS objective~\eqref{eq:lens_objective_function} over the simulation data, we first consider a strictly linear network with two dense (fully connected linear) layers of width $H=256$, where representations take the form $\va{\phi}(\va{x}) = \vb{P}\va{x} + \va{b}$. We refer to this architecture as the \textbf{linear} network (Fig.~\ref{fig:arch}). The goal of the linear network is to compare the learned embeddings and parameterizations with the analytical fixed points of the objective~\eqref{eq:lens_fixed_points}. The wide hidden layers enable arbitrary intermediate calculations over the original inputs; in practice, eliminating the hidden layers led to the model occasionally failing to find global minima.

The second architecture is a nonlinear network with a dense layer, a ResNet block~\cite{He2016} of width $H=256$, and a final dense layer mapping to the nonlinear representation $\vec{\phi}(\va{x})$. After each dense layer, LayerNorm~\cite{Ba2016} and SiLU~\cite{hendrycks2016gaussian} are applied (in that order), with $\text{SiLU}(z)=z\sigma(z)=z/(1+e^{-z})$.  We refer to this architecture as the \textbf{nonlinear} network in the remainder of this section. In addition to nonlinearities from the SiLU activation, this architecture normalizes feature activations to have zero mean and unit variance (LayerNorm) and includes skip connections between layers (ResNet), allowing gradients to ``skip'' intermediate layers during backpropagation. These features help stabilize training and are standard within the deep learning community~\cite{He2016,Ba2016}.

\begin{figure*}
    \centering
    \includegraphics[width=\linewidth]{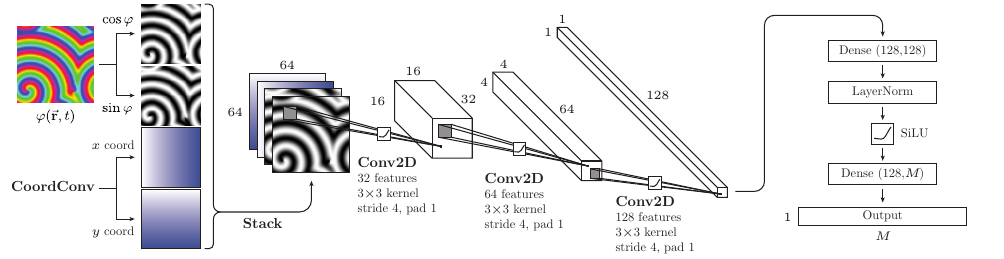}
    \caption{\textbf{Convolutional Neural Network (CNN) Architecture for LENS Analysis of CGL Dynamics}. For each $64\times 64$ image, the cosine and sine of the phase field $\varphi(\va{r},t)$ are extracted to form the first two channels. The $x$ and $y$ coordinates of each pixel are also appended as two further channels using the CoordConv procedure~\cite{Liu2018}. The resulting stack of four channels is then fed through three 2D convolution layers, each with 32, 64, and 128 features respectively. All convolutional layers use a $3\times 3$ kernel with stride 4 and zero padding of size 1, and have SiLU activations on their outputs. The output of the last convolutional layer is flattened into a single array (per image in the batch) and fed through a dense layer of width $H=128$, before using a LayerNorm and SiLU activation. The output of this nonlinear transformation is then linearly projected to $M$ dimensions using a final dense layer.}
    \label{fig:cnn_archi}
\end{figure*}

The third architecture is a convolutional neural network (CNN) (shown in Fig.~\ref{fig:cnn_archi}) for large image-based, partially-observed inputs like the 2D CGL equation dynamics. For these datasets, flattening the image inputs and passing the inputs through a ResNet would lead to a grossly overparameterized model. Thus, a CNN with aggressive downsizing of input dimensions is an appropriately lightweight architecture for image classification tasks. The CNN consists of a coordinate convolution (CoordConv) step~\cite{Liu2018}, where normalized pixel locations are appended to each pixel so the CNN can identify the locations of different salient features. The CNN consists of three convolution layers with $32$, $64$, and $128$ features respectively. Each convolution layer utilizes kernels of size $3 \times 3$, strides of size $4$, and equal padding. SiLU activations are also applied after each layer. While we experimented with GroupNorm activations after the convolutional layers (a standard feature normalization technique in computer vision that preserves local contrast), we found that GroupNorm slowed convergence in practice. The final output of the convolutions is flattened, passed through a dense layer of width 128 with applications of LayerNorm and SiLU, then linearly mapped to the representation space.

\subsubsection{Train-Test Splits}

 In the training step, simulation data points are partitioned into a training and test set comprising \SI{80}{\percent} and \SI{20}{\percent} of the total data respectively. Each dataset contains temporally consecutive data points with two data partitioning schemes. The first scheme, which we used for all systems except for the periodic potential experiments, allocates the first \SI{80}{\percent} of trajectory pairs (in time) to the training set and allocates the remaining \SI{20}{\percent} of trajectory pairs to the validation set. This scheme is suitable for systems with steady-state marginal distributions {\it in silico} with a single long trajectory (N-beads and CGL). In contrast, the periodic potential system exhibits a steady-state distribution in bulk (initialized with particles over the entire domain) but is not steady-state for a single-particle trajectory: with a constant positive drive, the position of a particle at time $t+1$ is greater than the position of the particle at $t$ at the expectation level. Thus, a randomized sampling scheme where any pair of states $(x_{t}, x_{t+1})$ is in the train/test set with \SI{80}{\percent}/\SI{20}{\percent} probability ensures that the sets are drawn from the same distribution; we refer to this scheme as the `split-by-pairs' scheme within the codebase.
 
 \subsubsection{Seeding, Hardware, and Codebase Details} All datapoints plotted are of the mean and standard deviation results over 5 random seeds, where the same random seeds initialize the starting condition of the physical simulations, parameters of the neural network, and batching during training. Furthermore, for each seed we average the test set results over the last 10 epochs of training. All simulations and experiments were run on single GPUs with at most \SI{48}{\giga\byte} of RAM. We prioritized dataset and model sizes that fit comfortably within these single-GPU memory constraints but provide additional functions within the code to handle larger datasets ($>\SI{50}{\percent}$ of the available RAM) where data-loading from CPU to the GPU is done on-the-fly. For code development, we used Weights and Biases~\cite{wandb} to track hyperparameters and training runs. Finally, the overall structure of our training code is adapted from the JaxGCRL codebase~\cite{Bortkiewicz2025} and incorporates several helper functions from that repository.

\subsection{\texorpdfstring{$N$-Beads Model}{N-Beads Model}}

The dynamics of the $N$-bead system can be described by the linear overdamped Langevin equation
\begin{equation}
    \label{eq:n_beads_langevin_app}
    \va{x}_{t+\dd t}=\va{x}_{t}+\vb{G}\va{x}_{t}\,\dd t+\vb{F}\va{\xi}_{t}\sqrt{\dd t}
\end{equation}
where the deterministic coupling matrix $\vb{G}$ has elements defined by $G_{ij}=(-2\delta_{i,j} + \delta_{i,j+1} + \delta_{i,j-1}) k/\gamma$, and the matrix $\vb{F}$ sets the local thermal noise scale. Here, we assume that all springs are identical and that each bead experiences uncorrelated noise whose strength is parametrized by the local temperature, so $\vb{F}$ is diagonal with elements $F_{ii}=\sqrt{2k_BT_i/\gamma}$ where $T_i$ is the temperature of the $i$-th bead and $\gamma$ is the viscous damping coefficient of the medium. We also define $\va{x}_{t}\in\mathbb{R}^N$ as a vector of bead displacements at time $t$, and $\va{\xi}_{t}$ as an independent realization of Gaussian white noise whose components satisfy $\mathbb{E}[\xi_i(t)]=0$ (zero mean) and $\mathbb{E}[(\xi_t)_i(\xi_{t'})_j]=\delta_{i,j}\delta(t-t')$ (independent components with unit variance), where the expectation is computed over infinitely many realizations of the noise. For all simulations, we set the autocorrelation timescale $\gamma/k=1$ and work in units where $k_B=1$.

Trajectories are obtained for these dynamics by numerically propagating this equation using the Euler-Maruyama scheme for some large $T$ number of time steps. We use a simulation time step of $\dd t=\num{e-2}$, which is much smaller than the autocorrelation time. The analytical expressions for the EPR in this model are provided in App.~\ref{app:n_bead_analytic_model}.

Due to function approximation errors, numerical errors, and difficulty in the binary classification task for small EPRs, the estimated ensemble-averaged rates occasionally fell below zero when the true EPR was of order $O(10^{-6})$. Thus, in the main text figures, we imposed a strict minimum ensemble EPR equal to a lower bound on the positive floating-point rounding error $\sim \mathcal{O}(\num{e-8})$ where $\hat{\dot{S}} = \min(\hat{\dot{S}}, \num{e-8})$. We did not impose this positivity condition in the training or simulation code itself.

\subsubsection{\texorpdfstring{Fully-Observed $N=2$ Bead System}{Fully-Observed N=2 Bead System}}

\begin{table}[ht]
    \centering
    \caption{{\bf Simulation and Training Hyperparameters of the $N=2$ Bead System.} Hyperparameter values correspond to those used for data plotted in Fig.~2b in the main text.}
    \label{tab:n=2_sim}
    \begin{tabular}{cc}
        \toprule
        Parameter & Value(s) \\
        \midrule
        Number of beads ($N$) & $2$ \\
        Temperature ratio $(T_c/T_h)$& \makecell{\{0.05, 0.1, 0.25, 0.5,\\0.90, 0.999, 0.9999\}} \\
        Time step $(\dd t_{\text{sim}})$ & \num{e-2} \\
        Trajectory length $(T)$ & \num{e7} \\
        Seeds & $\{0,1,2,3,4\}$ \\
        Learning rate ($\eta$) & \num{3e-5} \\
        Batch size & $4096$ \\
        Maximum training epochs & \num{e6} \\
        Hidden layer dimension ($H$) & $256$ \\
        Representation dimension ($M$) & $16$ \\
        Early stopping patience & $7$ \\
        \bottomrule
    \end{tabular}
\end{table}

In Fig.~2b of the main text, we used the linear and nonlinear network architectures as detailed in Fig.~\ref{fig:arch} with the LENS objective and representation dimension $M=16$. Deep ML techniques often suffer from overfitting, where expressive models improve prediction accuracy over a finite training set at the cost of poor generalization on an unseen test set. Thus, we paused training when the $7$-epoch change in validation loss was positive for $7$ epochs (see ``early stopping patience'' parameter). We note that early stopping is standard in ML and have been used in prior ML-based EPR estimation methods~\cite{Vodret2024,Lyu2025}. All relevant simulation and training hyperparameters are provided in Table~\ref{tab:n=2_sim}.

\subsubsection{\texorpdfstring{Fully-Observed $N=16$ Bead System}{Fully-Observed N=16 Bead System}}

Here, we include the experimental details of Fig.~2c. The goal of these experiments is to demonstrate that LENS can learn meaningful representations of high-dimensional systems. A secondary goal of these experiments is to test the efficacy of LENS at estimating EP across different EPR scales with a higher-dimensional system.

\begin{table}[ht]
    \centering
    \caption{{\bf Simulation and Training Hyperparameters of the $N=16$ Bead System.} Hyperparameter values correspond to those used for data plotted in Fig.~2c and 2d in the main text.}
    \label{tab:n=16_sim}
    \begin{tabular}{cc}
        \toprule
        Parameter & Value(s) \\
        \midrule
        Number of beads ($N$) & $16$ \\
        Temperature ratio $(T_c/T_h)$ & $0.1$\\
        Temperature profile & \makecell{\{linear, fourier,\\sine, alternate\}} \\
        Time step $(\dd t_{\text{sim}})$ & \num{e-2} \\
        Trajectory length $(T)$ & \num{e7} \\
        Seeds & $\{0,1,2,3,4\}$ \\
        Learning rate ($\eta$) & \num{3e-5} \\
        Batch size & $4096$ \\
        Maximum training epochs & \num{2e6} \\
        Hidden layer dimension ($H$) & $256$ \\
        Representation dimension ($M$) & $\{1,2,4,8,16\}$ \\
        Early stopping patience & $7$ \\
        \bottomrule
    \end{tabular}
\end{table}

Similar to the $N=2$ bead system, we simulate a system of $N=16$ beads following Eq.~\ref{eq:n_beads_langevin_app} over $T=\num{e7}$ transitions with the Euler-Maruyama scheme, where each transition corresponds to a (dimensionless) timestep of size $\dd t = \num{e-2}$. The $N=16$ bead system is simulated over four different temperature profiles pictured in Fig.~2a. We consider the first Fourier mode interpolating $T_{c} = 10$ and $T_{h} =100$ (\texttt{fourier}); linear temperature gradient between $T_{c}$ and $T_{h}$ (\texttt{linear}); alternating temperature profile that jumps between $T_{c}$ and $T_{h}$ for every subsequent bead (\texttt{alternate}); and half of a sine wave with ends at $T_{c}$ and a peak at temperature $T_{h}$ (\texttt{sine}). We then optimize the LENS objective for the $N=16$ bead systems with the four different temperature profiles. We consider both linear and nonlinear network architectures and different representation dimensions of $M=1,2,4,8,16$. Functionally, varying $M$ corresponds to varying output dimensions of the final dense layer and dimensions of the learned $M\times M$ matrices $\mathbf{A}_{\theta}$ and $\mathbf{B}_{\theta}$. All relevant simulation and training hyperparameters are provided in Table~\ref{tab:n=16_sim}.

The plots in Fig.~2d and 2f show the flows of learned linear representations $\phi_{\theta}$ for the `alternate' temperature profile. Here, the representations are of dimension $M=2 < 16$ and are of data points in the (held out) test set. In Fig.~2d, the inset shows the corresponding representation dynamics derived using a naive whitened PCA on the original dataset with the top two principal components (PCs). The flow is computed as the vector field $\va{\phi}(\va{x}_{t+1}) - \va{\phi}(\va{x}_{t})$ averaged over the trajectory and plotted at the interpolated midpoint $(\va{\phi}(\va{x}_{t+1}) + \va{\phi}(\va{x}_{t})) / 2$.

\subsubsection{\texorpdfstring{Fully-Observed $N=32$ and $N=64$ Bead Systems}{Fully-Observed N=32 and N=64 Bead Systems}}

\begin{figure*}[ht]
    \centering
    \includegraphics[width=0.75\linewidth]{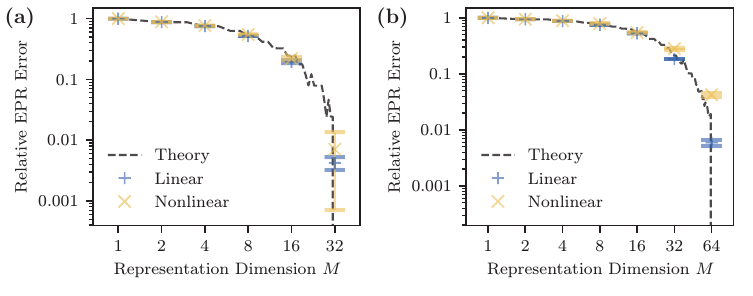}
    \caption{\textbf{LENS Application to $N$-Bead Systems for $N=32,64$}. (a) Relative (fractional) EPR error as a function of linear representation dimension $M$ for $N=32$ bead-spring model with alternate temperature profile. (b) Relative (fractional) EPR error as a function of linear representation dimension $M$ for $N=64$ bead-spring model with alternate temperature profile. The corresponding theory curves are obtained by optimizing the reduced LENS objective~\eqref{eq:reduced_lens_objective}, with small irregularities arising due to random variations in the stochastic gradient descent optimization procedure. All error bars represent the standard error of the mean (s.e.m) over 5 random seed initializations.}
    \label{fig:n=32,64}
\end{figure*}

We include additional experiments for the `alternate' profile with $N=32$ and $N=64$ beads to further stress-test the LENS method for higher-dimensional, fully-observed inputs. The simulation and training details are identical to that of the other fully-observed $N$-bead systems. The results in Fig.~\ref{fig:n=32,64} show that LENS successfully estimates EPR for the `alternate' $N=32$ and $N=64$ systems with representation dimensions $M=32$ and $M=64$ that match that of the original system. Lower dimensional representations are still able to capture a significant portion of the EPR in accordance with the theory prediction. All relevant simulation and training hyperparameters are provided in Table~\ref{tab:n=32_64_sim}.

\begin{table}[ht]
    \centering
    \caption{{\bf Simulation and Training Hyperparameters of the $N=32,64$ bead Systems.} Hyperparameter values correspond to those used for data plotted in Fig.~\ref{fig:n=32,64}. $^\dag M=64$ is only applied to the $N=64$ case.}
    \label{tab:n=32_64_sim}
    \begin{tabular}{cc}
        \toprule
        Parameter & Value(s) \\
        \midrule
        Number of beads ($N$) & $\{32,64\}$ \\
        Largest temperature ratio $(T_c/T_h)$ & $0.1$\\
        Temperature profile & alternate \\
        Time step $(\dd t_{\text{sim}})$ & \num{e-2} \\
        Trajectory length $(T)$ & \num{e7} \\
        Seeds & $\{0,1,2,3,4\}$ \\
        Learning rate ($\eta$) & \num{3e-5} \\
        Batch size & $4096$ \\
        Maximum training epochs & \num{2e6} \\
        Hidden layer dimension ($H$) & $256$ \\
        Representation dimension ($M$) & $\{1,2,4,8,16,32,64^\dag\}$ \\
        Early stopping patience & $7$ \\
        \bottomrule
    \end{tabular}
\end{table}

\subsubsection{\texorpdfstring{Rendered $N=16$ Bead System}{Rendered N=16 Bead System}}

We also conduct experiments involving ``partially observed'' trajectories, reminiscent of real experimental data where not all true degrees of freedom are directly observed, but rather a subset of observables (e.g. binned sensor data, pixel data in image stacks, etc.) is used to estimate the EPR for a system under study. We test LENS in this partially-observed regime with rendered image sequences of the $N=16$ bead system (example in Fig.~\ref{fig:rendered_n=16_example}) and detail the experimental procedures below.

\begin{figure*}[t]
    \centering
    \includegraphics[width=\linewidth]{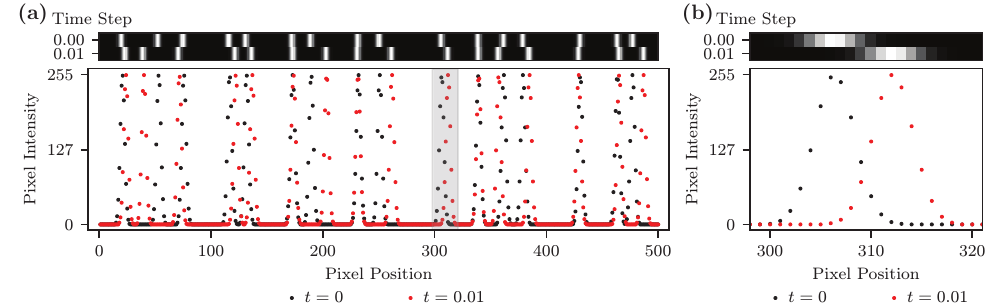}
    \caption{\textbf{Representative Frames of Rendered $N=16$ Bead-Spring System}. (a) (Top) Representative 1D renders of Gaussian-convolved bead displacements for two consecutive frames (vertically stretched for clarity). (Bottom) Scatterplot of corresponding pixel intensities from the two frames above. (b) Zoom of gray-highlighted region in (a).}
    \label{fig:rendered_n=16_example}
\end{figure*}

\begin{table}[t]
    \centering
    \caption{{\bf Simulation and Training Hyperparameters of the Rendered $N=16$ Bead System.} Hyperparameter values correspond to those used for data plotted in Fig.~2e and 2f in the main text.}
    \label{tab:rendered_n=16}
    \begin{tabular}{cc}
        \toprule
        Parameter & Value(s) \\
        \midrule
        Number of beads ($N$) & $16$ \\
        Temperature ratio $(T_c/T_h)$ & 0.1 \\
        Temperature profile & alternate \\
        Time step $(\dd t_{\text{sim}})$ & \num{e-2} \\
        Seeds & $\{0,1,2,3,4\}$ \\
        Learning rate ($\eta$) & \num{3e-5} \\
        Batch size & $4096$ \\
        Maximum training epochs & \num{e6} \\
        Hidden layer dimension ($H$) & $256$ \\
        Representation dimension ($M$) & \{$1,2,4,8,16$\} \\
        Early stopping patience & $10$ \\
        \bottomrule
    \end{tabular}
\end{table}

The frame renders are obtained by converting each of the \num{e7} bead position vectors into extended 1D arrays of ``blurred'' bead displacements. Since the dynamics of the system are in terms of the relative displacement and agnostic to the equilibrium distance between consecutive beads, the equilibrium distance between adjacent beads is initialized as twice of the maximal absolute displacement $x_i$ of any bead. This transformation guarantees that the ordering of the beads from left to right is constant throughout the trajectory, though the beads' point spread functions may overlap. Additional extensions to the current LENS parameterization can be made to take into account this form of partial observability. For example, one can instead construct logits that mix temporal degrees of freedom similar to CNEEP~\cite{Bae2022}, collapse together temporally adjacent states, or add an additional learning step for labeling and tracking objects in time with developed tools \cite{stamhuis2014pivlab}. However, these extensions are beyond the scope of this paper. 

The final rendered images were each of dimension $500 \times 1$, with a Gaussian approximately centered at each bead's true position. We opted not to increase the dimensionality of the rendered image to 2D Gaussian-blurred beads due to GPU VRAM limits, in light of rendering \num{e7} frames for an accurate comparison with the fully-observed results. We optimize the LENS objective using a variant of the nonlinear network architecture (without LayerNorm), and scan over representation dimensions $M$. Due to the relatively small input size of the rendered images (500-dimensional vector) compared to the $64 \times 64$ CGL images, we opted for the nonlinear ResNet architecture rather than the more complex CNN architecture. LayerNorms were omitted due to the observation that they slowed learning within CNN architectures; we opted to leave out the LayerNorm given the high performance of this simpler architecture.

\subsection{\texorpdfstring{Comparisons between Binary Cross Entropy (BCE) and Variational $f$-Divergence Objectives}{Comparisons between Binary Cross Entropy (BCE) and Variational f-Divergence Objectives}}

A core contribution of LENS is the parameterization of EP as a function of state representations, where the local (one-step) EP takes a quadratic form characteristic of linear stochastic dynamical systems (see Sec.~\ref{subsubsec:linear_stochastic_dynamics}). However, the question remains on the efficacy of the specific form of the LENS objective as a binary cross entropy (BCE) objective. We consider the alternative $f$-divergence formulation of the EP estimator with the LENS logit parameterization. Prior entropy production estimators have used the $f$-divergence formulation such as NEEP~\cite{Kim2020}, follow-ups to NEEP~\cite{Bae2022}, with parametrizations that allow the one-step EP $S(\va{x},\va{x}')$ to be an arbitrary function of $\va{x}$ and $\va{x}'$.

\subsubsection{\texorpdfstring{Stability of BCE and $f$-divergence Methods}{Stability of BCE and f-divergence Methods}}
While the BCE objective is equivalent to the variational $f$-divergence objective used in NEEP for infinitesimal time steps (see App.~\ref{app:lens_neep_equivalence}), the two objectives do not result in the same empirical performance. The simultaneous optimization of $S(\va{x}, \va{x}')$ and $e^{S(\va{x},\va{x}')}$ terms in the $f$-divergence lead to unstable training and exploding gradients for large EPR, a phenomenon that is well-documented in machine learning literature and recovered in Fig.~\ref{fig:lens_neep_comparison}. Thus, while $f$-divergence objectives provide compelling theoretical interpretations discussed in App.~\ref{app:lens_low_rank_approximation}, we explicitly use the BCE form of the LENS objective for all results in the main text.

\begin{figure*}
    \centering
    \includegraphics[width=0.75\linewidth]{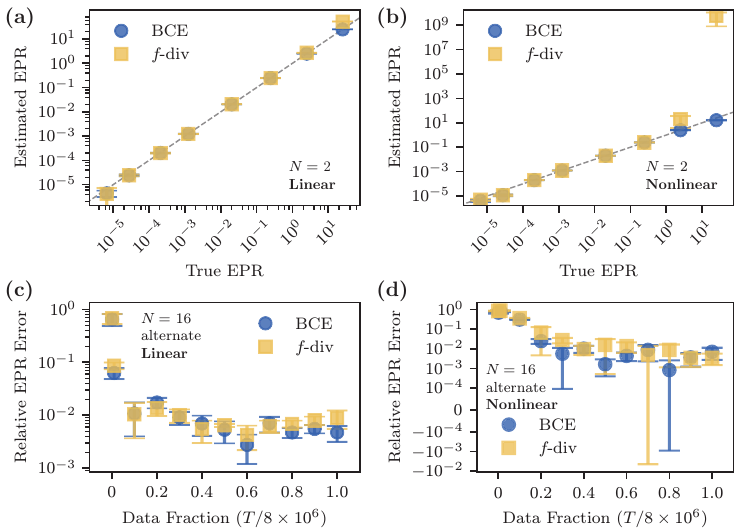}
    \caption{\textbf{Stability and Sample Efficiency Comparison between BCE and $f$-divergence Objectives}. (a, b) Estimated EPR for $N=2$ bead-spring model swept over $T_c/T_h$, with $M=16$ linear (a) and nonlinear (b) representations using BCE (LENS) (blue circles) and $f$-divergence minimization (NEEP) (yellow squares). All relevant simulation and training hyperparameters are provided in Table~\ref{tab:lens_neep_n=2}. (c, d) Relative (fractional) EPR error as a function of training dataset size for $N=16$ bead-spring model with alternate temperature profile. These are shown for $M=16$ linear (c) and nonlinear (d) representations with BCE (LENS) (blue circles) and $f$-divergence minimization (NEEP) (yellow squares). All error bars represent the standard error of the mean (s.e.m) over 5 random seed initializations. All relevant simulation and training hyperparameters are provided in Table~\ref{tab:lens_neep_n=16}.}
    \label{fig:lens_neep_comparison}
\end{figure*}

\begin{table}[ht]
    \centering
    \caption{{\bf Simulation and Training Hyperparameters of the $N=2$ Bead System for LENS/NEEP Accuracy Comparison.} Hyperparameter values correspond to those used for data plotted in Fig.~\ref{fig:lens_neep_comparison}(a),(\ref{fig:lens_neep_comparison}(b)).}
    \label{tab:lens_neep_n=2}
    \begin{tabular}{cc}
        \toprule
        Parameter & Value(s) \\
        \midrule
        Objective & \makecell{\{BCE (LENS),\\NEEP\}} \\
        Number of beads ($N$) & $2$ \\
        Temperature ratio $(T_c/T_h)$ & \makecell{\{0.05, 0.1, 0.25, 0.5,\\0.90, 0.999, 0.9999\}} \\
        Time step $(\dd t_{\text{sim}})$ & \num{e-2} \\
        Trajectory length $(T)$ & \num{e7} \\
        Seeds & $\{0,1,2,3,4\}$ \\
        Learning rate ($\eta$) & \num{3e-5} \\
        Batch size & $4096$ \\
        Maximum training epochs & \num{e6} \\
        Hidden layer dimension ($H$) & $256$ \\
        Representation dimension ($M$) & $16$ \\
        Early stopping patience & $7$ \\
        \bottomrule
    \end{tabular}
\end{table}

\subsubsection{Sample Efficiency of LENS and Overparametrization}
Although a thorough survey of EPR estimation and training set size is beyond the scope of this work, we present a brief set of empirical results on the sample efficiency of LENS which suggests LENS can be further optimized for small datasets. Notably, in the case of $N=16,M=16$ and alternate temperature profile, a training set size of $D = \num{8e4}$ captures $>10\%$ EP for the linear and MLP architectures, and a set of $D = \num{8e5}$ datapoints captures most of the entropy production (Fig.~\ref{fig:lens_neep_comparison}(c),(d)). We note that the model is overparameterized relative to the dataset size -- the linear networks have \num{74544} parameters and the nonlinear networks have \num{141872} parameters, comparable to or exceeding the number of training points. While this overparameterization is not a concern in itself with proper control of overfitting via early stopping, it is highly likely that sample efficiency can be further improved by limiting the model size. For example, in the case of linear stochastic dynamics, one can reduce `model size' by directly learning an orthogonal projection matrix $\vb{P}_{\theta}$ and bias $\va{b}_{\theta}$ that parameterizes representation $\va{\phi}_{\theta}(\va{x}) = \vb{P}\va{x} + \va{b}$. We include this functionality within the code (see \texttt{use\_ortho\_p} flag and Sec.~\ref{subsec:orthogonal_projection_learning}) and leave further explorations on the direction of sample efficiency for future work.

\begin{table}[ht]
    \centering
    \caption{{\bf Simulation and Training Hyperparameters of the $N=16$ Bead System for LENS/NEEP Sample Efficiency Comparison.} Hyperparameter values correspond to those used for data plotted in Fig.~\ref{fig:lens_neep_comparison}(c),(d).}
    \label{tab:lens_neep_n=16}
    \begin{tabular}{cc}
        \toprule
        Parameter & Value(s) \\
        \midrule
        Objective & \makecell{\{BCE (LENS),\\$f$-divergence (NEEP)\}} \\
        Number of beads ($N$) & $16$ \\
        Temperature ratio $(T_c/T_h)$ & \makecell{\{0.05, 0.1, 0.25, 0.5,\\0.90, 0.999, 0.9999\}} \\
        Time step $(\dd t_{\text{sim}})$ & \num{e-2} \\
        Training dataset fraction $(T)$ & \makecell{\{0.001, 0.01, 0.1, 0.2,\\0.3, 0.4, 0.5, 0.6,\\0.7, 0.8, 0.9, 1.0\}} \\
        Seeds & $\{0,1,2,3,4\}$ \\
        Learning rate ($\eta$) & \num{3e-5} \\
        Batch size & $1024$ \\
        Maximum training epochs & \num{e6} \\
        Hidden layer dimension ($H$) & $256$ \\
        Representation dimension ($M$) & $16$ \\
        Early stopping patience & $7$ \\
        \bottomrule
    \end{tabular}
\end{table}

\subsection{Directly Learning Projection Matrices for Linear Stochastic Dynamics} \label{subsec:orthogonal_projection_learning}
For linear systems, rather than learning a transformation of inputs via a linear neural network, we can directly fit an affine mapping $\va{\phi}(\va{x}) = \vb{P}\va{x} + \va{b}$. Directly learning $(\vb{P},\va{b})$ eliminates the need for a large network and significantly reduces parameter count. Furthermore, this parameterization can investigate whether LENS learns a consistent set of nonequilibrium degrees of freedom for linear systems across seeds, as the explicit parameterization of $\vb{P}$ allows us to easily `fix the gauge', i.e., specify which representations should be selected amongst those that estimate EPR equally well. 

We consider the set of possible constraints that can be imposed on $\vb{P}$. Given a function $S_{\theta}(\va{x},\va{x}')$, there are an infinite number of choices for $\vb{P}$, $\vb{A}$, $\vb{B}$, and $\va{b}$ that define the same function.  In general, it is not possible to simultaneously block-diagonalize the skew-symmetric matrix $\vb{A}$ (using a special orthogonal transformation) and diagonalize the symmetric matrix $\vb{B}$~\cite{Matsumoto2025}. We partially fix the gauge by setting $\vb{A}$ to be a skew-symmetric $M\times M$ block-diagonal matrix, $\vb{B}$ to be a symmetric (but otherwise unconstrained) matrix, and $\vb{P}$ to be an orthogonal matrix or semi-orthogonal matrix.

With these constraints, we find that the absolute values of the eigenvalues of the resulting learned $\vb{A}$ and $\vb{B}$ matrices are gauge invariant. That is, these values are highly reproducible (across $M$) over random initializations of $\vb{P}$ (Fig.~\ref{fig:p_matrix}(a)), where $\vb{P}$ is initialized as a random orthogonal matrix \cite{mezzadri2007generaterandommatricesclassical}. The results indicate that LENS consistently learns the same nonequilibrium degrees of freedom for a high-dimensional linear system (up to rotations within subspaces). While directly learning affine transformations of the input data is a meaningful approach here, it should be noted that we cannot use a linear parameterization for nonlinear stochastic dynamics, and so we do not adopt the method for other systems. Further empirical and analytical work is required to investigate whether LENS can learn consistent representations $\phi$ in nonlinear systems.

\begin{figure*}[t]
    \centering
    \includegraphics[width=\linewidth]{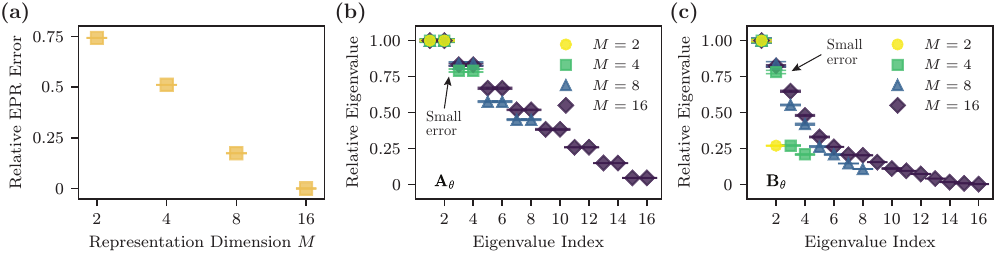}
    \caption{\textbf{Accuracy and Replicability of Directly-Learnt Projection Matrices}. (a) Relative (fractional) EPR error as a function of representation dimension $M$ for $N=16$ bead-spring model with alternate temperature profile, using a direct projection matrix. (b) Eigenvalues of antisymmetric matrix $\vb{A}_\theta$ for varying representation dimension $M$. (c) Eigenvalues of symmetric matrix $\vb{B}_\theta$ for varying representation dimension $M$. Eigenvalues are not manually sorted, and all ordering is determined by the block diagonalization procedure. All error bars represent the standard error of the mean (s.e.m) over 5 random seed initializations. All relevant simulation and training hyperparameters are provided in Table~\ref{tab:p_matrix}.}
    \label{fig:p_matrix}
\end{figure*}

\begin{table}[ht]
    \centering
    \caption{{\bf Simulation and Training Hyperparameters of the $N=16$ Bead System for Directly-Learned Linear Projections} Hyperparameter values correspond to those used for data plotted in Fig.~\ref{fig:p_matrix}.}
    \label{tab:p_matrix}
    \begin{tabular}{cc}
        \toprule
        Parameter & Value(s) \\
        \midrule
        Number of beads ($N$) & $16$ \\
        Projection Matrix $P$ Shape & $(16, 16)$\\
        Representation dimension ($M$) & $16$ \\
        Constraint on $P$ & Orthonormal \\
        Temperature ratio $(T_c/T_h)$ & $0.1$\\
        Temperature profile & alternate \\
        Time step $(\dd t_{\text{sim}})$ & \num{e-2} \\
        Trajectory length $(T)$ & \num{e7} \\
        Seeds & $\{0,1,2,3,4\}$ \\
        Learning rate ($\eta$) & \num{3e-3} \\
        Batch size & $4096$ \\
        Maximum training epochs & \num{2e6} \\
        Early stopping patience & $7$ \\
        \bottomrule
    \end{tabular}
    \label{tab:orthog}
\end{table}

\subsection{Driven Overdamped Particle in a Periodic Potential}

The dynamics of an overdamped particle in a periodic potential being driven by a constant force $h$ is given by
\begin{equation}
    \label{eq:periodic_potential_dynamics_equation}
    x_{t+1}=x_t+\qty(h-\dv{U}{x})\dd t+\sqrt{2D \,\dd t}\,\xi_t,
\end{equation}
where $x$ is the position of the particle, $U(x)$ is the potential with boundary condition $U(x+L)=U(x)$, $D$ is the diffusion constant, $t$ is time and $\xi_t$ is an independent and identically distributed (i.i.d) standard normal random variable at each time $t$. In our study, we use the specific potential $U(x)=U_0\cos(2\pi x/L)$, where $U_0$ is the potential amplitude. 

To aid the numerical simulation of trajectories, we introduce the following dimensionless variables:
\begin{align}
    \theta&=\frac{2\pi x}{L},\quad \tau=\frac{(2\pi)^2D}{L^2}t, \nonumber \\
    \kappa&=\frac{hL}{2\pi D},\quad q=\frac{U_0}{D}
\end{align}
where $\theta$ is the dimensionless position (angular variable), $\tau$ is the dimensionless time unit, $\kappa$ is the dimensionless drive strength, and $q$ is the dimensionless potential strength. This leads to a dimensionless form of~\eqref{eq:periodic_potential_dynamics_equation} given by
\begin{equation}
    \dd\theta=(\kappa+q\sin\theta)\,\dd\tau+\sqrt{2\,\dd \tau}\,\xi_t
\end{equation}
where $\xi_\tau\sim\mathcal{N}(0,1)$ is an i.i.d Gaussian random variable at each time step. In this form, the dimensionless autocorrelation time $\tau_{\mathrm{ac}}$ is identically equal to 1 in the absence of a drive $(\kappa=0)$ and potential $(q=0)$. In the presence of a potential and drive, this autocorrelation time generally decreases. It should be noted that the autocorrelation is computed using the observable $\cos\theta$ (as opposed to $\theta$) since the correlator $\expval{\theta(\tau)\theta(0)}$ diverges in the long time limit, while the periodic nature of the problem justifies the use of a periodic observable for computing the autocorrelator. For all simulations, we choose a time step of $\dd\tau=\num{e-5}$, which always satisfies the condition $\dd t\ll \tau_{\mathrm{ac}}$.

Trajectories are obtained for these dynamics by numerically integrating equation using the Euler-Maruyama scheme for some large $T$ number of time steps. We then impose the periodicity more explicitly by transforming $x\rightarrow x\mod{L}$, which keeps the trajectories bounded on the interval $[0,L)$. While we choose $\dd \tau=\num{e-5}$ for numerical evolution, we only record the positions at every \num{e2} steps to reduce systematic errors from discretization. In this manner, we save trajectories of length \num{e7} steps of effective time step $\dd \tau = \num{e-3}$ for training. Further hyperparameter details are in Table~\ref{tab:periodic_hyperparams}. The analytical expressions for the EPR in this model are provided in App.~\ref{app:periodic_potential_analytic}.

\begin{table}[ht]
    \centering
    \caption{{\bf Simulation and Training Hyperparameters of the Particle in a Periodic Potential System.} Hyperparameter values correspond to those used for data plotted in Fig.~\ref{fig:periodic_sine_activation}.}
    \label{tab:periodic_hyperparams}
    \begin{tabular}{cc}
        \toprule
        Parameter & Value(s) \\
        \midrule
        Potential amplitude $(q)$ & $1$\\
        Drive strength $(\kappa)$ & \makecell{\{0.4, 0.6, 0.8, 1.0,\\1.2, 1.4, 1.6, 1.8\}}\\ 
        Projection Matrix $P$ Shape & $(16, 16)$\\
        Representation dimension ($M$) & $\{1,2,4,8,16\}$ \\
        Time step $(\dd t_{\text{sim}})$ & \num{e-3} \\
        Trajectory length $(T)$ & \num{e7} \\
        Seeds & $\{0,1,2,3,4\}$ \\
        Learning rate ($\eta$) & \num{3e-5} \\
        Batch size & $4096$ \\
        Maximum training epochs & \num{1e6} \\
        Early stopping patience & $10$ \\
        Network architectures & \makecell{\{Linear,\\Nonlinear\}} \\
        Sine activation & \{True, False\} \\
        \bottomrule
    \end{tabular}
\end{table}

\subsubsection{Recovering Steady-State Distributions}

While the system exists in a steady state in bulk (where the system is initialized with particles over the entire domain), a single particle in a periodic potential is not in steady-state. Specifically, the position of a particle at time $t+1$ is greater than the position of the same particle at $t$ at the expectation level, due to the constant drive. Thus, additional modifications must be made to a single trajectory to obtain a distribution closer to the nonequilibrium steady-state. Since the potential is periodic, the steady state distribution is easily shown to be periodic (assuming normalization over the periodic domain). We thus apply a modulo $2\pi$ operation to trajectories $\{\theta_i\}_{i=1}^T$ which ensures the input domain to the networks is bounded.

\subsubsection{Optimization Details}

\begin{figure*}[t]
    \centering
    \includegraphics[width=\linewidth]{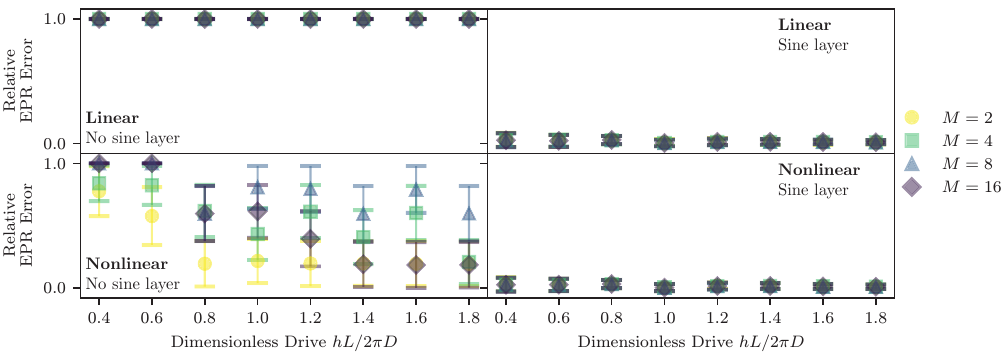}
    \caption{\textbf{Dependence of Optimization Success on Sine Activation Layer}. Relative (fractional) EPR error as a function of dimensionless drive strength $\kappa$ for various representation dimensions $M$, for linear and nonlinear representations (top and bottom rows respectively), with and without sine activations (left and right columns respectively). With the sine activation layer after the first dense (fully-connected linear) layer, the error drops across all representation dimensions to within \SI{3}{\percent}. Without the sine activation layer, linear networks fail to learn any signal of irreversibility while nonlinear networks struggle to do so consistently. All error bars represent the standard error of the mean (s.e.m) over 5 random seed initializations.}
    \label{fig:periodic_sine_activation}
\end{figure*}

For this system, we chose to add an additional sine activation layer ($\sigma_s(z)=\sin z$) at the output of the first dense (fully-connected linear) layer in each network. Prior work in ML has explored using sine activations functions for natural signals \cite{sitzmann2020implicit}. The addition of the sine layer encourages the networks to learn the periodic nature of the system by directly incorporating the periodicity into the computation of the LENS logit. This also has the added effect of ensuring that the output of the first learned dense layer lies within the interval $[-1,1]$. We found that the addition of this sine activation was necessary for stable maximization of the LENS objective in this system, as shown in Fig.~\ref{fig:periodic_sine_activation}. Particularly, the omission of the sine layer resulted in linear networks failing to converge to the true EPR across all drive strengths $\kappa$ and representation dimensions $M$. In the case of nonlinear networks, they were able to sporadically learn the right representations and thus arrive at a better estimate of the true EPR, but this was not consistent in any parameter regime.

\subsection{2D Complex Ginzburg-Landau Dynamics}
The 2D complex Ginzburg-Landau (CGL) equation is given by
\begin{align}
    \label{eq:cgl_dynamics_app}
    \pdv{\Tilde{A}(\va{r},t)}{t} &=\Tilde{A}(\va{r},t) + (1 + i\alpha)\grad^2 \Tilde{A}(\va{r},t) \nonumber \\
    &\qquad\qquad\qquad- (1 + i\beta)|A|^{2}\Tilde{A}(\va{x},t)
\end{align}
where $\Tilde{A}(\va{r},t)=|A(\va{r},t)|\exp(i\varphi(\va{r},t))$ is a complex field with spatiotemporal phase field $\varphi(\va{r},t)$ and amplitude $|A(\va{r},t)|$. Here, $\alpha$ and $\beta$ are the linear and nonlinear dispersions of the medium respectively.

We follow the ETDRK2 stepping scheme~\cite{Cox2002} for the 2D CGL equations suited for stiff, coupled PDEs. We simulate $64 \times 64$ grids of spins for trajectories of length $T=\num{e5}$ steps. Further hyperparameters for the simulation are in Table \ref{tab:cgl}.

\begin{table}[ht]
    \centering
    \caption{{\bf Simulation and Training Hyperparameters of the 2D Complex Ginzburg-Landau System.} Hyperparameter values correspond to those used for data plotted in Fig.~4 of the main text.}
    \label{tab:cgl_hyperparams}
    \begin{tabular}{cc}
        \toprule
        Parameter & Value(s) \\
        \midrule
        Grid size ($L$) & $64$ \\
        Linear dispersion ($\alpha$) & \makecell{$-0.8$} \\
        Nonlinear dispersion ($\beta$) & \makecell{\{$0.0, 0.4, 0.8, 1.2, 1.6, 2.0$\}}\\
        Time step $(\dd t_{\text{sim}})$ & \num{e-3} \\
        Trajectory length $(T)$ & \num{e5} \\
        Seeds & $\{0,1,2,3,4\}$ \\
        Learning rate ($\eta$) & \num{3e-5} \\
        Batch size & $1024$ \\
        Convolution feature count & \makecell{$[32, 64, 128]$}\\
        Maximum training epochs & \num{e5} \\
        Hidden layer dimension ($H$) & $256$ \\
        Representation dimension ($M$) & \makecell{$\{2, 4, 8, 16\}$} \\
        Early stopping patience & $7$ \\
        Coordinate convolution scale & \num{e-1} \\
        \bottomrule
    \end{tabular}
    \label{tab:cgl}
\end{table}

\subsubsection{Transition from Frozen States to Defect Turbulence}

The 2D CGL phase diagram is rich with different phases and transitions~\cite{Chate1996}. In this work, we scanned parameters where the simulated CGL dynamics visibly transition from stable or metastable states to chaotic states. Linear analysis of plane wave solutions gives the Benjamin-Feir (BF) line where stable plane waves exist for $1 - \alpha \beta > 0$ (up to finite size effects in a simulated grid). However, prior work has noted that crossing the BF line does not necessarily correspond to a transition to turbulence -- stable or metastable frozen states can exist past the BF line, where these states take the form of spirals that emanate from positive and negative-winding defects and meet at shock lines~\cite{Chate1996}. Thus, instead of considering the transition across the BF line, we consider the onset of defect turbulence, a phase where defect densities increase and correlations decay exponentially in time and space.

We specifically choose to sweep $\beta$ for a fixed $\alpha = -0.8$, where the real and imaginary $\tilde{A}$ components are randomly initialized as Gaussian random variables scaled by a small parameter $\varepsilon$. The scanned regime has a clean, first order-like transition from frozen stable/metastable states of spirals to defect turbulence that is visibly chaotic with little discernible order~\cite{Huber1992,Chate1996}. Notably, we scan a regime of the 2D CGL phase diagram that avoids persistent phase turbulence, a chaotic regime with no defects. This choice is to reduce the set of observed phases and analyze LENS performance for a single phase transition from visible stability to visible chaos. For additional discussion on the CGL phase diagram, we refer the reader to more detailed references~\cite{Chate1996,Hohenberg1992,Aranson2002}.

\subsubsection{Optimization Details}

For these experiments, we used the CNN architecture shown in Fig.~\ref{fig:cnn_archi}. The architecture was chosen to ensure that salient features of the dynamics could be inferred in a translation-invariant manner without constructing a model with an exceedingly large number of parameters.

In a similar fashion to the other experiments, the early stopping threshold is also set to 7 epochs. We observed that the estimated EPR could greatly increase within these 7 epochs between the minimal loss and early stopping trigger, but the detected onset of phase turbulence and general trend in estimated EPR is not affected by tuning the early stopping threshold.


\section{The LENS Objective} \label{app:lens_objective}

\subsection{Decomposition of Stochastic Dynamics} \label{subsec:decomposition_stochastic}

Consider a stationary Markovian system with transition probability density $p(\va{x}\rightarrow \va{x}')$, which represents the probability that the system exhibits a transition from state $\va{x}$ to a new state $\va{x}'$ at consecutive time steps $t$ and $t+1$, respectively. This joint probability can be equivalently written as $p(s_t=\va{x},s_{t+1}=\va{x}')$ and expanded as
\begin{equation}
    p(\va{x}\rightarrow \va{x}')=p_\rightarrow(s_{t+1}=\va{x}'|s_t=\va{x})p(s_t=\va{x})
\end{equation}
where the forward dynamics $p_\rightarrow(s_{t+1}|s_t)$ is the transition probability in a single forward time step. The global entropy production rate (EPR) is then defined as the Kullback-Leibler (KL) divergence between the forward dynamics $p(\va{x}\rightarrow \va{x}')$ and the reverse dynamics $p(\va{x}'\rightarrow \va{x})$~\cite{Kawai2007}. Explicitly, this KL divergence is given by
\begin{align}
    \dot{S}&=D_{KL}\big[p(\va{x}\rightarrow \va{x}')\| p(\va{x}'\rightarrow \va{x})\big] \label{eq:epr_kl_div_definition} \\
    &=\int\dd \va{x}\,\dd \va{x}'\, p(\va{x}\rightarrow \va{x}')\log\qty[\frac{p(\va{x}\rightarrow \va{x}')}{p(\va{x}'\rightarrow \va{x})}] \\
    &=\mathbb{E}_{p(\va{x}\rightarrow \va{x}')}\log\qty[\frac{p(\va{x}\rightarrow \va{x}')}{p(\va{x}' \rightarrow \va{x})}]. \label{eq:epr_expval_definition_app}
\end{align}
The object of interest here is the log joint probability ratio in Eq.~\ref{eq:epr_expval_definition_app}, which we call the target logit, as it will later form the learned logit of the classifier which enables the computation of the EPR. For a stochastic dynamical system in a steady state, we can show that the target logit is a function of the system's states $\va{x},\va{x}'$ at successive time points, the diffusion tensor $\vb{D}$, the stationary distribution $\rho_{ss}$, and the local velocity field $\va{u}(\va{x})$ in state space. Apart from the observed states $\va{x}$ and $\va{x}'$, the remaining objects are defined entirely by the dynamical equations of motion. Thus, learning this target logit from pairs of observations enables a statistical estimation of the EPR using Eq.~\ref{eq:epr_expval_definition_app} (App.~\ref{subsec:construction_objective_function}). Furthermore, particular systems admit tractable solutions for the stationary distribution and velocity field, enabling an analytical baseline for the EPR.

Suppose the underlying dynamics can be written using the Langevin equation
\begin{equation}
    \label{eq:general_langevin_start}
    \dd \va{x}=f(\va{x},t)\,\dd t+\vb{F}\,\dd\va{W}_t
\end{equation}
where $f(\va{x},t)$ is a spatiotemporal function which generates the deterministic part of the dynamics, $\dd \va{W}_t$ is the multivariate Wiener increment at time $t$ and $\vb{F}$ controls the covariance of the fluctuations. Assuming a (generally nonequilibrium) steady state $\rho_{ss}(\va{x})$, we drop the time dependencies. The probability density $\rho_{ss}(\va{x})$ then obeys the Fokker-Planck equation
\begin{align}
    \pdv{t}\rho_{ss}(\va{x})&=-\div{\va{J}_{ss}(\va{x})}=0 \\
    \text{where }\va{J}_{ss}(\va{x})&=f(\va{x})\rho_{ss}(\va{x})-\vb{D}\grad\rho_{ss}(\va{x}). \label{eq:general_probability_current}
\end{align}
where $\va{J}_{ss}(\va{x})$ is the steady state probability current and $\vb{D}=\vb{F}\vb{F}^\intercal/2$ is the diffusion tensor. Then, for processes of the form in Eq.~\ref{eq:general_langevin_start}, the target logit can be written in terms of the displacements $\va{\delta}=\va{x}'-\va{x}$, the diffusion tensor $\vb{D}$ and the steady state velocity field $\va{u}(\va{x})=\va{J}_{ss}(\va{x})/\rho_{ss}(\va{x})$:
\begin{equation}
    \label{eq:target_logit_general_linearized_form}
    S(\va{x},\va{x}')=\va{\delta}^\intercal\vb{D}^{-1}\va{u}(\va{x})+ \frac{1}{2}\va{\delta}^\intercal\vb{D}^{-1}\qty(\grad\va{u}(\va{x}))\va{\delta} +\mathcal{O}(\|\va{\delta}\|^3)
\end{equation}
To begin, we solve for the deterministic function to obtain
\begin{equation}
    \label{eq:ness_deterministic_function}
    f(\va{x})=\vb{D}\grad\log\rho_{ss}(\va{x})+ \frac{\va{J}(\va{x})}{\rho_{ss}(\va{x})}
\end{equation}
where the first term represents the gradient flow of the probability density and the second term is the velocity field in state space. A Helmholtz decomposition on the second term leads to
\begin{equation}
    \va{J}(\va{x})=\va{J}_s(\va{x})+\grad\psi(\va{x})
\end{equation}
to yield a solenoidal term $\va{J}_\mathrm{s}(\va{x})$ (divergenceless), and the gradient of some scalar potential $\psi(\va{x})$. In a steady state, we can absorb this scalar potential into the gradient flow term, so we are left with the solenoidal current. In the present setting, the target logit is given by
\begin{align}
    S(\va{x},\va{x}')&= \log\qty[\frac{p(\va{x}\rightarrow\va{x}')}{p(\va{x}'\rightarrow\va{x})}] \nonumber \\
    &=\log\qty[\frac{p_\rightarrow(\va{x}'|\va{x})}{p_\rightarrow(\va{x}|\va{x}')}]-\log\qty[\frac{\rho_{ss}(\va{x}')}{\rho_{ss}(\va{x})}].
\end{align}
We now define $\va{\delta}=\va{x}'-\va{x}\sim\mathcal{O}(\sqrt{\dd t})$, and expand the logit terms to $\mathcal{O}(\|\va{\delta}\|^2)$ to retain all relevant behavior up to $\mathcal{O}(\dd t)$. Starting with the log density ratio, we expand
\begin{multline}
    \label{eq:log_density_ratio_stationary}
    \log \rho_{ss}(\va{x}')=\log \rho_{ss}(\va{x})+\va{\delta}^\intercal\grad\log\rho_{ss}(\va{x}) \\
    +\frac{1}{2}\va{\delta}^\intercal\qty(\grad\grad^\intercal\log\rho_{ss}(\va{x}))\va{\delta}+\mathcal{O}(\|\va{\delta}\|^3)
\end{multline}
where $\grad\grad^\intercal\log\rho_{ss}$ is the Hessian of the log density evaluated at $\va{x}$. To deal with the log conditional probability ratio, we start by assuming a sufficiently short timestep $\dd t$ between states so the forward transition kernel is then
\begin{align}
    p_\rightarrow(\va{x}'|\va{x})&= \mathcal{N}\qty(\va{x}+f(\va{x})\,\dd t,2\vb{D}\,\dd t) \\
    \log p_\rightarrow(\va{x}'|\va{x})&= -\frac{1}{4\,\dd t}(\va{\delta}-f\,\dd t)^\intercal\vb{D}^{-1}(\va{\delta}-f\,\dd t)+\mathrm{const.}
\end{align}
where $\mathcal{N}(\va{\mu},\vb{\Sigma})$ is the multivariate normal distribution with mean $\va{\mu}$ and covariance matrix $\vb{\Sigma}$, and $f\equiv f(\va{x})$ for brevity. Now, we write the log conditional probability ratio
\begin{equation}
    \log\qty[\frac{p_\rightarrow(\va{x}'|\va{x})}{p_\rightarrow(\va{x}|\va{x}')}] = \frac{1}{2}\va{\delta}^\intercal\vb{D}^{-1}f+\frac{1}{2}\va{\delta}^\intercal\vb{D}^{-1}f'+\mathcal{O}(\dd t^{3/2})
\end{equation}
where $f'\equiv f(\va{x}')$ and terms like $f^\intercal\vb{D}^{-1}f$ are irrelevant since they are contributions of $\mathcal{O}(\dd t^{3/2})$ or higher order. Expanding the dynamics around $\va{x}$, we have
\begin{equation}
    \log\qty[\frac{p_\rightarrow(\va{x}'|\va{x})}{p_\rightarrow(\va{x}|\va{x}')}] = \va{\delta}^\intercal\vb{D}^{-1}f+\frac{1}{2}\va{\delta}^\intercal\vb{D}^{-1}\qty(\grad f)\va{\delta}+\mathcal{O}(\|\va{\delta}\|^3).
\end{equation}
We now use Eq.~\ref{eq:ness_deterministic_function} to obtain the log conditional probability ratio in terms of the stationary density and velocity field as
\begin{multline}
    \log\qty[\frac{p_\rightarrow(\va{x}'|\va{x})}{p_\rightarrow(\va{x}|\va{x}')}] = \va{\delta}^\intercal\grad\log\rho_{ss}(\va{x}) +\va{\delta}^\intercal\vb{D}^{-1}\va{u}(\va{x}) \\
    +\frac{1}{2}\va{\delta}^\intercal\grad\grad^\intercal\log\rho_{ss}(\va{x})\va{\delta}+\frac{1}{2}\va{\delta}^\intercal\vb{D}^{-1}\qty(\grad\va{u}(\va{x}))\va{\delta}
\end{multline}
where we write $\va{u}(\va{x})=\va{J}_s(\va{x})/\rho_{ss}(\va{x})$ for the velocity field. Combining this with Eq.~\ref{eq:log_density_ratio_stationary}, we finally obtain the expression~\eqref{eq:target_logit_general_linearized_form} for the target logit which is independent of the stationary density $\rho_{ss}(\va{x})$ and derivatives thereof. From the original definition of the target logit as a log probability ratio between the forward and reverse transitional distributions, we expect that $S(\va{x},\va{x}')$ should be antisymmetric under exchange of its arguments. This is easily confirmed in Eq.~\ref{eq:target_logit_general_linearized_form}, where we first note that the velocity field and its gradient are now evaluated at $\va{x}'$ after the exchange, so we expand the two as
\begin{align}
    \va{u}(\va{x}')&=\va{u}(\va{x})+\qty(\grad\va{u}(\va{x}))\va{\delta}+\mathcal{O}(\|\va{\delta}\|^2) \\
    \grad\va{u}(\va{x}')&=\grad\va{u}(\va{x})+\mathcal{O}(\|\va{\delta}\|).
\end{align}
The expressions are truncated at $\mathcal{O}(\|\va{\delta}\|^2)$ and $\mathcal{O}(\|\va{\delta}\|)$ respectively since they only appear in terms with corresponding factors of $\va{\delta}$ such that ignoring contributions of order $\mathcal{O}(\|\va{\delta}\|^3)$ allows us to consider lower order expansions of the velocity field and its gradient. Then, we explicitly evaluate
\begin{align}
    S(\va{x}',\va{x})&=\va{\delta}^\intercal\vb{D}^{-1}\va{u}(\va{x}')+ \frac{1}{2}\va{\delta}^\intercal\vb{D}^{-1}\qty(\grad\va{u}(\va{x}'))\va{\delta} +\mathcal{O}(\|\va{\delta}\|^3) \nonumber \\
    &=-\va{\delta}^\intercal\vb{D}^{-1}\va{u}(\va{x})-\va{\delta}^\intercal\vb{D}^{-1}\qty(\grad\va{u}(\va{x}))\va{\delta} \nonumber \\
    &\qquad+\frac{1}{2}\va{\delta}^\intercal\vb{D}^{-1}\qty(\grad\va{u}(\va{x}))\va{\delta}+\mathcal{O}(\|\va{\delta}\|^3) \nonumber \\
    &=-S(\va{x},\va{x}')+\mathcal{O}(\|\va{\delta}\|^3)
\end{align}
where we have used the fact that $\va{\delta}\rightarrow-\va{\delta}$ under the exchange. This confirms the antisymmetry of the target logit under short-time dynamics.

\subsubsection{Linear Stochastic Dynamics} \label{subsubsec:linear_stochastic_dynamics}

In the special case of linear stochastic dynamics, we may proceed significantly further. Suppose the dynamics are given by
\begin{equation}
    \dd \va{x}=\vb{G}\va{x}\,\dd t+\vb{F}\,\dd\va{W}_t
\end{equation}
where $\vb{G}$ is a matrix whose eigenvalues have negative real part, enforcing stability of the dynamics. This admits a nonequilibrium steady state
\begin{equation}
    \rho_{ss}(\va{x})=\frac{1}{\sqrt{\qty(2\pi)^N\qty|\det\vb{C}|}}\exp\qty(-\frac{1}{2}\va{x}^\intercal\vb{C}^{-1}\va{x})
\end{equation}
where $\vb{C}$ is the steady state covariance matrix obtained as the solution to the continuous-time Lyapunov equation $\vb{G}\vb{C}+\vb{C}\vb{G}^\intercal=-2\vb{D}$. The corresponding velocity field is then obtained from Eq.~\ref{eq:general_probability_current} as
\begin{equation}
    \va{u}(\va{x})=\qty(\vb{G}+\vb{D}\vb{C}^{-1})\va{x}.
\end{equation}
Noting that $\grad\va{u}(\va{x})=\vb{G}+\vb{D}\vb{C}^{-1}$, we proceed to write the target logit as
\begin{align}
    S(\va{x},\va{x}')&=(\va{x}'-\va{x})^\intercal\vb{D}^{-1}(\vb{G}+\vb{D}\vb{C}^{-1})\va{x} \nonumber \\
    &\qquad+\frac{1}{2}(\va{x}'-\va{x})^\intercal\vb{D}^{-1}(\vb{G}+\vb{D}\vb{C}^{-1})(\va{x}'-\va{x}) \nonumber \\
    &=\frac{1}{2}\va{x}'^\intercal\vb{\Lambda}\va{x}-\frac{1}{2}\va{x}^\intercal\vb{\Lambda}\va{x}'+\frac{1}{2}\va{x}'^\intercal\vb{\Lambda}\va{x}'-\frac{1}{2}\va{x}^\intercal\vb{\Lambda}\va{x}
\end{align}
where we define $\vb{\Lambda}=\vb{D}^{-1}\vb{G}+\vb{C}^{-1}$. Decomposing $\vb{\Lambda}$ into its symmetric and skew-symmetric parts
\begin{equation}
    \vb{\Lambda}=\underbrace{\frac{1}{2}(\vb{\Lambda}+\vb{\Lambda}^\intercal)}_{\vb{\Lambda}_s}+\underbrace{\frac{1}{2}(\vb{\Lambda}-\vb{\Lambda}^\intercal)}_{\vb{\Lambda}_a},
\end{equation}
we note that $\vb{\Lambda}_a$ vanishes within all quadratic self-terms of the form $\va{x}^\intercal\vb{\Lambda}\va{x}$, so we can write
\begin{equation}
    \label{eq:target_logit_linear_form}
    S(\va{x},\va{x}')=\va{x}'^\intercal\vb{\Lambda}_a\va{x}+\frac{1}{2}\qty(\va{x}'^\intercal\vb{\Lambda}_s\va{x}'-\va{x}^\intercal\vb{\Lambda}_s\va{x}).
\end{equation}
Finally, we insert the explicit expressions for $\vb{\Lambda}_s$ and $\vb{\Lambda}_a$ to obtain
\begin{align}
    S(\va{x},\va{x}')&=\frac{1}{2}\va{x}'^\intercal\qty(\vb{D}^{-1}\vb{G}-\vb{G}^\intercal\vb{D}^{-1})\va{x} \nonumber \\
    &\qquad + \frac{1}{4}\va{x}'^\intercal\qty(\vb{D}^{-1}\vb{G}+\vb{G}^\intercal\vb{D}^{-1}+2\vb{C}^{-1})\va{x}' \nonumber \\
    &\qquad - \frac{1}{4}\va{x}^\intercal\qty(\vb{D}^{-1}\vb{G}+\vb{G}^\intercal\vb{D}^{-1}+2\vb{C}^{-1})\va{x}
\end{align}

\subsubsection{General Nonlinear Stochastic Dynamics} \label{subsubsec:nonlinear_stochastic_dynamics}

Motivated by the form of Eq.~\ref{eq:target_logit_linear_form} which contains an antisymmetric quadratic form and the difference of symmetric quadratic forms resembling a potential difference, we now return to the case where $f(\va{x})$ is a smooth nonlinear function. We start from Eq.~\ref{eq:target_logit_general_linearized_form} and evaluate the velocity field and its gradient at the midpoint $\va{y}=(\va{x}+\va{x}')/2$ of the two points
\begin{equation}
    S(\va{x},\va{x}')=\va{\delta}^\intercal\vb{D}^{-1}\va{u}(\va{y})+ \frac{1}{2}\va{\delta}^\intercal\vb{D}^{-1}\qty(\grad\va{u}(\va{y}))\va{\delta} +\mathcal{O}(\|\va{\delta}\|^3)
\end{equation}
which only differs from the existing logit at $\mathcal{O}(\|\va{\delta}\|^3)$. We then define
\begin{equation}
    \va{p}(\va{y})=\vb{D}^{-1}\va{u}(\va{y}),\quad \vb{Q}(\va{y})=\frac{1}{2}\vb{D}^{-1}(\grad\va{u}(\va{y}))
\end{equation}
and rewrite the second term as
\begin{equation}
    \frac{1}{2}\va{\delta}^\intercal\vb{D}^{-1}\qty(\grad\va{u}(\va{y}))\va{\delta}=-\va{x}'^\intercal\vb{Q}_a\va{x}+\frac{1}{2}(\va{x}'^\intercal\vb{Q}_s\va{x}'-\va{x}^\intercal\vb{Q}_s\va{x})
\end{equation}
where $\vb{Q}_s$ and $\vb{Q}_a$ are the symmetric and skew-symmetric parts of $\vb{Q}(\va{y})$ respectively. The first term can be rewritten as (dropping the $\va{y}$ dependence for brevity)
\begin{align}
    \va{\delta}^\intercal\vb{D}^{-1}\va{u}(\va{y}) &= \va{x}'^\intercal\va{p}-\va{x}^\intercal\va{p} \nonumber \\
    &= \frac{1}{2}\qty[V_p(\va{x}')-V_p(\va{x})], \nonumber \\
    \text{where } V_p(\va{z})&=(\va{z}+\va{p})^\intercal(\va{z}+\va{p})-\va{z}^\intercal\va{z}.
\end{align}
Combining the two terms, we then obtain
\begin{multline}
    S(\va{x},\va{x}')=\va{x}'^\intercal[-\vb{Q}_a(\va{y})]\va{x} \\
    + \underbrace{\frac{1}{2}\qty[V_p(\va{x}')+\va{x}'^\intercal\vb{Q}_s(\va{y})\va{x}']}_{V(\va{x}')} - \underbrace{\frac{1}{2}\qty[V_p(\va{x})+\va{x}^\intercal\vb{Q}_s(\va{y})\va{x}]}_{V(\va{x})}
\end{multline}
which parametrizes the target logit using an antisymmetric quadratic form and the difference of two potential-like terms as desired.

\subsection{Construction of the Objective Function} \label{subsec:construction_objective_function}

The estimator is framed as a binary classifier $C_\theta(\va{x}, \va{x}')=\sigma\qty(S_\theta(\va{x},\va{x}'))$ where $S_\theta(\va{x},\va{x}')$ is the corresponding logit parametrized by $\theta$, and $\sigma(z)=1/(1+e^{-z})$ is the logistic sigmoid function. This classifier is trained to receive an adjacent pair of states $(\va{x},\va{x}')$ and return a prediction probability $C_\theta(\va{x},\va{x}')\in\qty(0,1)$ that the selected pair of states $(\va{x},\va{x}')$ originates from the forward dynamics $(\va{x}\rightarrow \va{x}')$ as opposed to the reverse dynamics $(\va{x}'\rightarrow \va{x})$. The classifier is trained by maximizing the LENS objective function (w.r.t $\theta$)
\begin{align}
    J(\theta)=\;&\mathbb{E}_{p(\va{x}\rightarrow \va{x}')}\qty[\log \sigma\qty(S_\theta(\va{x},\va{x}'))] \nonumber \\
    &\quad +\mathbb{E}_{p(\va{x}'\rightarrow \va{x})}\qty[\log\qty(1-\sigma\qty(S_\theta(\va{x},\va{x}')))]\label{eq:lens_objective_function_app}
\end{align}
with the mathematical expectations computed over consecutive-time state pairs sampled from a provided time series $\{\va{x}_t\}_{t=1}^T$. Maximizing the LENS objective is equivalent to minimizing the binary cross-entropy loss between the actual forward/reverse labels and the prediction probability $C_\theta(\va{x},\va{x}')$. In the limit of infinite sampling, the classifier converges to the Bayes-optimal solution obtained by taking a functional derivative of $J(\theta)$ with respect to $S_\theta$ and equating the derivative to zero. The Bayes-optimal solution then gives the classification probability $C^*_\theta$ as the posterior probability that the sample pair $(\va{x},\va{x}')$ originates from the forward dynamics as
\begin{align}
    C^*_\theta(\va{x},\va{x}')&=\frac{p(\va{x}\rightarrow \va{x}')}{p(\va{x}\rightarrow \va{x}')+p(\va{x}'\rightarrow \va{x})} \\
    \Longrightarrow\quad S^*_\theta(\va{x},\va{x}')&=\log\qty[\frac{p(\va{x}\rightarrow \va{x}')}{p(\va{x}'\rightarrow \va{x})}]. \label{eq:lens_bayes_optimal_logit}
\end{align}
The logit thus converges to the log probability ratio between the distributions of forward and reverse transitions and enables the direct estimation of the local and global EPRs by taking an expectation value of $S_{\theta}(\va{x},\va{x}')$ over many pairs of consecutive-time states. The Bayes-optimal form of $S_\theta(\va{x},\va{x}')$ also requires that it is antisymmetric in its arguments, which can be expressed at the lowest nontrivial order by an antisymmetric quadratic form and the difference of two symmetric quadratic forms, in line with the reasoning from the above section. This leads to the core parametrization of the LENS logit
\begin{align}
    S_\theta(\va{x},\va{x}') &=\va{\phi}_\theta(\va{x}')^\intercal\vb{A}_\theta\va{\phi}_\theta(\va{x}) +\frac{1}{2}\va{\phi}_\theta(\va{x}')^\intercal \vb{B}_\theta\va{\phi}_\theta(\va{x}') \nonumber \\
    &\qquad -\frac{1}{2}\va{\phi}_\theta(\va{x})^\intercal \vb{B}_\theta\va{\phi}_\theta(\va{x}) \label{eq:lens_logit_parametrization}
\end{align}
where we introduce the learned representation $\va{\phi}_\theta\qty(\va{x})$ of the system state $\va{x}$ which is structured as a feedforward network. The learned matrix $\vb{A}_\theta$ is constrained to be in block diagonal form (which is always attainable for a real skew-symmetric matrix through a Schur decomposition), while $\vb{B}_\theta$ is constrained to be symmetric (see Sec~\ref{subsec:representation_universality} for details). 

For a linear network, we expect the learned representation to correspond to the projection of the system's dynamics onto the subspace relevant for irreversible dynamics. Increasing the network depth and adding nonlinearities allows the representations to be nonlinear functions of the system's dynamics.	

\subsection{Universality of Representations} \label{subsec:representation_universality}

The logit $S(\va{x},\va{x}')$ is constructed to capture the log ratio between forward and reverse joint transition probabilities, with suitable constraints imposed on the learned matrices $\vb{A}_\theta$ and $\vb{B}_\theta$. These constraints increase the specificity of the learned representations while retaining sufficient flexibility to capture general irreversible dynamics. Here, we prove the universality of these learned representations.

First, we show a relatively trivial result which justifies the form of the LENS logit. For the following proofs, we define $K\subset\mathbb{R}^N$ as a compact subset of $\mathbb{R}^N$.

\begin{lemma}
    \label{lemma:general_antisymmetric}
    Any function $S:K\times K \mapsto \mathbb{R}$ which is antisymmetric under exchange of its two arguments s.t. $S(\va{y},\va{x})=-S(\va{x},\va{y})$ can be written as:
    \begin{equation}
        S(\va{x},\va{y})=f(\va{x},\va{y})+g(\va{y})-g(\va{x}),
    \end{equation}
    where $f(\va{x},\va{y})$ is a similarly antisymmetric function.
\end{lemma}
\begin{proof}[Proof of Lemma~\ref{lemma:general_antisymmetric}]
    Pick any smooth bounded function $g:K\mapsto\mathbb{R}$, then write:
    \begin{equation}
        f(\va{x},\va{y})=S(\va{x},\va{y})-g(\va{y})+g(\va{x})
    \end{equation}
    which is easily shown to be antisymmetric under exchange of $\va{x}\longleftrightarrow \va{y}$.
\end{proof}

With the motivation of the LENS logit, we now state the central universality theorem.

\begin{theorem}
    \label{thm:universality}
    For any $\epsilon>0$, there exists a representation $\phi:K\mapsto\mathbb{R}^M$, skew-symmetric matrix $\vb{A}$ of size $M\times M$, and symmetric matrix $\vb{B}$ of size $M\times M$ for sufficiently large $M$ s.t.
    \begin{align}
        &\sup_{\va{x},\va{y}\in K}|S(\va{x},\va{y})-\hat{S}(\va{x},\va{y})|<\epsilon \\
        \text{\emph{where} }&\hat{S}(\va{x},\va{y})=\va{\phi}^\intercal(\va{y})\vb{A}\va{\phi}(\va{x}) \nonumber \\
        &\quad+\frac{1}{2}\va{\phi}^\intercal(\va{y})\vb{B}\va{\phi}(\va{y})-\frac{1}{2}\va{\phi}^\intercal(\va{x})\vb{B}\va{\phi}(\va{x}).
    \end{align}
\end{theorem}

The strategy of the proof will be to first use the Stone-Weierstrass theorem to show that $f(\va{x},\va{y})$ can be approximated by a sufficiently high degree polynomial in terms involving the components of $\va{x}$ and $\va{y}$. The polynomial will be constructed using a representation $\phi$ that is comprised of an appropriate monomial basis whose coefficients can be matricized to obtain a skew-symmetric matrix $\vb{A}$. With the representation at hand, we will then show that $g(\va{x})$ can be approximated by a sufficiently high degree polynomial in the components of $\va{x}$, and the coefficients can be similarly matricized to obtain a symmetric matrix $\vb{B}$. Finally, we assemble the approximation for $S(\va{x},\va{y})$ using the three approximation terms and ensure that the representation is large enough to approximate each of the three terms to the required accuracy.

We now proceed to obtain the representation which ensures the antisymmetric function $f(\va{x},\va{y})$ is suitably approximated.

\begin{lemma}
    \label{lemma:antisymmetric_term}
    Let $f(\va{x},\va{y}):K\times K\mapsto\mathbb{R}$ be an antisymmetric function. For any $\epsilon>0$, there exists a suitable function $\phi:K\mapsto\mathbb{R}^M$ and $M\times M$ skew-symmetric matrix $\vb{A}$ for sufficiently large $M$ s.t.
    \begin{equation}
        \sup_{\va{x},\va{y}\in K}|f(\va{x},\va{y})-\va{\phi}^\intercal(\va{y})\vb{A}\va{\phi}(\va{x})|<\frac{\epsilon}{3}.
    \end{equation}
\end{lemma}

\begin{proof}[Proof of Lemma~\ref{lemma:antisymmetric_term}]
    First fix $\epsilon>0$, then consider $N=1$ and choose the monomial basis representation
    \begin{equation}
        \phi(\va{x})=\mqty(1 & x & x^2 & \cdots & x^{d_a})^\intercal
    \end{equation}
    so $M=d_a+1$. By the Stone-Weierstrass theorem, there exists a $d_a$-degree polynomial for sufficiently high $d_a$ that approximates $f(\va{x},\va{y})$ as
    \begin{equation}
        \sup_{\va{x},\va{y}\in K}\qty|f(\va{x},\va{y})-\sum_{i,j=0}^{d_a}a_{ij}x^iy^j|<\frac{\epsilon}{3}.
    \end{equation}
    We then form the matrix $\vb{A}$ whose $(i,j)$-th element is given by $a_{ij}$, which enables us to write the above sum as a quadratic form $\va{\phi}^\intercal(\va{y})\vb{A}\va{\phi}(\va{x})$. Now, the antisymmetry of $f(\va{x},\va{y})$ implies that
    \begin{equation}
        \sum_{i,j=0}^{d_a}a_{ij}y^ix^j=-\sum_{i,j=0}^{d_a}a_{ij}x^iy^j
    \end{equation}
    which leads to the result $a_{ji}=-a_{ij}$. This implies that $\vb{A}^\intercal=\vb{A}$, so $\vb{A}$ is skew-symmetric as required.
    
    To extend the Lemma to $N>1$, we now consider $\va{x}=(x_1,x_2,\cdots,x_N)$ and $\va{y}=(y_1,y_2,\cdots,y_N)$. Define the multi-index $\va{\alpha}=(\alpha_1,\alpha_2,\cdots,\alpha_N)\in\mathbb{N}^N$, and $\va{x}^{\va{\alpha}}\equiv\prod_{k=1}^N x_k^{\alpha_k}$ with total degree $|\va{\alpha}|=\alpha_1+\alpha_2+\cdots+\alpha_N$. This gives us a representation
    \begin{equation}
        \va{\phi}(\va{x})=\qty(\va{x}^{\va{\alpha}})^\intercal_{|\va{\alpha}|\leq d_a}\in\mathbb{R}^{M(d_a)},\quad M(d_a)=\mqty(N+d_a \\ d_a)
    \end{equation}
    where the subscript indicates that we only allow multi-indices whose total degree is at most $d_a$. This returns a column vector of dimensionality $M(d_a)$ whose entries are all possible monomials comprising powers of $x_1,\cdots,x_N$ such that the total power is at most $d_a$. The ordering of monomials is arbitrary, but can be fixed by imposing a suitable well-ordering such as graded reverse lexicographic order so each multi-index $\va{\alpha}$ has a corresponding integer index $i(\va{\alpha})\in\{0,1,\cdots,M(d_a)-1\}$ that records its position in the entries of $\va{\phi}$.
    We then apply the Stone-Weierstrass theorem to $f$ as
    \begin{multline}
        \sup_{\va{x},\va{y}\in K}|f(\va{x},\va{y})-P_a(\va{x},\va{y})|<\frac{\epsilon}{3}, \\
        \quad\text{where }P_a(\va{x},\va{y})=\sum_{|\va{\alpha}|,|\va{\beta}|\leq d_a}c_{\va{\alpha}\va{\beta}}\va{x}^{\va{\alpha}}\va{y}^{\va{\beta}}.
    \end{multline}
    Next, we construct the matrix $\vb{C}\in\mathbb{R}^{M(d_a)\times M(d_a)}$ whose $(i(\va{\alpha}),i(\va{\beta)})$-th element is $c_{\va{\alpha}\va{\beta}}$. This allows the above sum to be expressible as a quadratic form
    \begin{align}
        P_a(\va{x},\va{y})&=\sum_{|\va{\alpha}|,|\va{\beta}|\leq d_a}c_{\va{\alpha}\va{\beta}}\va{x}^{\va{\alpha}}\va{y}^{\va{\beta}} \nonumber \\
        &=\va{\phi}^\intercal(\va{y})\vb{C}\va{\phi}(\va{x}).
    \end{align}
    Finally, we ensure the quadratic form is antisymmetric by defining $\vb{A}=\frac{1}{2}(\vb{C}-\vb{C}^\intercal)$ so
    \begin{equation}
        \va{\phi}^\intercal(\va{y})\vb{A}\va{\phi}(\va{x})=\frac{1}{2}\qty[P_a(\va{x},\va{y})-P_a(\va{y},\va{x})].
    \end{equation}
    Despite antisymmetrizing the matrix of coefficients, this does not affect the approximation error and we can easily show that this antisymmetric quadratic form achieves the desired level of accuracy as
    \begin{widetext}
    \begin{align}
        \sup_{\va{x},\va{y}\in K}\qty|f(\va{x},\va{y})-\va{\phi}^\intercal(\va{y})\vb{A}\va{\phi}(\va{x})| &=\sup_{\va{x},\va{y}\in K}\qty|f(\va{x},\va{y})-\frac{1}{2}[P_a(\va{x},\va{y})-P_a(\va{y},\va{x})]| \\
        &\leq \sup_{\va{x},\va{y}\in K}\qty[\frac{1}{2}\qty|f(\va{x},\va{y})-P_a(\va{x},\va{y})|+\frac{1}{2}\qty|-f(\va{y},\va{x})+P_a(\va{y},\va{x})|] \\
        &<\frac{\epsilon}{6}+\frac{\epsilon}{6}=\frac{\epsilon}{3}
    \end{align}
    \end{widetext}
    which completes the proof.
\end{proof}

Now, we show that the symmetric function $g(\va{x})$ can be similarly approximated using the form of the representations obtained from Lemma~\ref{lemma:antisymmetric_term}.

\begin{lemma}
    \label{lemma:symmetric_term}
    Let $g(\va{x}):K\mapsto\mathbb{R}$ be a smooth function. For any $\epsilon>0$, there exists a $M\times M$ symmetric matrix $\vb{B}$ and suitable function $\phi:K\mapsto \mathbb{R}^M$ for sufficiently large $M$ s.t.
    \begin{equation}
        \sup_{\va{x}\in K}\qty|g(\va{x})-\frac{1}{2}\va{\phi}^\intercal(\va{x})\vb{B}\va{\phi}(\va{x})|<\frac{\epsilon}{3},
    \end{equation}
    where the representation function $\phi$ is of the form
    \begin{equation}
        \va{\phi}(\va{x})=\qty(\va{x}^{\va{\alpha}})^\intercal_{|\va{\alpha}|\leq d_x}\in\mathbb{R}^{M(d_x)},\quad M(d_x)=\mqty(N+d_x \\ d_x)
    \end{equation}
\end{lemma}

\begin{proof}[Proof of Lemma~\ref{lemma:symmetric_term}]
    First, fix $\epsilon>0$ and use the Stone-Weierstrass theorem to state that there exists a $d_x$-degree polynomial for sufficiently high $d_x$ that approximates $g(\va{x})$ as
    \begin{multline}
        \label{eq:symmetric_lemma_polynomial_approx}
        \sup_{\va{x}\in K}\qty|g(\va{x})-P_x(\va{x})|<\frac{\epsilon}{3}, \\
        \text{where }P_x(\va{x})=\sum_{|\va{\alpha}|\leq d_x}c_{\va{\alpha}}\va{x}^{\va{\alpha}}=\sum_{p=0}^{M(d_x)-1}c_{\va{\alpha}^{(p)}}\va{x}^{\va{\alpha}^{(p)}}.
    \end{multline}
    Then, we can use a representation of the form $\va{\phi}(\va{x})=(\va{x}^{\va{\alpha}^{(0)}},\va{x}^{\va{\alpha}^{(1)}},\cdots, \va{x}^{\va{\alpha}^{(M(d_x)-1)}})^\intercal$ to write the quadratic form
    \begin{equation}
        \frac{1}{2}\va{\phi}^\intercal(\va{x})\vb{B}\va{\phi}(\va{x})=\frac{1}{2}\sum_{p,q=0}^{M(d_x)-1}B_{pq} \va{x}^{\va{\alpha}^{(p)}+\va{\alpha}^{(q)}}
    \end{equation}
    where $\vb{B}\in\mathbb{R}^{M(d_x)\times M(d_x)}$ is a symmetric matrix, and each index is understood to be matched with the multi-index ordering from before through $p=i(\va{\alpha}^{(p)})$. Each combined exponent $\va{\gamma}=\va{\alpha}^{(p)}+\va{\alpha}^{(q)}$ is then a multi-index with total degree $|\va{\gamma}|\leq 2d_x$. For each of these combined multi-indices, we define the count $n_{\va{\gamma}}$ as the number of ordered pairs $(p,q)$ whose corresponding multi-indices add to $\va{\gamma}$, denoted by
    \begin{equation}
        n_{\va{\gamma}}=\#\qty{(i,j):\, \va{\alpha}^{(i)}+\va{\alpha}^{(j)}=\va{\gamma}\in\mathbb{N}^N}.
    \end{equation}
    We thus construct the matrix $\vb{B}$ by assigning to the $(p,q)$-th element a uniform fraction of the required polynomial coefficient as
    \begin{equation}
        \frac{1}{2}B_{pq}=\frac{c_{\va{\alpha}^{(p)}+\va{\alpha}^{(q)}}}{n_{\va{\alpha}^{(p)}+\va{\alpha}^{(q)}}}=\frac{c_{\va{\gamma}}}{n_{\va{\gamma}}},
    \end{equation}
    where the symmetry is clearly satisfied since swapping $p$ and $q$ does not affect the combined multi-index $\va{\gamma}$. Now, we proceed to expand the quadratic form as
    \begin{align}
        \frac{1}{2}\va{\phi}^\intercal(\va{x})\vb{B}\va{\phi}(\va{x}) &=\frac{1}{2}\sum_{p,q=0}^{M(d_x)-1}B_{pq} \va{x}^{\va{\alpha}^{(p)}+\va{\alpha}^{(q)}} \\
        &=\sum_{p,q}\frac{c_{\va{\alpha}^{(p)}+\va{\alpha}^{(q)}}}{n_{\va{\alpha}^{(p)}+\va{\alpha}^{(q)}}} \va{x}^{\va{\alpha}^{(p)}+\va{\alpha}^{(q)}}.
    \end{align}
    Grouping together terms whose exponent vector is the same multi-index $\va{\gamma}$, we see that for each $\va{\gamma}$ there are exactly $n_{\va{\gamma}}$ pairs $(p,q)$ which meet the condition $\va{\alpha}^{(p)}+\va{\alpha}^{(q)}=\va{\gamma}$. Thus, the total contribution of the $\va{x}^{\va{\gamma}}$ term is
    \begin{equation}
        \sum_{\substack{p,q \\ \va{\alpha}^{(p)}+\va{\alpha}^{(q)}=\va{\gamma}}}\frac{c_{\va{\gamma}}}{n_{\va{\gamma}}}=n_{\va{\gamma}}\frac{c_{\va{\gamma}}}{n_{\va{\gamma}}}=c_{\va{\gamma}}.
    \end{equation}
    Finally, we impose the condition that $c_{\va{\gamma}}=0$ for all $|\va{\gamma}|>d_x$, which simply demands that combined multi-indices only contribute if they have a total degree of up to $d_x$. This allows us to match the polynomials as
    \begin{equation}
        \frac{1}{2}\va{\phi}^\intercal(\va{x})\vb{B}\va{\phi}(\va{x})= \sum_{|\va{\gamma}|\leq 2d_x}c_{\va{\gamma}}\va{x}^{\va{\gamma}}= \sum_{|\va{\alpha}|\leq d_x}c_{\va{\alpha}}\va{x}^{\va{\alpha}}= P_x(\va{x})
    \end{equation}
    which completes the proof by substitution into Eq.~\ref{eq:symmetric_lemma_polynomial_approx}.
\end{proof}

It should be noted that Lemma~\ref{lemma:symmetric_term} also applies to the approximation of $g(y)$ using the quadratic form $\va{\phi}^\intercal(\va{y})\vb{B}\va{\phi}(\va{y})$, albeit with a different polynomial degree $d_y$ chosen such that the approximation lies within the $\epsilon/3$ error bound.

\begin{proof}[Proof of Theorem~\ref{thm:universality}]
    Recall the relevant definitions:
    \begin{align}
        S(\va{x},\va{y})&=f(\va{x},\va{y})+g(\va{y})-g(\va{x}) \\
        \hat{S}(\va{x},\va{y})&=\va{\phi}^\intercal(\va{y})\vb{A}\va{\phi}(\va{x})+\frac{1}{2}\va{\phi}^\intercal(\va{y})\vb{B}\va{\phi}(\va{y})-\frac{1}{2}\va{\phi}^\intercal(\va{x})\vb{B}\va{\phi}(\va{x})
    \end{align}
    Fix $\epsilon>0$, then we write
    \begin{widetext}
    \begin{align}
        \sup_{\va{x},\va{y}\in K}|S(\va{x},\va{y})-\hat{S}(\va{x},\va{y})| &\leq\sup_{\va{x},\va{y}\in K}\qty|f(\va{x},\va{y})-\va{\phi}^\intercal(\va{y})\vb{A}\va{\phi}(\va{x})| \nonumber \\
        &\qquad\qquad+\sup_{\va{y}\in K}\qty|g(\va{y})-\frac{1}{2}\va{\phi}^\intercal(\va{y})\vb{B}\va{\phi}(\va{y})| +\sup_{\va{x}\in K}\qty|g(\va{x})-\frac{1}{2}\va{\phi}^\intercal(\va{x})\vb{B}\va{\phi}(\va{x})|
    \end{align}
    \end{widetext}
    using the triangle inequality. Thus, we simply need to show that all three terms can be simultaneously approximated to within an error bound of $\epsilon/3$ with the same representation $\phi:K\mapsto\mathbb{R}^M$. We apply Lemma~\ref{lemma:antisymmetric_term} to the first term using a $d_a$-degree polynomial, and apply Lemma~\ref{lemma:symmetric_term} to the second and third terms using $d_y$ and $d_x$-degree polynomials respectively. Then, we pick the representation dimension to be $\max(M(d_a),M(d_x),M(d_y))$. This can be achieved by padding the remaining terms' matrices with the appropriate number of zero rows and columns, which leaves the approximation unchanged. Crucially, this ensures that each term is simultaneously approximated to within the required error bound, which leads to
    \begin{equation}
        \sup_{\va{x},\va{y}\in K}|S(\va{x},\va{y})-\hat{S}(\va{x},\va{y})| < \frac{\epsilon}{3}+\frac{\epsilon}{3}+\frac{\epsilon}{3}=\epsilon
    \end{equation}
    which completes the proof.
\end{proof}

While the universality argument presented above guarantees that any short-time irreversible dynamics can be captured by a suitably parametrized logit $S(\va{x},\va{y})$, it places no further constraint on the learned representations except that $\vb{A}$ is skew-symmetric and $\vb{B}$ is symmetric. It is advantageous to impose constraints on these matrices to aid the interpretability of the resulting learned logit, so we will prove a simple corollary here which extends the universality result in the presence of these constraints.

\begin{corollary}
    \label{coro:block_diag_antisymmetric_structure}
    For any $\epsilon>0$, there exists a representation $\tilde{\phi}:K\mapsto\mathbb{R}^M$, and symmetric matrix $\vb{B}$ of size $M\times M$ for sufficiently large $M$ s.t. Theorem~\ref{thm:universality} holds for skew-symmetric $\vb{\Sigma}$ of size $M\times M$ with structure
    \begin{equation}
        \label{eq:block_diagonal_structure_thm}
        \vb{\Sigma}=\mqty(\dmat{\mqty{0 & -1 \\ 1 & 0}, \mqty{0 & -1 \\ 1 & 0}, \ddots}).
    \end{equation}
\end{corollary}

\begin{proof}[Proof of Corollary~\ref{coro:block_diag_antisymmetric_structure}]
    Suppose we have already obtained a representation function $\phi:K\mapsto\mathbb{R}^M$, a skew-symmetric matrix $\vb{A}$ and symmetric matrix $\vb{B}$ from approximating $S(\va{x},\va{y})$ according to Theorem~\ref{thm:universality}. We first recognize that every real skew-symmetric matrix admits a real orthogonal Schur decomposition of the form $\vb{A}=\vb{Q}^\intercal\tilde{\vb{\Sigma}}\vb{Q}$, where $\vb{Q}$ is an orthogonal matrix and $\tilde{\vb{\Sigma}}$ is of the form
    \begin{equation}
        \tilde{\vb{\Sigma}}=\mqty(\dmat{\mqty{0 & -\lambda_1 \\ \lambda_1 & 0}, \mqty{0 & -\lambda_2 \\ \lambda_2 & 0}, \ddots}),
    \end{equation}
    where each block entry $\lambda_i\in\mathbb{R}_0^+$. Define the diagonal matrix $\vb{\Gamma}=\mathrm{diag}(\sqrt{\lambda_1},\sqrt{\lambda_1},\sqrt{\lambda_2},\sqrt{\lambda_2},\cdots)$. Then, we can decompose $\vb{A}$ as
    \begin{align}
        \vb{A} &=(\vb{\Gamma}\vb{Q})^\intercal \vb{\Sigma}(\vb{\Gamma}\vb{Q}),
    \end{align}
    where $\vb{\Sigma}$ takes the form defined by Eq.~\ref{eq:block_diagonal_structure_thm}. We thus define modified representations $\tilde{\va{\phi}}(\va{x})=\vb{\Gamma}\vb{Q}\va{\phi}(\va{x})$, and write the antisymmetric quadratic form as
    \begin{equation}
        \va{\phi}^\intercal(\va{y})\vb{A}\va{\phi}(\va{x})= \tilde{\va{\phi}}^\intercal(\va{y})\vb{\Sigma}\tilde{\va{\phi}}(\va{x})
    \end{equation}
    which leaves the approximation error of the antisymmetric term unchanged. In the same fashion, we can use the modified representation to rewrite the symmetric quadratic forms as
    \begin{align}
        \va{\phi}^\intercal(\va{x})\vb{B}\va{\phi}(\va{x}) &=\tilde{\va{\phi}}^\intercal(\va{x})\qty[(\vb{\Gamma}\vb{Q})^{-1}]^\intercal\vb{B}(\vb{\Gamma}\vb{Q})^{-1}\tilde{\va{\phi}}(\va{x}) \\
        &=\tilde{\va{\phi}}^\intercal(\va{x})\tilde{\vb{B}}\tilde{\va{\phi}}(\va{x})
    \end{align}
    and similarly for the term involving $\va{y}$. It is easy to see that $\tilde{\vb{B}}$ is also symmetric so the approximation error is unaffected in either of the symmetric terms, which completes the proof.
\end{proof}

\section{\texorpdfstring{Entropy Production Rate of the $N$-Bead System}{Analytical Model for N-Bead System}} \label{app:n_bead_analytic_model}

While Eq.~\ref{eq:n_beads_langevin_app} describes the evolution of each independent realization of the stochastic dynamics, it is advantageous to instead study the corresponding Fokker-Planck equation
\begin{equation}
    \pdv{\rho(\va{x},t)}{t}=-\div\qty[\vb{G}\va{x}\rho(\va{x},t)-\vb{D}\grad\rho(\va{x},t)]=-\div\va{J}(\va{x},t)
\end{equation}
where $\rho(\va{x},t)$ is the probability density of the state space with coordinates $\va{x}$ at time $t$, $\va{J}(\va{x},t)$ is the corresponding probability current, and $\vb{D}=\vb{F}\vb{F}^\intercal/2$ is the diffusion tensor. In a nonequilibrium steady state (NESS), we expect the existence of a steady state probability density $\rho_{ss}(\va{x})$ whose time derivative identically vanishes. The steady state probability density can then be written as
\begin{equation}
    \rho_{ss}(\va{x})=\frac{1}{\sqrt{\qty(2\pi)^N|\vb{C}|}}\exp\qty(-\frac{1}{2}\va{x}^\intercal\vb{C}^{-1}\va{x})
\end{equation}
where $\vb{C}$ is the steady state covariance matrix, obtained as the solution to the continuous-time Lyapunov equation $\vb{G}\vb{C}+\vb{C}\vb{G}^\intercal=-2\vb{D}$ obtained by imposing that that time derivative of the covariance matrix vanishes. The corresponding steady state probability current is then given by
\begin{align}
    \vb{J}_{ss}(\va{x})&=\qty(\vb{G}+\vb{D}\vb{C}^{-1})\va{x}\rho_{ss}(\va{x})
\end{align}
The (ensemble) averaged entropy production rate is then easily found to be
\begin{equation}
    \label{eq:epr_general_expression_linear}
    \dot{S}=\int\dd \va{x}\, \frac{\va{J}_{ss}(\va{x})^\intercal\vb{D}^{-1}\va{J}_{ss}(\va{x})}{\rho_{ss}(\va{x})}=\Tr\qty[\vb{\Lambda}^\intercal\vb{D}\vb{\Lambda}\vb{C}]
\end{equation}
where $\vb{\Lambda}=\vb{D}^{-1}\vb{G}+\vb{C}^{-1}$ is the thermodynamic force matrix (related to the thermodynamic force field by $\vb{F}_\mathrm{th}=\vb{\Lambda}\va{x}$).

\subsection{Principal Component Analysis (PCA) Estimation of Entropy Production} \label{subsec:pca_linear_dynamics}

In the linear setting, we may wish to decompose the irreversible dynamics of the system to a set of reduced coordinates which exhibit the largest variance by performing principal component analysis (PCA) on the stochastic trajectory data. With the dynamics given by Eq.~\ref{eq:n_beads_langevin_app}, a trajectory of infinite length would would correspond to a point cloud of zero mean and covariance matrix $\vb{C}$ as the solution to the Lyapunov equation $\vb{G}\vb{C}+\vb{C}\vb{G}^\intercal=-2\vb{D}$. The covariance matrix is diagonalized by an orthogonal matrix $\vb{W}$ as
\begin{equation}
    \vb{C}=\vb{W}\vb{\Sigma}\vb{W}^\intercal
\end{equation}
where $\vb{\Sigma}$ is diagonal with descending ordered eigenvalues, and the columns of $\vb{W}$ contain the (ordered) normalized principal components (PCs). A naive dimensional reduction of the trajectory data is achieved by projecting the state $\va{x}$ onto the subspace spanned by the top $M$ PCs through
\begin{equation}
    \va{y}=\vb{W}^\intercal\va{x}
\end{equation}
with $\vb{W}_M$ containing the first $M$ columns of $\vb{W}$. The reduced covariance matrix is then
\begin{equation}
    \vb{C}_y=\vb{W}_M^\intercal\vb{C}\vb{W}_M.
\end{equation}

In this reduced subspace, we can then approximate the dynamics by substituting $\va{x}\approx\vb{W}_M\va{y}$, yielding the effective dynamics
\begin{equation}
    \dot{\va{y}}=\Tilde{\vb{G}}\va{y}+\Tilde{\vb{F}}\va{\xi}
\end{equation}
with the reduced matrices $\Tilde{\vb{G}}=\vb{W}_M^\intercal\vb{G}\vb{W}_M$ and $\Tilde{\vb{F}}=\vb{W}_M^\intercal\vb{F}$, along with a reduced diffusion tensor $\Tilde{\vb{D}}=\frac{1}{2}\Tilde{\vb{F}}\Tilde{\vb{F}}^\intercal=\vb{W}_M^\intercal\vb{D}\vb{W}_M$. To mirror the expression for the EPR in the reduced space, we define
\begin{equation}
    \Tilde{\vb{\Lambda}}=\Tilde{\vb{D}}^{-1}\Tilde{\vb{G}}+\vb{C}_y^{-1}
\end{equation}
which then similarly leads to the reduced EPR
\begin{equation}
    \dot{S}(M)= \Tr\qty[\tilde{\vb{\Lambda}}^\intercal\tilde{\vb{D}}\tilde{\vb{\Lambda}}\vb{C}_y].
\end{equation}
While the PCA basis can be whitened to yield a set of coordinate axes whose variance is normalized, this does not change the value of the reduced EPR $\dot{S}(M)$ since it introduces scaling factors into $\tilde{\vb{\Lambda}}$ and $\tilde{\vb{D}}$ whose effects exactly cancel out.

\section{Entropy Production Rate of a Driven Particle in a Periodic Potential} \label{app:periodic_potential_analytic}

In this model, the stochastic dynamics of a particle being driven in a periodic potential in the overdamped limit (equivalent to a particle exhibiting an overdamped Brownian motion in a tilted periodic potential) are governed by
\begin{equation}
    x_{t+1}=x_t+\qty(h-\dv{U}{x})\dd t+\sqrt{2D \,\dd t}\,\xi_t
\end{equation}
where $x$ is the one-dimensional position of the particle, $D$ is the diffusion constant, $U(x)$ is the potential periodic over $L$ so $U(x)=U(x+L)$, $h$ is the drive strength (or tilt, equivalently), and $\xi(t)$ is Gaussian random noise with $\expval{\xi(t)}=0$ and $\expval{\xi(t)\xi(t')}=\delta(t-t')$. Moving to the equivalent Fokker-Planck equation, we write
\begin{align}
    \pdv{\rho(x,t)}{t}&=-\pdv{J}{x}, \\
    \text{where }
    J(x,t)&=\qty(h-\dv{U(x)}{x})\rho(x,t)-D\pdv{\rho(x,t)}{x}. \nonumber
\end{align}
In a nonequilibrium steady state, the probability current $J(x)=J$ is spatially constant. This leads to
\begin{equation}
    D\dv{\rho_{ss}}{x}=(h-U')\rho_{ss}-J
\end{equation}
where $U'\equiv \dv{U(x)}{x}$ for brevity. This ordinary differential equation can be solved using an integrating factor $\mu(x)$, obtained as
\begin{align}
    \mu(x)&=\exp\qty[-\frac{1}{D}\int_0^x\dd y\, \qty(h-U'(y))] \nonumber \\
    &=\exp\qty(-\frac{hx-U(x)}{D}).
\end{align}
This gives the stationary distribution as
\begin{multline}
    \label{eq:tilted_stationary_distribution_constants}
    \rho_{ss}(x)=\exp\qty(\frac{hx-U(x)}{D}) \\
    \times\qty[c-\frac{J}{D}\int_0^x\dd y\,\exp\qty(-\frac{hy-U(y)}{D})]
\end{multline}
where $c$ is an integration constant such that $c$ and $J$ jointly normalize the distribution over one period. This periodicity constraint implies that $\rho_{ss}(x+L)=\rho_{ss}(x)$, and applying this constraint to the distribution at $x=0$ and $x=L$ gives
\begin{equation}
    c=\frac{J}{D(1-e^{-hL/D})}\int_0^L\dd y\,\exp\qty(-\frac{hy-U(y)}{D}).
\end{equation}
Now, we insert this into Eq.~\ref{eq:tilted_stationary_distribution_constants} to obtain the stationary distribution without any appearance of the integration constant as
\begin{widetext}
    \begin{equation}
    \rho_{ss}(x)=\exp\qty(\frac{hx-U(x)}{D})\frac{J}{D}\left[\frac{1}{1-e^{-hL/D}}\int_0^L\dd y\,\exp\qty(-\frac{hy-U(y)}{D})-\int_0^x\dd y\,\exp\qty(-\frac{hy-U(y)}{D})\right].
\end{equation}
\end{widetext}
To simplify the quantity $[\cdots]$, we define $\alpha=e^{-hL/D}$ and $q(y)=\exp\qty(-\frac{hy-U(y)}{D})$. Then, we break up the integral as
\begin{equation}
    \int_0^L\dd y\, q(y)=\underbrace{\int_0^x\dd y\, q(y)}_{I_x}+\underbrace{\int_x^L\dd y\, q(y)}_{I_{x\rightarrow L}}
\end{equation}
defining $I_x$ and $I_{x\rightarrow L}$ in the process. Now, we rewrite the expression $[\cdots]$ as
\begin{align}
    \frac{\int_0^L\dd y\, q(y)}{1-\alpha}-I_x&=\frac{I_x+I_{x\rightarrow L}}{1-\alpha}-I_x \nonumber \\
    &=\frac{I_{x\rightarrow L}+\alpha I_x}{1-\alpha}. \label{eq:periodic_intermediate_1}
\end{align}
Consider the integral of $q(y)$ from some $x$ over a period
\begin{align}
    \int_x^{x+L}\dd y\, q(y)&=\int_x^L\dd y\, q(y)+\int_L^{x+L}\dd y\, q(y) \nonumber \\
    &=I_{x\rightarrow L}+\int_0^x\dd y\, q(y+L) \nonumber \\
    &=I_{x\rightarrow L}+\alpha I_x
\end{align}
which is exactly the numerator from Eq.~\ref{eq:periodic_intermediate_1}. In the last line, we use the periodicity of $U(x)$ to impose the condition $q(y+L)=\alpha q(y)$. With this, we return to the original expression to rewrite the stationary density as
\begin{multline}
    \rho_{ss}(x)=\frac{J}{D}\exp\qty(\frac{hx-U(x)}{D})\frac{1}{1-e^{hL/D}}\\
    \int_x^{x+L}\dd y\,\exp\qty(-\frac{hy-U(y)}{D}).
\end{multline}
To remove the appearance of the probability current $J$, we simply impose the normalization condition of the density over a period $L$ to finally obtain the stationary density
\begin{align}
    \rho_{ss}(x)&=\frac{1}{Z}\exp\qty(\frac{hx-U(x)}{D}) \nonumber \\
    &\qquad\int_x^{x+L}\dd y\,\exp\qty(-\frac{hy-U(y)}{D}) \\
    \text{where } Z&=\int_0^L\dd x\,\exp\qty(\frac{hx-U(x)}{D}) \nonumber \\
    &\qquad\int_x^{x+L}\dd y\,\exp\qty(-\frac{hy-U(y)}{D})
\end{align}
which is in a desired form independent of the probability current $J$. Now, to solve for $J$ we write the explicit normalization condition
\begin{widetext}
    \begin{align}
    1=\frac{J}{D}\frac{Z}{1-e^{-hL/D}}\Longrightarrow
    J&=\frac{D(1-e^{-hL/D})}{Z} \nonumber \\
    &=\frac{D(1-e^{-hL/D})}{\displaystyle \int_0^L\dd x\,\exp\qty(-\frac{U(x)-hx}{D})\int_x^{x+L}\dd y\,\exp\qty(\frac{U(y)-hy}{D})}. \label{eq:periodic_intermediate_2}
    \end{align}
\end{widetext}
To proceed further, we need to insert an explicit form of $U(x)$, so we use the sinusoidal potential
\begin{equation}
    U(x)=U_0\cos\qty(\frac{2\pi x}{L})
\end{equation}
which has periodicity $L$, and can be shifted by an irrelevant constant such that $U(0)=0$. We now define the dimensionless quantities
\begin{equation}
    \label{eq:periodic_potential_dimensionless_variables}
    \kappa=\frac{hL}{2\pi D},\quad q=\frac{U_0}{D},\quad \theta=\frac{2\pi x}{L}\in[0,2\pi)
\end{equation}
where the first two constants are dimensionless if we work in units where $D=k_BT$. With these, we write the denominator of Eq.~\ref{eq:periodic_intermediate_2} as
\begin{equation}
    Z=\qty(\frac{L}{2\pi})^2\int_0^{2\pi}\dd\theta\,e^{-q\cos\theta+\kappa\theta}\int_\theta^{\theta+2\pi}\dd\phi\,e^{q\cos\phi-\kappa\phi}.
\end{equation}
Using the Jacobi-Anger expansion, we tame the exponents as
\begin{align}
    e^{q\cos\phi-\kappa\phi}&=e^{-\kappa\phi}\sum_{m=-\infty}^\infty I_m(q)e^{im\phi} \\
    e^{-q\cos\theta+\kappa\theta}&=e^{\kappa\theta}\sum_{n=-\infty}^\infty (-1)^n I_n(q)e^{in\theta}
\end{align}
where $I_n(z)$ is the $n$-th modified Bessel function of the first kind as a function of $z$. Performing the integral over $\phi$ first, we have
\begin{equation}
    \int_\theta^{\theta+2\pi}\dd\phi\, e^{-(\kappa+im)\phi}=\frac{1-e^{-2\pi\kappa}}{\kappa-im}e^{-(\kappa-im)\theta}.
\end{equation}
We plug this back into the denominator and complete the remaining integral to obtain
\begin{equation}
    Z=\frac{L^2}{2\pi}\qty(1-e^{-2\pi\kappa})\qty[\frac{I_0(q)^2}{\kappa}+2\sum_{n=1}^\infty(-1)^n\frac{\kappa I_n(q)^2}{\kappa^2 +n^2}]
\end{equation}
where we use the fact that $I_{n}(q)=I_{-n}(q)$ for integer order $n$. Finally, we return to Eq.~\ref{eq:periodic_intermediate_2} to obtain the final probability current
\begin{equation}
    J=\frac{2\pi D}{L^2}\qty[\frac{I_0(q)^2}{\kappa}+2\sum_{n=1}^\infty(-1)^n\frac{\kappa I_n(q)^2}{\kappa^2 +n^2}]^{-1}.
\end{equation}
To compute the EPR, we once again use the relation from Eq.~\ref{eq:epr_general_expression_linear} to write
\begin{equation}
    \dot{S}=\int_0^L\dd x\,\frac{J^2}{D\rho_{ss}(x)}.
\end{equation}
Before attempting to integrate this, we note that the Fokker-Planck equation reads
\begin{equation}
    \frac{J}{\rho_{ss}(x)}=h-U'-D\dv{x}\log\rho_{ss}(x).
\end{equation}
Inserting this into the integrand for the EPR, we finally obtain
\begin{align}
    \dot{S}&=\frac{J}{D}\int_0^L\dd x\,\qty[h-U'-D\dv{x}\log\rho_{ss}(x)] \nonumber \\
    &=\frac{J}{D}hL-\underbrace{\qty[U(x)-D\log\rho_{ss}(x)]_{x=0}^{x=L}}_{=0} \nonumber \\
    &=\frac{2\pi h}{L}\qty[\frac{I_0(q)^2}{\kappa}+2\sum_{n=1}^\infty(-1)^n\frac{\kappa I_n(q)^2}{\kappa^2 +n^2}]^{-1}
\end{align}
which is the EPR in terms of the problem variables as desired. In the limit of a vanishing drive (zero tilt), we confirm that the EPR vanishes as expected. Additionally, setting the amplitude of the periodic potential $U_0=0$ produces a purely linear potential term, and we recover a uniform stationary density $\rho_{ss}(x)=1/L$, and constant current $J=h/L$ as expected. This gives an EPR of $\dot{S}=hv/D=h^2/D$, which quantifies the continual dissipation of work $h^2$ into heat at a bath temperature determined by $D$.

\section{Equivalence between LENS and NEEP Objectives} \label{app:lens_neep_equivalence}

In prior work on studying irreversibility with feedforward neural networks, the neural estimator for entropy production (NEEP) objective is defined as~\cite{Kim2020}
\begin{align}
    \label{eq:neep_objective}
    J(\theta)&=\mathbb{E}_{t}\mathbb{E}_{p(\va{x}\rightarrow \va{x}')}\qty[ S_\theta(\va{x},\va{x}')-e^{- S_\theta(\va{x},\va{x}')}] \\
    \text{where } S_\theta(\va{x},\va{x}')&=h_\theta(s_t=\va{x}, s_{t+1}=\va{x}') \nonumber \\
    &\qquad-h_\theta(s_t=\va{x}', s_{t+1}=\va{x}).
\end{align}

Here, $h_\theta(s_t,s_{t+1})$ is a learned function of the system's states at consecutive times $t$ and $t+1$, with model parameters $\theta$. The NEEP objective can be rewritten as
\begin{equation}
    \label{eq:neep_f_div_form}
    J(\theta)=\mathbb{E}_{t}\mathbb{E}_{p(\va{x}\rightarrow \va{x}')}\qty[ S_\theta(\va{x},\va{x}')]-\mathbb{E}_{t}\mathbb{E}_{p(\va{x}'\rightarrow \va{x})}\qty[e^{ S_\theta(\va{x},\va{x}')}]
\end{equation}
where we use the fact that $S_\theta(\va{x},\va{x}')$ is antisymmetric under exchange of its arguments. This has the form of an expectation of a function over one distribution $p(x)$ and the expectation of a variant of that function over another distribution $q(x)$, reminiscent of a variational $f$-divergence estimator. We restate a known result for placing a lower bound on an $f$-divergence~\cite{Nowozin2016}
\begin{equation}
    D_f\qty[p(x)\|q(x)]\geq \max_{T(x)}\qty{\mathbb{E}_{p(x)}\qty[T(x)]-\mathbb{E}_{q(x)}\qty[f^*(T(x))]}
\end{equation}
where $f^*(x)$ is the Fenchel (convex) conjugate of the generator $f(x)$ used to define the particular $f$-divergence. This bound is tight when $T(x)=\mathrm{d}f(\frac{p(x)}{q(x)})/\mathrm{d}x$. Now, we note that the KL divergence is an $f$-divergence with $f(x)=x\log x$, which has Fenchel conjugate $e^{x-1}$. We thus obtain the inequality
\begin{equation}
    D_{KL}\qty[p(x)\|q(x)]\geq \max_{T(x)}\qty{\mathbb{E}_{p(x)}\qty[T(x)+1]-\mathbb{E}_{q(x)}\qty[e^{T(x)}]}
\end{equation}
after a change of variables $x\rightarrow x+1$. Now, we contextualize this to the case of estimating entropy production rates by substituting $x\rightarrow (\va{x},\va{x}')$, and choose the distributions $p(x)\rightarrow p(\va{x}\rightarrow \va{x}')$ and $q(x)\rightarrow p(\va{x}'\rightarrow \va{x})$. Then, we write the KL divergence using this variational form as
\begin{widetext}
    \begin{equation}
    D_{KL}\qty[p(\va{x}\rightarrow \va{x}')\| p(\va{x}'\rightarrow \va{x})]\geq \max_{S(\va{x},\va{x}')\in\mathcal{F}}\qty{\mathbb{E}_{p(\va{x}\rightarrow \va{x}')}\qty[S(\va{x},\va{x}')]-\mathbb{E}_{p(\va{x}'\rightarrow \va{x})}\qty[e^{S(\va{x},\va{x}')}]+1}
\end{equation}
\end{widetext}
where $\mathcal{F}$ is the set of functions $h:\Omega\times\Omega\rightarrow\mathbb{R}$ such that $h(\va{x}',\va{x})=-h(\va{x},\va{x}')$ for all states $\va{x}$ in the state space $\Omega$. In this form, we see that the variational form of the KL divergence for the entropy estimation matches the NEEP objective in Eq.~\ref{eq:neep_f_div_form} up to a constant, implying that the NEEP estimator computes the entropy production rate by maximizing the lower bound on the KL divergence through the learned function $S_\theta(\va{x},\va{x}')$. Interestingly, this suggests that the validation loss in training through the NEEP objective is itself relevant, since it is the lower bound on the KL divergence, and thus directly estimates the entropy production rate. We also note that the function $S_\theta(\va{x},\va{x}')$ is $\mathcal{O}(\sqrt{\dd t})$ at leading order, so we can expand the NEEP objective to second order to obtain a reduced objective
\begin{equation}
    \mathcal{V}(\theta)=\frac{1}{\dd t}\mathbb{E}_t\mathbb{E}_{p(\va{x}\rightarrow \va{x}')}\qty[2S_\theta(\va{x},\va{x}')-\frac{1}{2}S_\theta(\va{x},\va{x}')^2].
\end{equation}
Similarly expanding the LENS objective function to $\mathcal{O}(\dd t)$ returns an objective which is proportional to the reduced NEEP objective up to a factor, implying that the NEEP and LENS objectives are indeed equivalent and return the same optimal learned function $S_\theta(\va{x},\va{x}')$ under maximization of the reduced objective. This also reconciles LENS as being an equivalent variational $f$-divergence estimator of the entropy production rate.

\section{LENS as a Low-Rank Approximation} \label{app:lens_low_rank_approximation}

Here, we consider a geometric interpretation of the LENS framework rooted in low-rank approximations of the underlying entropy producing dynamics. From the linear stochastic dynamical system in Eq.~\ref{eq:n_beads_langevin_app} and corresponding representational form of the logit in Eq.~\ref{eq:lens_logit_parametrization}, we write the logit as
\begin{equation}
    S(\va{x},\va{x}')=\va{\phi}(\va{x}')^\intercal\vb{K}\va{\phi}(\va{x})+\frac{1}{2}\va{\phi}(\va{x}')^\intercal\vb{L}\va{\phi}(\va{x}')-\frac{1}{2}\va{\phi}(\va{x})^\intercal\vb{L}\va{\phi}(\va{x})
\end{equation}
where $\vb{K}$ is a skew-symmetric matrix, and $\vb{L}$ is a symmetric matrix. Restricting ourselves to linear representations of the form $\va{\phi}(\va{x})=\vb{P}\va{x}$ where $\vb{P}$ is a linear projection from $\mathbb{R}^N$ to $\mathbb{R}^M$ where $M\leq N$, we then have
\begin{equation}
    S(\va{x},\va{x}')=\va{x}'^\intercal\tilde{\vb{K}}\va{x}+\frac{1}{2}\va{x}'^\intercal\tilde{\vb{L}}\va{x}'-\frac{1}{2}\va{x}^\intercal\tilde{\vb{L}}\va{x}.
\end{equation}
where $\tilde{\vb{K}}=\vb{P}^\intercal\vb{K}\vb{P}$ and $\tilde{\vb{L}}=\vb{P}^\intercal\vb{L}\vb{P}$ are the projected skew-symmetric and symmetric learned matrices respectively. When $M<N$, the optimal logit necessarily underestimates the entropy production rate~\cite{Gnesotto2020}, with the accuracy of the estimate improving with increasing $M$. Inserting this explicit parametrization into the LENS objective returns
\begin{widetext}
    \begin{align}
    \mathcal{V}_M(\vb{P},\vb{K},\vb{L})&= 2\Tr\qty[(\tilde{\vb{K}}+\tilde{\vb{L}})\vb{C}\vb{G}^\intercal]+4\Tr\qty[(\tilde{\vb{K}}+\tilde{\vb{L}})\vb{D}] -\Tr\qty[(\tilde{\vb{K}}+\tilde{\vb{L}})^\intercal\vb{D}(\tilde{\vb{K}}+\tilde{\vb{L}})\vb{C}] \nonumber \\
    &= -2\Tr\qty[(\tilde{\vb{K}}+\tilde{\vb{L}})\vb{G}\vb{C}]-\Tr\qty[(\tilde{\vb{K}}+\tilde{\vb{L}})^\intercal\vb{D}(\tilde{\vb{K}}+\tilde{\vb{L}})\vb{C}]\label{eq:reduced_lens_objective_app}
    \end{align}
\end{widetext}
The optimal value $\mathcal{V}_M^*(\vb{P},\vb{K},\vb{L})$ (from maximizing $\mathcal{V}_M$) is thus the best underestimate of the entropy production rate with the representation constrained to be linear in $M$ dimensions, given the form of the logit in Eq.~\ref{eq:lens_logit_parametrization}. We compute the fixed points by taking a derivative of the objective with respect to $\tilde{\vb{K}}+\tilde{\vb{L}}$, from which we finally obtain the optimality condition
\begin{equation}
    \label{eq:lens_fixed_points}
    \tilde{\vb{K}}+\tilde{\vb{L}}= \vb{P}^\intercal(\vb{K}+\vb{L})\vb{P}= \vb{D}^{-1}\vb{G}+\vb{C}^{-1}=\vb{\Lambda}.
\end{equation}
This indicates that the underlying parametrization of the LENS logit encourages the linear representations to learn a low-rank approximation of $\vb{\Lambda}$ through the projection matrix $\vb{P}$. This approximation is attained through the learning process such that $\vb{A}_\theta$ approximates the skew-symmetric part of $\vb{\Lambda}$, while $\vb{B}_\theta$ approximates the symmetric part of $\vb{\Lambda}$.


%% file: arxiv.bbl
%

%% file: arxiv.bbl
\begin{thebibliography}{52}%
\makeatletter
\providecommand \@ifxundefined [1]{%
 \@ifx{#1\undefined}
}%
\providecommand \@ifnum [1]{%
 \ifnum #1\expandafter \@firstoftwo
 \else \expandafter \@secondoftwo
 \fi
}%
\providecommand \@ifx [1]{%
 \ifx #1\expandafter \@firstoftwo
 \else \expandafter \@secondoftwo
 \fi
}%
\providecommand \natexlab [1]{#1}%
\providecommand \enquote  [1]{``#1''}%
\providecommand \bibnamefont  [1]{#1}%
\providecommand \bibfnamefont [1]{#1}%
\providecommand \citenamefont [1]{#1}%
\providecommand \href@noop [0]{\@secondoftwo}%
\providecommand \href [0]{\begingroup \@sanitize@url \@href}%
\providecommand \@href[1]{\@@startlink{#1}\@@href}%
\providecommand \@@href[1]{\endgroup#1\@@endlink}%
\providecommand \@sanitize@url [0]{\catcode `\\12\catcode `\$12\catcode `\&12\catcode `\#12\catcode `\^12\catcode `\_12\catcode `\%12\relax}%
\providecommand \@@startlink[1]{}%
\providecommand \@@endlink[0]{}%
\providecommand \url  [0]{\begingroup\@sanitize@url \@url }%
\providecommand \@url [1]{\endgroup\@href {#1}{\urlprefix }}%
\providecommand \urlprefix  [0]{URL }%
\providecommand \Eprint [0]{\href }%
\providecommand \doibase [0]{https://doi.org/}%
\providecommand \selectlanguage [0]{\@gobble}%
\providecommand \bibinfo  [0]{\@secondoftwo}%
\providecommand \bibfield  [0]{\@secondoftwo}%
\providecommand \translation [1]{[#1]}%
\providecommand \BibitemOpen [0]{}%
\providecommand \bibitemStop [0]{}%
\providecommand \bibitemNoStop [0]{.\EOS\space}%
\providecommand \EOS [0]{\spacefactor3000\relax}%
\providecommand \BibitemShut  [1]{\csname bibitem#1\endcsname}%
\let\auto@bib@innerbib\@empty
\bibitem [{\citenamefont {Battle}\ \emph {et~al.}(2016)\citenamefont {Battle}, \citenamefont {Broedersz}, \citenamefont {Fakhri}, \citenamefont {Geyer}, \citenamefont {Howard}, \citenamefont {Schmidt},\ and\ \citenamefont {MacKintosh}}]{Battle2016}%
  \BibitemOpen
  \bibfield  {author} {\bibinfo {author} {\bibfnamefont {C.}~\bibnamefont {Battle}}, \bibinfo {author} {\bibfnamefont {C.~P.}\ \bibnamefont {Broedersz}}, \bibinfo {author} {\bibfnamefont {N.}~\bibnamefont {Fakhri}}, \bibinfo {author} {\bibfnamefont {V.~F.}\ \bibnamefont {Geyer}}, \bibinfo {author} {\bibfnamefont {J.}~\bibnamefont {Howard}}, \bibinfo {author} {\bibfnamefont {C.~F.}\ \bibnamefont {Schmidt}},\ and\ \bibinfo {author} {\bibfnamefont {F.~C.}\ \bibnamefont {MacKintosh}},\ }\bibfield  {title} {\bibinfo {title} {Broken detailed balance at mesoscopic scales in active biological systems},\ }\href {https://doi.org/10.1126/science.aac8167} {\bibfield  {journal} {\bibinfo  {journal} {Science}\ }\textbf {\bibinfo {volume} {352}},\ \bibinfo {pages} {604} (\bibinfo {year} {2016})}\BibitemShut {NoStop}%
\bibitem [{\citenamefont {Gladrow}\ \emph {et~al.}(2017)\citenamefont {Gladrow}, \citenamefont {Broedersz},\ and\ \citenamefont {Schmidt}}]{Gladrow2017}%
  \BibitemOpen
  \bibfield  {author} {\bibinfo {author} {\bibfnamefont {J.}~\bibnamefont {Gladrow}}, \bibinfo {author} {\bibfnamefont {C.~P.}\ \bibnamefont {Broedersz}},\ and\ \bibinfo {author} {\bibfnamefont {C.~F.}\ \bibnamefont {Schmidt}},\ }\bibfield  {title} {\bibinfo {title} {Nonequilibrium dynamics of probe filaments in actin-myosin networks},\ }\href {https://doi.org/10.1103/PhysRevE.96.022408} {\bibfield  {journal} {\bibinfo  {journal} {Physical Review E}\ }\textbf {\bibinfo {volume} {96}},\ \bibinfo {pages} {022408} (\bibinfo {year} {2017})}\BibitemShut {NoStop}%
\bibitem [{\citenamefont {Lynn}\ \emph {et~al.}(2022)\citenamefont {Lynn}, \citenamefont {Holmes}, \citenamefont {Bialek},\ and\ \citenamefont {Schwab}}]{Lynn2022a}%
  \BibitemOpen
  \bibfield  {author} {\bibinfo {author} {\bibfnamefont {C.~W.}\ \bibnamefont {Lynn}}, \bibinfo {author} {\bibfnamefont {C.~M.}\ \bibnamefont {Holmes}}, \bibinfo {author} {\bibfnamefont {W.}~\bibnamefont {Bialek}},\ and\ \bibinfo {author} {\bibfnamefont {D.~J.}\ \bibnamefont {Schwab}},\ }\bibfield  {title} {\bibinfo {title} {Emergence of local irreversibility in complex interacting systems},\ }\href {https://doi.org/10.1103/PhysRevE.106.034102} {\bibfield  {journal} {\bibinfo  {journal} {Physical Review E}\ }\textbf {\bibinfo {volume} {106}},\ \bibinfo {pages} {034102} (\bibinfo {year} {2022})}\BibitemShut {NoStop}%
\bibitem [{\citenamefont {Seifert}(2012)}]{Seifert2012}%
  \BibitemOpen
  \bibfield  {author} {\bibinfo {author} {\bibfnamefont {U.}~\bibnamefont {Seifert}},\ }\bibfield  {title} {\bibinfo {title} {Stochastic thermodynamics, fluctuation theorems and molecular machines},\ }\href {https://doi.org/10.1088/0034-4885/75/12/126001} {\bibfield  {journal} {\bibinfo  {journal} {Reports on Progress in Physics}\ }\textbf {\bibinfo {volume} {75}},\ \bibinfo {pages} {126001} (\bibinfo {year} {2012})}\BibitemShut {NoStop}%
\bibitem [{\citenamefont {Qian}\ \emph {et~al.}(2016)\citenamefont {Qian}, \citenamefont {Kjelstrup}, \citenamefont {Kolomeisky},\ and\ \citenamefont {Bedeaux}}]{Qian2016}%
  \BibitemOpen
  \bibfield  {author} {\bibinfo {author} {\bibfnamefont {H.}~\bibnamefont {Qian}}, \bibinfo {author} {\bibfnamefont {S.}~\bibnamefont {Kjelstrup}}, \bibinfo {author} {\bibfnamefont {A.~B.}\ \bibnamefont {Kolomeisky}},\ and\ \bibinfo {author} {\bibfnamefont {D.}~\bibnamefont {Bedeaux}},\ }\bibfield  {title} {\bibinfo {title} {Entropy production in mesoscopic stochastic thermodynamics: nonequilibrium kinetic cycles driven by chemical potentials, temperatures, and mechanical forces},\ }\href {https://doi.org/10.1088/0953-8984/28/15/153004} {\bibfield  {journal} {\bibinfo  {journal} {Journal of Physics: Condensed Matter}\ }\textbf {\bibinfo {volume} {28}},\ \bibinfo {pages} {153004} (\bibinfo {year} {2016})}\BibitemShut {NoStop}%
\bibitem [{\citenamefont {Gnesotto}\ \emph {et~al.}(2018)\citenamefont {Gnesotto}, \citenamefont {Mura}, \citenamefont {Gladrow},\ and\ \citenamefont {Broedersz}}]{Gnesotto2018}%
  \BibitemOpen
  \bibfield  {author} {\bibinfo {author} {\bibfnamefont {F.~S.}\ \bibnamefont {Gnesotto}}, \bibinfo {author} {\bibfnamefont {F.}~\bibnamefont {Mura}}, \bibinfo {author} {\bibfnamefont {J.}~\bibnamefont {Gladrow}},\ and\ \bibinfo {author} {\bibfnamefont {C.~P.}\ \bibnamefont {Broedersz}},\ }\bibfield  {title} {\bibinfo {title} {Broken detailed balance and non-equilibrium dynamics in living systems: a review},\ }\href {https://doi.org/10.1088/1361-6633/aab3ed} {\bibfield  {journal} {\bibinfo  {journal} {Reports on Progress in Physics}\ }\textbf {\bibinfo {volume} {81}},\ \bibinfo {pages} {066601} (\bibinfo {year} {2018})}\BibitemShut {NoStop}%
\bibitem [{\citenamefont {Mura}\ \emph {et~al.}(2018)\citenamefont {Mura}, \citenamefont {Gradziuk},\ and\ \citenamefont {Broedersz}}]{Mura2018}%
  \BibitemOpen
  \bibfield  {author} {\bibinfo {author} {\bibfnamefont {F.}~\bibnamefont {Mura}}, \bibinfo {author} {\bibfnamefont {G.}~\bibnamefont {Gradziuk}},\ and\ \bibinfo {author} {\bibfnamefont {C.~P.}\ \bibnamefont {Broedersz}},\ }\bibfield  {title} {\bibinfo {title} {Nonequilibrium {Scaling} {Behavior} in {Driven} {Soft} {Biological} {Assemblies}},\ }\href {https://doi.org/10.1103/PhysRevLett.121.038002} {\bibfield  {journal} {\bibinfo  {journal} {Physical Review Letters}\ }\textbf {\bibinfo {volume} {121}},\ \bibinfo {pages} {038002} (\bibinfo {year} {2018})}\BibitemShut {NoStop}%
\bibitem [{\citenamefont {Tan}\ \emph {et~al.}(2021)\citenamefont {Tan}, \citenamefont {Watson}, \citenamefont {Chao}, \citenamefont {Li}, \citenamefont {Gingrich}, \citenamefont {Horowitz},\ and\ \citenamefont {Fakhri}}]{Tan2021}%
  \BibitemOpen
  \bibfield  {author} {\bibinfo {author} {\bibfnamefont {T.~H.}\ \bibnamefont {Tan}}, \bibinfo {author} {\bibfnamefont {G.~A.}\ \bibnamefont {Watson}}, \bibinfo {author} {\bibfnamefont {Y.-C.}\ \bibnamefont {Chao}}, \bibinfo {author} {\bibfnamefont {J.}~\bibnamefont {Li}}, \bibinfo {author} {\bibfnamefont {T.~R.}\ \bibnamefont {Gingrich}}, \bibinfo {author} {\bibfnamefont {J.~M.}\ \bibnamefont {Horowitz}},\ and\ \bibinfo {author} {\bibfnamefont {N.}~\bibnamefont {Fakhri}},\ }\href {https://doi.org/10.48550/arXiv.2107.05701} {\bibinfo {title} {Scale-dependent irreversibility in living matter}} (\bibinfo {year} {2021})\BibitemShut {NoStop}%
\bibitem [{\citenamefont {Wang}\ \emph {et~al.}(2005)\citenamefont {Wang}, \citenamefont {Kulkarni},\ and\ \citenamefont {Verdu}}]{Wang2005}%
  \BibitemOpen
  \bibfield  {author} {\bibinfo {author} {\bibfnamefont {Q.}~\bibnamefont {Wang}}, \bibinfo {author} {\bibfnamefont {S.}~\bibnamefont {Kulkarni}},\ and\ \bibinfo {author} {\bibfnamefont {S.}~\bibnamefont {Verdu}},\ }\bibfield  {title} {\bibinfo {title} {Divergence estimation of continuous distributions based on data-dependent partitions},\ }\href {https://doi.org/10.1109/TIT.2005.853314} {\bibfield  {journal} {\bibinfo  {journal} {IEEE Transactions on Information Theory}\ }\textbf {\bibinfo {volume} {51}},\ \bibinfo {pages} {3064} (\bibinfo {year} {2005})}\BibitemShut {NoStop}%
\bibitem [{\citenamefont {Ziv}\ and\ \citenamefont {Merhav}(1993)}]{Ziv1993}%
  \BibitemOpen
  \bibfield  {author} {\bibinfo {author} {\bibfnamefont {J.}~\bibnamefont {Ziv}}\ and\ \bibinfo {author} {\bibfnamefont {N.}~\bibnamefont {Merhav}},\ }\bibfield  {title} {\bibinfo {title} {A measure of relative entropy between individual sequences with application to universal classification},\ }\href {https://doi.org/10.1109/18.243444} {\bibfield  {journal} {\bibinfo  {journal} {IEEE Transactions on Information Theory}\ }\textbf {\bibinfo {volume} {39}},\ \bibinfo {pages} {1270} (\bibinfo {year} {1993})}\BibitemShut {NoStop}%
\bibitem [{\citenamefont {Rold\'an}\ and\ \citenamefont {Parrondo}(2012)}]{Roldan2012}%
  \BibitemOpen
  \bibfield  {author} {\bibinfo {author} {\bibfnamefont {R.}~\bibnamefont {Rold\'an}}\ and\ \bibinfo {author} {\bibfnamefont {J.~M.~R.}\ \bibnamefont {Parrondo}},\ }\bibfield  {title} {\bibinfo {title} {Entropy production and {Kullback}-{Leibler} divergence between stationary trajectories of discrete systems},\ }\href {https://doi.org/10.1103/PhysRevE.85.031129} {\bibfield  {journal} {\bibinfo  {journal} {Physical Review E}\ }\textbf {\bibinfo {volume} {85}},\ \bibinfo {pages} {031129} (\bibinfo {year} {2012})}\BibitemShut {NoStop}%
\bibitem [{\citenamefont {Barato}\ and\ \citenamefont {Seifert}(2015)}]{Barato2015}%
  \BibitemOpen
  \bibfield  {author} {\bibinfo {author} {\bibfnamefont {A.~C.}\ \bibnamefont {Barato}}\ and\ \bibinfo {author} {\bibfnamefont {U.}~\bibnamefont {Seifert}},\ }\bibfield  {title} {\bibinfo {title} {Thermodynamic {Uncertainty} {Relation} for {Biomolecular} {Processes}},\ }\href {https://doi.org/10.1103/PhysRevLett.114.158101} {\bibfield  {journal} {\bibinfo  {journal} {Physical Review Letters}\ }\textbf {\bibinfo {volume} {114}},\ \bibinfo {pages} {158101} (\bibinfo {year} {2015})}\BibitemShut {NoStop}%
\bibitem [{\citenamefont {Gingrich}\ \emph {et~al.}(2016)\citenamefont {Gingrich}, \citenamefont {Horowitz}, \citenamefont {Perunov},\ and\ \citenamefont {England}}]{Gingrich2016}%
  \BibitemOpen
  \bibfield  {author} {\bibinfo {author} {\bibfnamefont {T.~R.}\ \bibnamefont {Gingrich}}, \bibinfo {author} {\bibfnamefont {J.~M.}\ \bibnamefont {Horowitz}}, \bibinfo {author} {\bibfnamefont {N.}~\bibnamefont {Perunov}},\ and\ \bibinfo {author} {\bibfnamefont {J.~L.}\ \bibnamefont {England}},\ }\bibfield  {title} {\bibinfo {title} {Dissipation bounds all steady-state current fluctuations},\ }\href {https://doi.org/10.1103/PhysRevLett.116.120601} {\bibfield  {journal} {\bibinfo  {journal} {Physical Review Letters}\ }\textbf {\bibinfo {volume} {116}},\ \bibinfo {pages} {120601} (\bibinfo {year} {2016})}\BibitemShut {NoStop}%
\bibitem [{\citenamefont {Horowitz}\ and\ \citenamefont {Gingrich}(2020)}]{Horowitz2020}%
  \BibitemOpen
  \bibfield  {author} {\bibinfo {author} {\bibfnamefont {J.~M.}\ \bibnamefont {Horowitz}}\ and\ \bibinfo {author} {\bibfnamefont {T.~R.}\ \bibnamefont {Gingrich}},\ }\bibfield  {title} {\bibinfo {title} {Thermodynamic uncertainty relations constrain non-equilibrium fluctuations},\ }\href {https://doi.org/10.1038/s41567-019-0702-6} {\bibfield  {journal} {\bibinfo  {journal} {Nature Physics}\ }\textbf {\bibinfo {volume} {16}},\ \bibinfo {pages} {15} (\bibinfo {year} {2020})}\BibitemShut {NoStop}%
\bibitem [{\citenamefont {Li}\ \emph {et~al.}(2019)\citenamefont {Li}, \citenamefont {Horowitz}, \citenamefont {Gingrich},\ and\ \citenamefont {Fakhri}}]{Li2019}%
  \BibitemOpen
  \bibfield  {author} {\bibinfo {author} {\bibfnamefont {J.}~\bibnamefont {Li}}, \bibinfo {author} {\bibfnamefont {J.~M.}\ \bibnamefont {Horowitz}}, \bibinfo {author} {\bibfnamefont {T.~R.}\ \bibnamefont {Gingrich}},\ and\ \bibinfo {author} {\bibfnamefont {N.}~\bibnamefont {Fakhri}},\ }\bibfield  {title} {\bibinfo {title} {Quantifying dissipation using fluctuating currents},\ }\href {https://doi.org/10.1038/s41467-019-09631-x} {\bibfield  {journal} {\bibinfo  {journal} {Nature Communications}\ }\textbf {\bibinfo {volume} {10}},\ \bibinfo {pages} {1666} (\bibinfo {year} {2019})}\BibitemShut {NoStop}%
\bibitem [{\citenamefont {Rold\'an}\ \emph {et~al.}(2021)\citenamefont {Rold\'an}, \citenamefont {Barral}, \citenamefont {Martin}, \citenamefont {Parrondo},\ and\ \citenamefont {J\"ulicher}}]{Roldan2021}%
  \BibitemOpen
  \bibfield  {author} {\bibinfo {author} {\bibfnamefont {E.}~\bibnamefont {Rold\'an}}, \bibinfo {author} {\bibfnamefont {J.}~\bibnamefont {Barral}}, \bibinfo {author} {\bibfnamefont {P.}~\bibnamefont {Martin}}, \bibinfo {author} {\bibfnamefont {J.~M.~R.}\ \bibnamefont {Parrondo}},\ and\ \bibinfo {author} {\bibfnamefont {F.}~\bibnamefont {J\"ulicher}},\ }\bibfield  {title} {\bibinfo {title} {Quantifying entropy production in active fluctuations of the hair-cell bundle from time irreversibility and uncertainty relations},\ }\href {https://doi.org/10.1088/1367-2630/ac0f18} {\bibfield  {journal} {\bibinfo  {journal} {New Journal of Physics}\ }\textbf {\bibinfo {volume} {23}},\ \bibinfo {pages} {083013} (\bibinfo {year} {2021})}\BibitemShut {NoStop}%
\bibitem [{\citenamefont {Gnesotto}\ \emph {et~al.}(2020)\citenamefont {Gnesotto}, \citenamefont {Gradziuk}, \citenamefont {Ronceray},\ and\ \citenamefont {Broedersz}}]{Gnesotto2020}%
  \BibitemOpen
  \bibfield  {author} {\bibinfo {author} {\bibfnamefont {F.~S.}\ \bibnamefont {Gnesotto}}, \bibinfo {author} {\bibfnamefont {G.}~\bibnamefont {Gradziuk}}, \bibinfo {author} {\bibfnamefont {P.}~\bibnamefont {Ronceray}},\ and\ \bibinfo {author} {\bibfnamefont {C.~P.}\ \bibnamefont {Broedersz}},\ }\bibfield  {title} {\bibinfo {title} {Learning the non-equilibrium dynamics of {Brownian} movies},\ }\href {https://doi.org/10.1038/s41467-020-18796-9} {\bibfield  {journal} {\bibinfo  {journal} {Nature Communications}\ }\textbf {\bibinfo {volume} {11}},\ \bibinfo {pages} {5378} (\bibinfo {year} {2020})}\BibitemShut {NoStop}%
\bibitem [{\citenamefont {Kim}\ \emph {et~al.}(2020)\citenamefont {Kim}, \citenamefont {Bae}, \citenamefont {Lee},\ and\ \citenamefont {Jeong}}]{Kim2020}%
  \BibitemOpen
  \bibfield  {author} {\bibinfo {author} {\bibfnamefont {D.-K.}\ \bibnamefont {Kim}}, \bibinfo {author} {\bibfnamefont {Y.}~\bibnamefont {Bae}}, \bibinfo {author} {\bibfnamefont {S.}~\bibnamefont {Lee}},\ and\ \bibinfo {author} {\bibfnamefont {H.}~\bibnamefont {Jeong}},\ }\bibfield  {title} {\bibinfo {title} {Learning {Entropy} {Production} via {Neural} {Networks}},\ }\href {https://doi.org/10.1103/PhysRevLett.125.140604} {\bibfield  {journal} {\bibinfo  {journal} {Physical Review Letters}\ }\textbf {\bibinfo {volume} {125}},\ \bibinfo {pages} {140604} (\bibinfo {year} {2020})}\BibitemShut {NoStop}%
\bibitem [{\citenamefont {Bae}\ \emph {et~al.}(2022)\citenamefont {Bae}, \citenamefont {Kim},\ and\ \citenamefont {Jeong}}]{Bae2022}%
  \BibitemOpen
  \bibfield  {author} {\bibinfo {author} {\bibfnamefont {Y.}~\bibnamefont {Bae}}, \bibinfo {author} {\bibfnamefont {D.-K.}\ \bibnamefont {Kim}},\ and\ \bibinfo {author} {\bibfnamefont {H.}~\bibnamefont {Jeong}},\ }\bibfield  {title} {\bibinfo {title} {Inferring dissipation maps from videos using convolutional neural networks},\ }\href {https://doi.org/10.1103/PhysRevResearch.4.033094} {\bibfield  {journal} {\bibinfo  {journal} {Physical Review Research}\ }\textbf {\bibinfo {volume} {4}},\ \bibinfo {pages} {033094} (\bibinfo {year} {2022})}\BibitemShut {NoStop}%
\bibitem [{\citenamefont {Lyu}\ \emph {et~al.}(2024)\citenamefont {Lyu}, \citenamefont {Ray},\ and\ \citenamefont {Crutchfield}}]{Lyu2024}%
  \BibitemOpen
  \bibfield  {author} {\bibinfo {author} {\bibfnamefont {J.}~\bibnamefont {Lyu}}, \bibinfo {author} {\bibfnamefont {K.~J.}\ \bibnamefont {Ray}},\ and\ \bibinfo {author} {\bibfnamefont {J.~P.}\ \bibnamefont {Crutchfield}},\ }\bibfield  {title} {\bibinfo {title} {Learning entropy production from underdamped {Langevin} trajectories},\ }\href {https://doi.org/10.1103/PhysRevE.110.064151} {\bibfield  {journal} {\bibinfo  {journal} {Physical Review E}\ }\textbf {\bibinfo {volume} {110}},\ \bibinfo {pages} {064151} (\bibinfo {year} {2024})}\BibitemShut {NoStop}%
\bibitem [{\citenamefont {Vodret}\ \emph {et~al.}(2024)\citenamefont {Vodret}, \citenamefont {Pacini},\ and\ \citenamefont {Bongiorno}}]{Vodret2024}%
  \BibitemOpen
  \bibfield  {author} {\bibinfo {author} {\bibfnamefont {M.}~\bibnamefont {Vodret}}, \bibinfo {author} {\bibfnamefont {C.}~\bibnamefont {Pacini}},\ and\ \bibinfo {author} {\bibfnamefont {C.}~\bibnamefont {Bongiorno}},\ }\bibfield  {title} {\bibinfo {title} {Functional decomposition and estimation of irreversibility in time series via machine learning},\ }\href {https://doi.org/10.1103/PhysRevE.110.064310} {\bibfield  {journal} {\bibinfo  {journal} {Physical Review E}\ }\textbf {\bibinfo {volume} {110}},\ \bibinfo {pages} {064310} (\bibinfo {year} {2024})}\BibitemShut {NoStop}%
\bibitem [{\citenamefont {Boffi}\ and\ \citenamefont {Vanden-Eijnden}(2024{\natexlab{a}})}]{Boffi2024}%
  \BibitemOpen
  \bibfield  {author} {\bibinfo {author} {\bibfnamefont {N.~M.}\ \bibnamefont {Boffi}}\ and\ \bibinfo {author} {\bibfnamefont {E.}~\bibnamefont {Vanden-Eijnden}},\ }\href {https://doi.org/10.48550/arXiv.2411.14317} {\bibinfo {title} {Model-free learning of probability flows: {Elucidating} the nonequilibrium dynamics of flocking}} (\bibinfo {year} {2024}{\natexlab{a}})\BibitemShut {NoStop}%
\bibitem [{\citenamefont {Boffi}\ and\ \citenamefont {Vanden-Eijnden}(2024{\natexlab{b}})}]{Boffi2024a}%
  \BibitemOpen
  \bibfield  {author} {\bibinfo {author} {\bibfnamefont {N.~M.}\ \bibnamefont {Boffi}}\ and\ \bibinfo {author} {\bibfnamefont {E.}~\bibnamefont {Vanden-Eijnden}},\ }\bibfield  {title} {\bibinfo {title} {Deep learning probability flows and entropy production rates in active matter},\ }\href {https://doi.org/10.1073/pnas.2318106121} {\bibfield  {journal} {\bibinfo  {journal} {Proceedings of the National Academy of Sciences}\ }\textbf {\bibinfo {volume} {121}},\ \bibinfo {pages} {e2318106121} (\bibinfo {year} {2024}{\natexlab{b}})}\BibitemShut {NoStop}%
\bibitem [{\citenamefont {Ma}\ and\ \citenamefont {Collins}(2018)}]{Ma2018}%
  \BibitemOpen
  \bibfield  {author} {\bibinfo {author} {\bibfnamefont {Z.}~\bibnamefont {Ma}}\ and\ \bibinfo {author} {\bibfnamefont {M.}~\bibnamefont {Collins}},\ }\href {https://doi.org/10.48550/arXiv.1809.01812} {\bibinfo {title} {Noise {Contrastive} {Estimation} and {Negative} {Sampling} for {Conditional} {Models}: {Consistency} and {Statistical} {Efficiency}}} (\bibinfo {year} {2018})\BibitemShut {NoStop}%
\bibitem [{\citenamefont {Kawai}\ \emph {et~al.}(2007)\citenamefont {Kawai}, \citenamefont {Parrondo},\ and\ \citenamefont {den Broeck}}]{Kawai2007}%
  \BibitemOpen
  \bibfield  {author} {\bibinfo {author} {\bibfnamefont {R.}~\bibnamefont {Kawai}}, \bibinfo {author} {\bibfnamefont {J.~M.~R.}\ \bibnamefont {Parrondo}},\ and\ \bibinfo {author} {\bibfnamefont {C.~V.}\ \bibnamefont {den Broeck}},\ }\bibfield  {title} {\bibinfo {title} {Dissipation: {The} {Phase}-{Space} {Perspective}},\ }\href {https://doi.org/10.1103/PhysRevLett.98.080602} {\bibfield  {journal} {\bibinfo  {journal} {Physical Review Letters}\ }\textbf {\bibinfo {volume} {98}},\ \bibinfo {pages} {080602} (\bibinfo {year} {2007})}\BibitemShut {NoStop}%
\bibitem [{\citenamefont {Rieder}\ \emph {et~al.}(1967)\citenamefont {Rieder}, \citenamefont {Lebowitz},\ and\ \citenamefont {Lieb}}]{Rieder1967}%
  \BibitemOpen
  \bibfield  {author} {\bibinfo {author} {\bibfnamefont {Z.}~\bibnamefont {Rieder}}, \bibinfo {author} {\bibfnamefont {J.~L.}\ \bibnamefont {Lebowitz}},\ and\ \bibinfo {author} {\bibfnamefont {E.}~\bibnamefont {Lieb}},\ }\bibfield  {title} {\bibinfo {title} {Properties of a harmonic crystal in a stationary nonequilibrium state},\ }\href {https://doi.org/10.1063/1.1705319} {\bibfield  {journal} {\bibinfo  {journal} {Journal of Mathematical Physics}\ }\textbf {\bibinfo {volume} {8}},\ \bibinfo {pages} {1073} (\bibinfo {year} {1967})},\ \Eprint {https://arxiv.org/abs/https://pubs.aip.org/aip/jmp/article-pdf/8/5/1073/19085581/1073\_1\_online.pdf} {https://pubs.aip.org/aip/jmp/article-pdf/8/5/1073/19085581/1073\_1\_online.pdf} \BibitemShut {NoStop}%
\bibitem [{\citenamefont {Bonetto}\ \emph {et~al.}(2004)\citenamefont {Bonetto}, \citenamefont {Lebowitz},\ and\ \citenamefont {Lukkarinen}}]{Bonetto2004}%
  \BibitemOpen
  \bibfield  {author} {\bibinfo {author} {\bibfnamefont {F.}~\bibnamefont {Bonetto}}, \bibinfo {author} {\bibfnamefont {J.~L.}\ \bibnamefont {Lebowitz}},\ and\ \bibinfo {author} {\bibfnamefont {J.}~\bibnamefont {Lukkarinen}},\ }\bibfield  {title} {\bibinfo {title} {Fourier's {Law} for a {Harmonic} {Crystal} with {Self}-{Consistent} {Stochastic} {Reservoirs}},\ }\href {https://doi.org/10.1023/B:JOSS.0000037232.14365.10} {\bibfield  {journal} {\bibinfo  {journal} {Journal of Statistical Physics}\ }\textbf {\bibinfo {volume} {116}},\ \bibinfo {pages} {783} (\bibinfo {year} {2004})}\BibitemShut {NoStop}%
\bibitem [{\citenamefont {Yildiz}\ \emph {et~al.}(2003)\citenamefont {Yildiz}, \citenamefont {Forkey}, \citenamefont {McKinney}, \citenamefont {Ha}, \citenamefont {Goldman},\ and\ \citenamefont {Selvin}}]{Yildiz2003}%
  \BibitemOpen
  \bibfield  {author} {\bibinfo {author} {\bibfnamefont {A.}~\bibnamefont {Yildiz}}, \bibinfo {author} {\bibfnamefont {J.~N.}\ \bibnamefont {Forkey}}, \bibinfo {author} {\bibfnamefont {S.~A.}\ \bibnamefont {McKinney}}, \bibinfo {author} {\bibfnamefont {T.}~\bibnamefont {Ha}}, \bibinfo {author} {\bibfnamefont {Y.~E.}\ \bibnamefont {Goldman}},\ and\ \bibinfo {author} {\bibfnamefont {P.~R.}\ \bibnamefont {Selvin}},\ }\bibfield  {title} {\bibinfo {title} {Myosin v walks hand-over-hand: Single fluorophore imaging with 1.5-nm localization},\ }\href {https://doi.org/10.1126/science.1084398} {\bibfield  {journal} {\bibinfo  {journal} {Science}\ }\textbf {\bibinfo {volume} {300}},\ \bibinfo {pages} {2061} (\bibinfo {year} {2003})},\ \Eprint {https://arxiv.org/abs/https://www.science.org/doi/pdf/10.1126/science.1084398} {https://www.science.org/doi/pdf/10.1126/science.1084398} \BibitemShut {NoStop}%
\bibitem [{\citenamefont {Yildiz}\ \emph {et~al.}(2004)\citenamefont {Yildiz}, \citenamefont {Tomishige}, \citenamefont {Vale},\ and\ \citenamefont {Selvin}}]{Yildiz2004}%
  \BibitemOpen
  \bibfield  {author} {\bibinfo {author} {\bibfnamefont {A.}~\bibnamefont {Yildiz}}, \bibinfo {author} {\bibfnamefont {M.}~\bibnamefont {Tomishige}}, \bibinfo {author} {\bibfnamefont {R.~D.}\ \bibnamefont {Vale}},\ and\ \bibinfo {author} {\bibfnamefont {P.~R.}\ \bibnamefont {Selvin}},\ }\bibfield  {title} {\bibinfo {title} {Kinesin walks hand-over-hand},\ }\href {https://doi.org/10.1126/science.1093753} {\bibfield  {journal} {\bibinfo  {journal} {Science}\ }\textbf {\bibinfo {volume} {303}},\ \bibinfo {pages} {676} (\bibinfo {year} {2004})},\ \Eprint {https://arxiv.org/abs/https://www.science.org/doi/pdf/10.1126/science.1093753} {https://www.science.org/doi/pdf/10.1126/science.1093753} \BibitemShut {NoStop}%
\bibitem [{\citenamefont {Risken}(1996)}]{Risken1996}%
  \BibitemOpen
  \bibfield  {author} {\bibinfo {author} {\bibfnamefont {H.}~\bibnamefont {Risken}},\ }\href {https://doi.org/10.1007/978-3-642-61544-3} {\emph {\bibinfo {title} {The Fokker-Planck Equation: Methods of Solution and Applications}}},\ edited by\ \bibinfo {editor} {\bibfnamefont {H.}~\bibnamefont {Haken}},\ \bibinfo {series} {Springer {Series} in {Synergetics}}, Vol.~\bibinfo {volume} {18}\ (\bibinfo  {publisher} {Springer},\ \bibinfo {address} {Berlin, Heidelberg},\ \bibinfo {year} {1996})\BibitemShut {NoStop}%
\bibitem [{\citenamefont {Chat\'e}\ and\ \citenamefont {Manneville}(1996)}]{Chate1996}%
  \BibitemOpen
  \bibfield  {author} {\bibinfo {author} {\bibfnamefont {H.}~\bibnamefont {Chat\'e}}\ and\ \bibinfo {author} {\bibfnamefont {P.}~\bibnamefont {Manneville}},\ }\bibfield  {title} {\bibinfo {title} {Phase {Diagram} of the {Two}-{Dimensional} {Complex} {Ginzburg}-{Landau} {Equation}},\ }\href {https://doi.org/10.1016/0378-4371(95)00361-4} {\bibfield  {journal} {\bibinfo  {journal} {Physica A: Statistical Mechanics and its Applications}\ }\textbf {\bibinfo {volume} {224}},\ \bibinfo {pages} {348} (\bibinfo {year} {1996})},\ \bibinfo {note} {arXiv:1608.07519 [nlin]}\BibitemShut {NoStop}%
\bibitem [{\citenamefont {Aranson}\ and\ \citenamefont {Kramer}(2002)}]{Aranson2002}%
  \BibitemOpen
  \bibfield  {author} {\bibinfo {author} {\bibfnamefont {I.~S.}\ \bibnamefont {Aranson}}\ and\ \bibinfo {author} {\bibfnamefont {L.}~\bibnamefont {Kramer}},\ }\bibfield  {title} {\bibinfo {title} {The world of the complex {Ginzburg}-{Landau} equation},\ }\href {https://doi.org/10.1103/RevModPhys.74.99} {\bibfield  {journal} {\bibinfo  {journal} {Reviews of Modern Physics}\ }\textbf {\bibinfo {volume} {74}},\ \bibinfo {pages} {99} (\bibinfo {year} {2002})},\ \bibinfo {note} {publisher: American Physical Society}\BibitemShut {NoStop}%
\bibitem [{\citenamefont {Tan}\ \emph {et~al.}(2020)\citenamefont {Tan}, \citenamefont {Liu}, \citenamefont {Miller}, \citenamefont {Tekant}, \citenamefont {Dunkel},\ and\ \citenamefont {Fakhri}}]{tan2020topological}%
  \BibitemOpen
  \bibfield  {author} {\bibinfo {author} {\bibfnamefont {T.~H.}\ \bibnamefont {Tan}}, \bibinfo {author} {\bibfnamefont {J.}~\bibnamefont {Liu}}, \bibinfo {author} {\bibfnamefont {P.~W.}\ \bibnamefont {Miller}}, \bibinfo {author} {\bibfnamefont {M.}~\bibnamefont {Tekant}}, \bibinfo {author} {\bibfnamefont {J.}~\bibnamefont {Dunkel}},\ and\ \bibinfo {author} {\bibfnamefont {N.}~\bibnamefont {Fakhri}},\ }\bibfield  {title} {\bibinfo {title} {Topological turbulence in the membrane of a living cell},\ }\href@noop {} {\bibfield  {journal} {\bibinfo  {journal} {Nature Physics}\ }\textbf {\bibinfo {volume} {16}},\ \bibinfo {pages} {657} (\bibinfo {year} {2020})}\BibitemShut {NoStop}%
\bibitem [{\citenamefont {Li}\ \emph {et~al.}(2024)\citenamefont {Li}, \citenamefont {Liu}, \citenamefont {Szurek},\ and\ \citenamefont {Fakhri}}]{Li2024}%
  \BibitemOpen
  \bibfield  {author} {\bibinfo {author} {\bibfnamefont {J.}~\bibnamefont {Li}}, \bibinfo {author} {\bibfnamefont {C.-W.~J.}\ \bibnamefont {Liu}}, \bibinfo {author} {\bibfnamefont {M.}~\bibnamefont {Szurek}},\ and\ \bibinfo {author} {\bibfnamefont {N.}~\bibnamefont {Fakhri}},\ }\bibfield  {title} {\bibinfo {title} {Measuring {Irreversibility} from {Learned} {Representations} of {Biological} {Patterns}},\ }\href {https://doi.org/10.1103/PRXLife.2.033013} {\bibfield  {journal} {\bibinfo  {journal} {PRX Life}\ }\textbf {\bibinfo {volume} {2}},\ \bibinfo {pages} {033013} (\bibinfo {year} {2024})}\BibitemShut {NoStop}%
\bibitem [{sup()}]{suppvideo}%
  \BibitemOpen
  \href@noop {} {}\bibinfo {note} {See Supplementary Movies M1 and M2 at [URL will be inserted by publisher]}\BibitemShut {NoStop}%
\bibitem [{\citenamefont {Watter}\ \emph {et~al.}(2015)\citenamefont {Watter}, \citenamefont {Springenberg}, \citenamefont {Boedecker},\ and\ \citenamefont {Riedmiller}}]{Watter2015}%
  \BibitemOpen
  \bibfield  {author} {\bibinfo {author} {\bibfnamefont {M.}~\bibnamefont {Watter}}, \bibinfo {author} {\bibfnamefont {J.}~\bibnamefont {Springenberg}}, \bibinfo {author} {\bibfnamefont {J.}~\bibnamefont {Boedecker}},\ and\ \bibinfo {author} {\bibfnamefont {M.}~\bibnamefont {Riedmiller}},\ }\bibfield  {title} {\bibinfo {title} {Embed to control: A locally linear latent dynamics model for control from raw images},\ }in\ \href {https://proceedings.neurips.cc/paper_files/paper/2015/file/a1afc58c6ca9540d057299ec3016d726-Paper.pdf} {\emph {\bibinfo {booktitle} {Advances in Neural Information Processing Systems}}},\ Vol.~\bibinfo {volume} {28},\ \bibinfo {editor} {edited by\ \bibinfo {editor} {\bibfnamefont {C.}~\bibnamefont {Cortes}}, \bibinfo {editor} {\bibfnamefont {N.}~\bibnamefont {Lawrence}}, \bibinfo {editor} {\bibfnamefont {D.}~\bibnamefont {Lee}}, \bibinfo {editor} {\bibfnamefont {M.}~\bibnamefont {Sugiyama}},\ and\ \bibinfo {editor} {\bibfnamefont {R.}~\bibnamefont {Garnett}}}\ (\bibinfo  {publisher}
  {Curran Associates, Inc.},\ \bibinfo {year} {2015})\BibitemShut {NoStop}%
\bibitem [{cod()}]{code}%
  \BibitemOpen
  \href@noop {} {}\bibinfo {note} {Code Repository for LENS -- \url{https://github.com/cj7280/LENS}}\BibitemShut {NoStop}%
\bibitem [{\citenamefont {He}\ \emph {et~al.}(2016)\citenamefont {He}, \citenamefont {Zhang}, \citenamefont {Ren},\ and\ \citenamefont {Sun}}]{He2016}%
  \BibitemOpen
  \bibfield  {author} {\bibinfo {author} {\bibfnamefont {K.}~\bibnamefont {He}}, \bibinfo {author} {\bibfnamefont {X.}~\bibnamefont {Zhang}}, \bibinfo {author} {\bibfnamefont {S.}~\bibnamefont {Ren}},\ and\ \bibinfo {author} {\bibfnamefont {J.}~\bibnamefont {Sun}},\ }\bibfield  {title} {\bibinfo {title} {Deep residual learning for image recognition},\ }in\ \href@noop {} {\emph {\bibinfo {booktitle} {Proceedings of the IEEE Conference on Computer Vision and Pattern Recognition (CVPR)}}}\ (\bibinfo {year} {2016})\BibitemShut {NoStop}%
\bibitem [{\citenamefont {Ba}\ \emph {et~al.}(2016)\citenamefont {Ba}, \citenamefont {Kiros},\ and\ \citenamefont {Hinton}}]{Ba2016}%
  \BibitemOpen
  \bibfield  {author} {\bibinfo {author} {\bibfnamefont {J.~L.}\ \bibnamefont {Ba}}, \bibinfo {author} {\bibfnamefont {J.~R.}\ \bibnamefont {Kiros}},\ and\ \bibinfo {author} {\bibfnamefont {G.~E.}\ \bibnamefont {Hinton}},\ }\href {https://arxiv.org/abs/1607.06450} {\bibinfo {title} {Layer normalization}} (\bibinfo {year} {2016}),\ \Eprint {https://arxiv.org/abs/1607.06450} {arXiv:1607.06450 [stat.ML]} \BibitemShut {NoStop}%
\bibitem [{\citenamefont {Hendrycks}\ and\ \citenamefont {Gimpel}(2016)}]{hendrycks2016gaussian}%
  \BibitemOpen
  \bibfield  {author} {\bibinfo {author} {\bibfnamefont {D.}~\bibnamefont {Hendrycks}}\ and\ \bibinfo {author} {\bibfnamefont {K.}~\bibnamefont {Gimpel}},\ }\bibfield  {title} {\bibinfo {title} {Gaussian error linear units (gelus)},\ }\href@noop {} {\bibfield  {journal} {\bibinfo  {journal} {arXiv preprint arXiv:1606.08415}\ } (\bibinfo {year} {2016})}\BibitemShut {NoStop}%
\bibitem [{\citenamefont {Liu}\ \emph {et~al.}(2018)\citenamefont {Liu}, \citenamefont {Lehman}, \citenamefont {Molino}, \citenamefont {Such}, \citenamefont {Frank}, \citenamefont {Sergeev},\ and\ \citenamefont {Yosinski}}]{Liu2018}%
  \BibitemOpen
  \bibfield  {author} {\bibinfo {author} {\bibfnamefont {R.}~\bibnamefont {Liu}}, \bibinfo {author} {\bibfnamefont {J.}~\bibnamefont {Lehman}}, \bibinfo {author} {\bibfnamefont {P.}~\bibnamefont {Molino}}, \bibinfo {author} {\bibfnamefont {F.~P.}\ \bibnamefont {Such}}, \bibinfo {author} {\bibfnamefont {E.}~\bibnamefont {Frank}}, \bibinfo {author} {\bibfnamefont {A.}~\bibnamefont {Sergeev}},\ and\ \bibinfo {author} {\bibfnamefont {J.}~\bibnamefont {Yosinski}},\ }\href {https://arxiv.org/abs/1807.03247} {\bibinfo {title} {An intriguing failing of convolutional neural networks and the coordconv solution}} (\bibinfo {year} {2018}),\ \Eprint {https://arxiv.org/abs/1807.03247} {arXiv:1807.03247 [cs.CV]} \BibitemShut {NoStop}%
\bibitem [{\citenamefont {Biewald}(2020)}]{wandb}%
  \BibitemOpen
  \bibfield  {author} {\bibinfo {author} {\bibfnamefont {L.}~\bibnamefont {Biewald}},\ }\href {https://www.wandb.com/} {\bibinfo {title} {Experiment tracking with weights and biases}} (\bibinfo {year} {2020}),\ \bibinfo {note} {software available from wandb.com}\BibitemShut {NoStop}%
\bibitem [{\citenamefont {Bortkiewicz}\ \emph {et~al.}(2025)\citenamefont {Bortkiewicz}, \citenamefont {Pa\l{}ucki}, \citenamefont {Myers}, \citenamefont {Dziarmaga}, \citenamefont {Arczewski}, \citenamefont {Kuci\'{n}ski},\ and\ \citenamefont {Eysenbach}}]{Bortkiewicz2025}%
  \BibitemOpen
  \bibfield  {author} {\bibinfo {author} {\bibfnamefont {M.}~\bibnamefont {Bortkiewicz}}, \bibinfo {author} {\bibfnamefont {W.}~\bibnamefont {Pa\l{}ucki}}, \bibinfo {author} {\bibfnamefont {V.}~\bibnamefont {Myers}}, \bibinfo {author} {\bibfnamefont {T.}~\bibnamefont {Dziarmaga}}, \bibinfo {author} {\bibfnamefont {T.}~\bibnamefont {Arczewski}}, \bibinfo {author} {\bibfnamefont {L.}~\bibnamefont {Kuci\'{n}ski}},\ and\ \bibinfo {author} {\bibfnamefont {B.}~\bibnamefont {Eysenbach}},\ }\bibfield  {title} {\bibinfo {title} {{Accelerating Goal-Conditioned RL Algorithms} and {Research}},\ }in\ \href {https://arxiv.org/pdf/2408.11052} {\emph {\bibinfo {booktitle} {{International Conference} on {Learning Representations}}}}\ (\bibinfo {year} {2025})\BibitemShut {NoStop}%
\bibitem [{\citenamefont {Lyu}\ \emph {et~al.}(2025)\citenamefont {Lyu}, \citenamefont {Ray},\ and\ \citenamefont {Crutchfield}}]{Lyu2025}%
  \BibitemOpen
  \bibfield  {author} {\bibinfo {author} {\bibfnamefont {J.}~\bibnamefont {Lyu}}, \bibinfo {author} {\bibfnamefont {K.~J.}\ \bibnamefont {Ray}},\ and\ \bibinfo {author} {\bibfnamefont {J.~P.}\ \bibnamefont {Crutchfield}},\ }\href {https://doi.org/10.48550/arXiv.2504.19007} {\bibinfo {title} {Learning {Stochastic} {Thermodynamics} {Directly} from {Correlation} and {Trajectory}-{Fluctuation} {Currents}}} (\bibinfo {year} {2025})\BibitemShut {NoStop}%
\bibitem [{\citenamefont {Stamhuis}\ and\ \citenamefont {Thielicke}(2014)}]{stamhuis2014pivlab}%
  \BibitemOpen
  \bibfield  {author} {\bibinfo {author} {\bibfnamefont {E.}~\bibnamefont {Stamhuis}}\ and\ \bibinfo {author} {\bibfnamefont {W.}~\bibnamefont {Thielicke}},\ }\bibfield  {title} {\bibinfo {title} {Pivlab--towards user-friendly, affordable and accurate digital particle image velocimetry in matlab},\ }\href@noop {} {\bibfield  {journal} {\bibinfo  {journal} {Journal of open research software}\ }\textbf {\bibinfo {volume} {2}},\ \bibinfo {pages} {30} (\bibinfo {year} {2014})}\BibitemShut {NoStop}%
\bibitem [{\citenamefont {Matsumoto}\ \emph {et~al.}(2025)\citenamefont {Matsumoto}, \citenamefont {Sasa},\ and\ \citenamefont {Dechant}}]{Matsumoto2025}%
  \BibitemOpen
  \bibfield  {author} {\bibinfo {author} {\bibfnamefont {K.}~\bibnamefont {Matsumoto}}, \bibinfo {author} {\bibfnamefont {S.-i.}\ \bibnamefont {Sasa}},\ and\ \bibinfo {author} {\bibfnamefont {A.}~\bibnamefont {Dechant}},\ }\href {https://doi.org/10.48550/arXiv.2504.09981} {\bibinfo {title} {Learning rate matrix and information-thermodynamic trade-off relation}} (\bibinfo {year} {2025})\BibitemShut {NoStop}%
\bibitem [{\citenamefont {Mezzadri}(2007)}]{mezzadri2007generaterandommatricesclassical}%
  \BibitemOpen
  \bibfield  {author} {\bibinfo {author} {\bibfnamefont {F.}~\bibnamefont {Mezzadri}},\ }\href {https://arxiv.org/abs/math-ph/0609050} {\bibinfo {title} {How to generate random matrices from the classical compact groups}} (\bibinfo {year} {2007}),\ \Eprint {https://arxiv.org/abs/math-ph/0609050} {arXiv:math-ph/0609050 [math-ph]} \BibitemShut {NoStop}%
\bibitem [{\citenamefont {Sitzmann}\ \emph {et~al.}(2020)\citenamefont {Sitzmann}, \citenamefont {Martel}, \citenamefont {Bergman}, \citenamefont {Lindell},\ and\ \citenamefont {Wetzstein}}]{sitzmann2020implicit}%
  \BibitemOpen
  \bibfield  {author} {\bibinfo {author} {\bibfnamefont {V.}~\bibnamefont {Sitzmann}}, \bibinfo {author} {\bibfnamefont {J.}~\bibnamefont {Martel}}, \bibinfo {author} {\bibfnamefont {A.}~\bibnamefont {Bergman}}, \bibinfo {author} {\bibfnamefont {D.}~\bibnamefont {Lindell}},\ and\ \bibinfo {author} {\bibfnamefont {G.}~\bibnamefont {Wetzstein}},\ }\bibfield  {title} {\bibinfo {title} {Implicit neural representations with periodic activation functions},\ }\href@noop {} {\bibfield  {journal} {\bibinfo  {journal} {Advances in neural information processing systems}\ }\textbf {\bibinfo {volume} {33}},\ \bibinfo {pages} {7462} (\bibinfo {year} {2020})}\BibitemShut {NoStop}%
\bibitem [{\citenamefont {Cox}\ and\ \citenamefont {Matthews}(2002)}]{Cox2002}%
  \BibitemOpen
  \bibfield  {author} {\bibinfo {author} {\bibfnamefont {S.}~\bibnamefont {Cox}}\ and\ \bibinfo {author} {\bibfnamefont {P.}~\bibnamefont {Matthews}},\ }\bibfield  {title} {\bibinfo {title} {Exponential time differencing for stiff systems},\ }\href {https://doi.org/https://doi.org/10.1006/jcph.2002.6995} {\bibfield  {journal} {\bibinfo  {journal} {Journal of Computational Physics}\ }\textbf {\bibinfo {volume} {176}},\ \bibinfo {pages} {430} (\bibinfo {year} {2002})}\BibitemShut {NoStop}%
\bibitem [{\citenamefont {Huber}\ \emph {et~al.}(1992)\citenamefont {Huber}, \citenamefont {Alstr\o{}m},\ and\ \citenamefont {Bohr}}]{Huber1992}%
  \BibitemOpen
  \bibfield  {author} {\bibinfo {author} {\bibfnamefont {G.}~\bibnamefont {Huber}}, \bibinfo {author} {\bibfnamefont {P.}~\bibnamefont {Alstr\o{}m}},\ and\ \bibinfo {author} {\bibfnamefont {T.}~\bibnamefont {Bohr}},\ }\bibfield  {title} {\bibinfo {title} {Nucleation and transients at the onset of vortex turbulence},\ }\href {https://doi.org/10.1103/PhysRevLett.69.2380} {\bibfield  {journal} {\bibinfo  {journal} {Phys. Rev. Lett.}\ }\textbf {\bibinfo {volume} {69}},\ \bibinfo {pages} {2380} (\bibinfo {year} {1992})}\BibitemShut {NoStop}%
\bibitem [{\citenamefont {van Saarloos}\ and\ \citenamefont {Hohenberg}(1992)}]{Hohenberg1992}%
  \BibitemOpen
  \bibfield  {author} {\bibinfo {author} {\bibfnamefont {W.}~\bibnamefont {van Saarloos}}\ and\ \bibinfo {author} {\bibfnamefont {P.}~\bibnamefont {Hohenberg}},\ }\bibfield  {title} {\bibinfo {title} {Fronts, pulses, sources and sinks in generalized complex {Ginzburg}-{Landau} equations},\ }\href {https://doi.org/10.1016/0167-2789(92)90175-m} {\bibfield  {journal} {\bibinfo  {journal} {Physica D: Nonlinear Phenomena}\ }\textbf {\bibinfo {volume} {56}},\ \bibinfo {pages} {303–367} (\bibinfo {year} {1992})}\BibitemShut {NoStop}%
\bibitem [{\citenamefont {Nowozin}\ \emph {et~al.}(2016)\citenamefont {Nowozin}, \citenamefont {Cseke},\ and\ \citenamefont {Tomioka}}]{Nowozin2016}%
  \BibitemOpen
  \bibfield  {author} {\bibinfo {author} {\bibfnamefont {S.}~\bibnamefont {Nowozin}}, \bibinfo {author} {\bibfnamefont {B.}~\bibnamefont {Cseke}},\ and\ \bibinfo {author} {\bibfnamefont {R.}~\bibnamefont {Tomioka}},\ }\bibfield  {title} {\bibinfo {title} {f-{GAN}: {Training} {Generative} {Neural} {Samplers} using {Variational} {Divergence} {Minimization}},\ }in\ \href {https://papers.nips.cc/paper_files/paper/2016/hash/cedebb6e872f539bef8c3f919874e9d7-Abstract.html} {\emph {\bibinfo {booktitle} {Advances in {Neural} {Information} {Processing} {Systems}}}},\ Vol.~\bibinfo {volume} {29}\ (\bibinfo  {publisher} {Curran Associates, Inc.},\ \bibinfo {year} {2016})\BibitemShut {NoStop}%
\end{thebibliography}
